\documentclass[12pt]{article}

\let\vide\emptyset 
\let\appartient\in
\let\firstunion\cup
\let\firstinter\cap
\let\savesqcup\sqcup
\let\savesqcap\sqcap

\usepackage[T1]{fontenc}
\usepackage{amsmath}
\usepackage{amsfonts}
\usepackage{amssymb}
\usepackage{proof}
\usepackage{extarrows}
\usepackage{multicol}
\usepackage{enumerate}
\usepackage{stmaryrd}
\usepackage{multicol}
\usepackage[rflt]{floatflt}
\usepackage{eepic}
\usepackage{epic}
\usepackage{tikz}
\usepackage{ifthen} 

\let\emptyset\vide 
\let\in\appartient
\let\cup\firstunion
\let\sqcup\savesqcup
\let\cap\firstinter
\let\sqcap\savesqcap
\renewcommand{\subset}{\subseteq}
\newcommand{\egdef}{\stackrel{\text{\tiny def}}{=}}

\let\n\newcommand

\newlength{\flgth}
\newlength{\flgthh}
\n{\TOG}{The Open Group}

\renewcommand{\wp}{\mathcal{P}}

\DeclareMathAlphabet{\mathpzc}{OT1}{pzc}{m}{it}
\DeclareMathAlphabet{\concr}{OT1}{cmtt}{m}{n}
\DeclareMathAlphabet{\comm}{OT1}{cmss}{m}{sl}

\newcommand{\textc}[1]%
{{\fontfamily{cmtt}%
\fontseries{m}%
\fontshape{n}%
\selectfont%
#1%
}}
\n{\mathc}[1]{\text{\textnormal{\textc{#1}}}}

\n{\citerinard}{\cite{rugina99pointer,rugina-pointer}}

\newcommand{\lab}{  {}^{\ell}}
\newcommand{\li}[1]{ {}^{\ell_{#1}}  }
\n{\fin}{\ell_{\star}} 
\n{\lend}{\ell_{\infty}}
\newcommand{\lip}[1]{ {}^{\ell{#1}}  }
\newcounter{labels}[figure]
\n{\nl}[1][{}]{\nls \ifthenelse{\equal{#1}{}}{\relax}{\nllab{#1}}}
\n{\nls}{\refstepcounter{labels}\lip{_{\arabic{labels}}}}
\n{\nlf}{}
\n{\nlnew}{\setcounter{labels}{0}}
\n{\nllab}[1]{\label{command-label:#1-6151752}}
\n{\nlrefu}[1]{\li{\ref{command-label:#1-6151752}}}
\n{\nlref}[1]{\nlrefb{#1}}
\n{\nlrefb}[1]{\ell_{\ref{command-label:#1-6151752}}}
\n{\nr}[1][{}]{\nlrefu{#1}}

\n{\cwhile}{\comm{while}}
\n{\cskip}{\comm{skip}}
\n{\ccreate}{\comm{create}}
\n{\crcreate}{\comm{spawn}}
\n{\cspawn}{\comm{spawn}}
\n{\cguard}{\comm{guard}}
\n{\cpar}{\comm{par}}
\n{\cifte}[2]{\comm{if}(#1)\comm{then}#2\comm{else}}
\n{\cif}{\comm{if}}

\n{\cond}{\mathit{cond}}
\n{\guardp}{\cond}
\n{\guardn}{\neg \cond}

\n{\cmd}{\mathit{cmd}}
\n{\stmt}{\mathit{stmt}}

\n{\excmd}{\mathit{example}}

\n{\restrict}[2]{#1_{|#2}}
\n{\compl}[1]{\overline{#1}}

\n{\Lab}{\mathbf{Labels}} 
\n{\Labs}[1]{\mathit{Labs}(#1)}
\n{\Labsf}[1]{\mathit{Labs}_{\mathit{father}}(#1)}
\n{\Labss}[1]{\mathit{Labs}_{\mathit{child}}(#1)}
\newcommand{\va}{
\mathpzc{Var}} 
\newcommand{\m}{\sigma} 
\n{\D}{\wp(\St)}
\n{\DI}{\wp(\Tra)}
\n{\Da}{\mathscr{D}}
\n{\DIa}{\mathscr{R}}
\n{\ids}{\mathbf{Ids}}
\n{\main}{\text{\textit{\textbf{main}}}}

\n{\return}{return}
\n{\returns}{returns}
\n{\returning}{returning}
\n{\returned}{returned}

\n{\sinit}{s_0}
\n{\Pinit}{P_0}
\n{\minit}{\m_0}
\n{\hinit}{\h_0}
\n{\send}{s_1}
\n{\iend}{i_1}
\n{\Pend}{P_1}
\n{\mend}{\m_1}
\n{\hendf}{\h_1}
\n{\csend}{\cstate{_1}}
\n{\sinter}{s}
\n{\Pinter}{P}
\n{\minter}{\m}
\n{\hinterf}{\h}
\n{\sinterb}{s'}
\n{\csinterb}{\cstate{'}}

\n{\stores}{\mathbf{Stores}}
\n{\St}{\mathbf{States}}
\n{\Tra}{\mathbf{Tr}}
\n{\Hist}{\mathbf{Genealogies}}
\n{\h}{g} 
\n{\ini}{\mathit{Init}} 

\n{\TR}[1]{\mathpzc{Tr}_{#1}}
\n{\TRi}[1]{\TR{#1}}

\n{\Tr}{\mathpzc{Tr}}

\n{\cstate}[1]{(i#1,P#1,\m#1,\h#1)}

\n{\regle}[3][{}]{ \infer[\hspace{-0.1cm}\text{\footnotesize #1}]{#3}{#2} }
\n{\sepprem}{ \qquad } 
 \n{\rightlocarrow}{\rightarrow}

\n{\cname}{\(\concr{G}\)-collecting }
\n{\Cname}{\(\concr{G}\)-collecting }
\n{\Cnamep}{\textit{\textc{G}}-collecting }

\n{\abstr}[1]{\mathpzc{#1}}

\n{\asem}[1]{\llparenthesis #1 \rrparenthesis}

\newlength{\lengthacolade}
\n{\lacol}{\big\{\hspace{\lengthacolade}\big| }
\n{\racol}{\big|\hspace{\lengthacolade}\big\} }
\n{\lsemantic}{\big[\hspace{\lengthacolade}\big| }
\n{\rsemantic}{\big|\hspace{\lengthacolade}\big] }
\n{\osem}[1]{\lsemantic #1 \rsemantic} 
\n{\ssem}[1]{\lacol #1 \racol}

\n{\eomega}{^{\uparrow\omega}}
\n{\eint}[1]{^{\uparrow #1}}
\n{\widen}{\triangledown}
\n{\ewiden}{^{\uparrow\widen}}

\newcommand{\lbl}{\mathit{label}}
\newcommand{\thread}{\mathit{thread}}

\newcommand{\ecr}{\mathit{write}}
\newcommand{\evb}{\mathit{bool}}
\newcommand{\vrai}{\mathit{true}}
\newcommand{\faux}{\mathit{false}}
\n{\desce}{\mathit{desc}}
\n{\after}{\mathit{after}}
\n{\dom}{\mathit{Dom}}
\n{\Sche}{\mathit{Schedule}}
\n{\av}[1]{\osem{#1}}

\n{\ur}{\mathit{reach}}

\n{\jc}[4]{#1\vdash_{#2}#3[#4]}
\n{\jcr}[5]{#1\vdash_{#2}#3[#4]}
\n{\jcrs}[3]{#1\vdash_{#2}#3}

\n{\aecr}[1]{\abstr{write}_{#1}}
\newcommand{\aecra}[2]{\aecr{#1}(#2)}
\newcommand{\iaecr}[1]{\abstr{write\text{-}inter}_{#1}}
\newcommand{\iaecra}[2]{\iaecr{#1}(#2)}

\newcommand{\force}{\abstr{enforce}}

\n{\afinit}[1]{\abstr{child}\text{-}\abstr{spawn}_{#1}%
}
\n{\afp}[1]{\abstr{guarantee}_{#1}}
\n{\afw}[1]{\abstr{iter}_{#1}}
\n{\guard}[1]{\abstr{guard}_{#1} }
\n{\assign}[1]{\abstr{assign}_{#1}}
\n{\glue}[1]{\abstr{combine}_{#1}}
\n{\loopwa}{\abstr{loop}}
\n{\acinter}[1]{\abstr{inter}_{#1}}
\n{\aspawn}[1]{\abstr{spawn}_{#1} }
\n{\aexe}{\abstr{execute\text{-}thread}}

\n{\gen}{\mathit{gen}}
\n{\ki}{\mathit{kill}}
\n{\addr}[2]{\mathit{addr}_{#1}(#2)}
\n{\valeur}[2]{\mathit{val}_{#1}(#2)}

\n{\fpd}[1]{\mathc{guarantee}_{#1}}
\n{\fp}[1]{\fpd{#1}}
\n{\cinter}[1]{\mathc{interfere}_{#1}}
\n{\qinter}{\mathc{interfere}}
\n{\cpost}{\mathc{post}}
\n{\cextract}[1]{\cpost({#1})}
\n{\cschedule}{\mathc{schedule-child}}
\n{\ccomb}{\mathc{combine}}
\n{\cinit}{\mathc{init-child}}
\n{\loopw}{\mathc{loop}}
\n{\cexe}{\mathc{execute\text{-}thread}}

\n{\Dval}{_{\text{\tiny{\(\mathscr{D}\)}}}}
\n{\Vval}{_{\text{\tiny{\(V\)}}}}
\n{\DIval}{_{\text{\tiny{\({\mathscr{R}}\)}}}}
\n{\galstore}{_{store}}
\n{\galstate}{_{st}}
\n{\gams}{\gamma\Dval}
\n{\alphs}{\alpha\Dval}
\n{\galtrans}{\DIval}
\n{\gamt}{\gamma\galtrans}
\n{\alpht}{\alpha\galtrans}
\n{\gamconf}{\gamma_{\text{cfg}}}
\n{\alphconf}{\alpha_{\text{cfg}}}
\n{\alphr}{\alpha\DIval}
\n{\gamr}{\gamma\DIval}
\n{\alphl}{\alpha_{\text{L}}}
\n{\gaml}{\gamma_{\text{L}}}
\n{\alphk}{\alpha_{\text{K}}}
\n{\gamk}{\gamma_{\text{K}}}

\n{\C}{\mathscr{C}} 
\n{\A}{\mathscr{A}} 

\n{\confinit}{\langle \ini,\Sche,\Sche \rangle}

\n{\qc}{\langle\concr{S},\concr{G},\concr{A}\rangle}
\n{\qcp}[1]{\langle\concr{S}#1,\concr{G}#1,\concr{A}#1\rangle}

\n{\qu}{\langle \abstr{C},\abstr{L},\abstr{K},\abstr{I}\rangle}
\n{\qup}[1]{\langle \abstr{C#1},\abstr{L#1},\abstr{K#1},\abstr{I#1}\rangle}

\n{\col}{\concr{Reach}}
\n{\colp}{{\concr{Reach}^{+}}}
\n{\extname}{\concr{Ext}}
\n{\ext}[3][]{\extname#1(#2,#3)}
\n{\Se}{\concr{Self}}
\n{\Su}{{\concr{Par}}}
\n{\Sud}{{\concr{Sub}}}
\n{\nq}{[\col,\extname,  \Se,\Su,\Sud]}
\n{\nqp}[1]{[\col #1,\extname#1, \Se #1,\Su #1,\Sud #1]}

\n{\am}{\m^{\sharp}}

\n{\progint}{Parint}
\n{\progP}{MT-Penjili}

\n{\loc}{L.o.C.}
\n{\messag}{Message}
\n{\embarque}{Embedded}
\n{\douze}{Test 12}
\n{\quinze}{Test 15}

\usepackage{paralist}
\usepackage[listofformat=parens]{subfig}
\usepackage{amsthm}

\usepackage{mathrsfs}
\usepackage{url}
\usepackage{graphics}
\usepackage{ifpdf}

\n{\mytopfigure}[3][]{%
\begin{figure}[t]%
\fbox{\hspace{-0.048\textwidth}\parbox{1.02\textwidth}{%
\begin{itemize}%
\vspace{\flgth}%
#3
\vspace{\flgth}%
\end{itemize}}}%
\caption{#2}%
\ifthenelse{\equal{#1}{}}{\relax}{\label{#1}}%
\end{figure}%
}

\n{\mytopfigureb}[3][]{%
\begin{figure}[t]%
\fbox{\parbox{\textwidth}{%
#3
}}
\caption{#2}%
\ifthenelse{\equal{#1}{}}{\relax}{\label{#1}}%
\end{figure}%
}

\n{\myfigure}{\mytopfigure}

\n{\mynumfigure}[3][]{%
\begin{figure}[t]%
\fbox{\hspace{-0.018\textwidth}\parbox{\textwidth}{%
\begin{enumerate}%
\vspace{\flgth}%
#3
\vspace{\flgth}%
\end{enumerate}}}%
\caption{#2}%
\ifthenelse{\equal{#1}{}}{\relax}{\label{#1}}%
\end{figure}%
}

\n{\nummain}{5}
\n{\epais}{\linethickness{0.56mm}}

\let\macomtempamoi\dottedline
\renewcommand{\dottedline}{\linethickness{0.1mm}\macomtempamoi}

\n{\grosseligne}{\epais\macomtempamoi{0.5}}

\n{\mondessinstates}{
\begin{floatingfigure}{3.6cm} \hspace{-4mm}\hspace{-0.5cm}
\fbox{\begin{picture}(36,30)
\put(3,30){
\put(10,0){
\put(0.8,-2){\small\(j_0\)}
\drawline(0,0)(0,-9)
\put(-0.75,-9){\(\bullet\)}
\put(1,-9){{\small\(\sinit\)}}
{
\grosseligne(0,-8)(0,-30)
{
\thinlines\drawline(0,-12)(10,-12) \put(10.8,-12){\small\(j_2\)}
\drawline(0,-5)(-5,-5) \put(-7.5,-5){\small\(j_1\)}
\drawline(0,-22)(5,-22) \put(5.8,-22){\small\(j_5\)}
}
\grosseligne(5,-22)(5,-30) 
}
}
\put(5,-5){
\drawline(0,0)(0,-25)
{\thinlines\drawline(0,-8)(-5,-8)\put(-7.5,-8){\small\(j_3\)}}
\drawline(-5,-8)(-5,-25)
}
\put(20,-12){
\grosseligne(0,0)(0,-18)
{\thinlines
\drawline(0,-4)(10,-4) \put(10.8,-4){\small\(j_4\)}
\put(-0.75,-8){\(\bullet\)}
\put(1,-8){{\small\(\sinter\)}}
\drawline(0,-13)(5,-13) \put(5.8,-13){\small\(j_6\)}
}
\grosseligne(10,-4)(10,-18)
\grosseligne(5,-13)(5,-18)
}
}

\end{picture}}
\caption{States}\label{fig:states}
\end{floatingfigure}
}

\n{\mondessinlocdef}{
\begin{figure}\fbox{

\subfloat[\(\col\)]{\label{figsub:singlethread}
\begin{picture}(14,30)
\put(3,30){
\put(0,0){
\put(0.8,-2){\small\(j_0\)}
\drawline(0,0)(0,-9)
\put(-0.75,-9){\(\bullet\)}
\put(1,-9){{\small\(\sinit\)}}
{
\grosseligne(0,-8)(0,-19.1)}
\put(-0.75,-20.1){\(\bullet\)}
\put(1,-20.1){{\small\(\send\)}}
\drawline(0,-18)(0,-30)
{
\thinlines
\drawline(0,-22)(5,-22) \put(5.8,-22){\small\(j_5\)}
}
\drawline(5,-22)(5,-30) 
}
}
\end{picture}}

\subfloat[\(\col\)]{\label{figsub:interference}
\begin{picture}(22,30)
\put(3,30){
\put(10,0){
\put(0.8,-2){\small\(j_0\)}
\drawline(0,0)(0,-9)
\put(-0.75,-9){\(\bullet\)}
\put(1,-9){{\small\(\sinit\)}}
{
\grosseligne(0,-8)(0,-19.1)}
\put(-0.75,-20.1){\(\bullet\)}
\put(1,-20.1){{\small\(\send\)}}
\drawline(0,-18)(0,-30)
{
\thinlines
\drawline(0,-5)(-5,-5) \put(-7.5,-5){\small\(j_1\)}
\drawline(0,-22)(5,-22) \put(5.8,-22){\small\(j_5\)}
}
\drawline(5,-22)(5,-30) 
}
\put(5,-5){
\drawline(0,0)(0,-25)
{\thinlines\drawline(0,-8)(-5,-8)\put(-7.5,-8){\small\(j_3\)}}
\drawline(-5,-8)(-5,-25)
}
 }

\end{picture}}

\subfloat[\(\Su\)]{\label{figsub:par}
\begin{picture}(36,30)
\put(3,30){
\put(10,0){
\put(0.8,-2){\small\(j_0\)}
\drawline(0,0)(0,-9)
\put(-0.75,-9){\(\bullet\)}
\put(1,-9){{\small\(\sinit\)}}
\dottedline{0.5}(-13,-8)(23,-8)
{ \drawline[0](0,-8)(0,-19.1)}
\put(-0.75,-20.1){\(\bullet\)}
\put(1,-20.1){{\small\(\send\)}}
\dottedline{0.5}(-13,-19.1)(23,-19.1)
\drawline(0,-18)(0,-30)
{
\thinlines
\drawline(0,-12)(10,-12) \put(10.8,-12){\small\(j_2\)}
\drawline(0,-5)(-5,-5) \put(-7.5,-5){\small\(j_1\)}
\drawline(0,-22)(5,-22) \put(5.8,-22){\small\(j_5\)}
}
\drawline(5,-22)(5,-30) 
}
\put(5,-5){
\drawline(0,0)(0,-25)
{\thinlines\drawline(0,-8)(-5,-8)\put(-7.5,-8){\small\(j_3\)}}
\drawline(-5,-8)(-5,-25)
}
\put(20,-12){
\grosseligne(0,0)(0,-7)
{\thinlines
\drawline(0,-4)(10,-4) \put(10.8,-4){\small\(j_4\)}
}
\grosseligne(10,-4)(10,-7)
}
 }

\end{picture}}

\subfloat[\(\Sud\)]{\label{figsub:sub}
\begin{picture}(36,30)
\put(3,30){
\put(10,0){
\put(0.8,-2){\small\(j_0\)}
\drawline(0,0)(0,-9)
\put(-0.75,-9){\(\bullet\)}
\put(1,-9){{\small\(\sinit\)}}
{\drawline[0](0,-8)(0,-19.1)}
\put(-0.75,-20.1){\(\bullet\)}
\put(1,-20.1){{\small\(\send\)}}
\dottedline{0.5}(-13,-19.1)(23,-19.1)
\drawline(0,-18)(0,-30)
{
\thinlines\drawline(0,-12)(10,-12) \put(10.8,-12){\small\(j_2\)}
\drawline(0,-5)(-5,-5) \put(-7.5,-5){\small\(j_1\)}
\drawline(0,-22)(5,-22) \put(5.8,-22){\small\(j_5\)}
}
\drawline(5,-22)(5,-30) 
}
\put(5,-5){
\drawline(0,0)(0,-25)
{\thinlines\drawline(0,-8)(-5,-8)\put(-7.5,-8){\small\(j_3\)}}
\drawline(-5,-8)(-5,-25)
}
\put(20,-12){
\drawline(0,0)(0,-7.5)
{\thinlines
\drawline(0,-4)(10,-4) \put(10.8,-4){\small\(j_4\)}
\drawline(0,-13)(5,-13) \put(5.8,-13){\small\(j_6\)}
}
\drawline(10,-4)(10,-7.5)
{
\grosseligne(0,-7)(0,-18)
\grosseligne(10,-7)(10,-18)
\grosseligne(5,-13)(5,-18)}
}
}
\end{picture}}

}
\caption{\Cname semantics}\label{fig:illusconcr}
\end{figure}
}

\n{\figthschun}{
\begin{figure}\fbox{

\subfloat[\(s_{k_0}\)]{
\begin{picture}(36,30)
\put(3,30){
\put(10,0){
\put(0.8,-2){\small\(j_0\)}
\drawline(0,0)(0,-9)
\put(-0.75,-9){\(\bullet\)}
\put(1,-9){{\small\(s_{0}\)}}
{\drawline[0](0,-8)(0,-19.1)}
\put(-0.75,-19.1){\(\bullet\)}
\put(1,-19.1){{\small\(s_{k_0}\)}
}
\put(-0.75,-27.1){\(\bullet\)}
\put(1,-27.1){{\small\(s_n\)}
}
\drawline(0,-18)(0,-30)
{
\thinlines\drawline(0,-12)(10,-12) \put(10.8,-12){\small\(j_2\)}
\drawline(0,-5)(-5,-5) \put(-7.5,-5){\small\(j_1\)}
\drawline(0,-22)(5,-22) \put(5.8,-22){\small\(j_5\)}
}
\drawline(5,-22)(5,-30) 
}
\put(5,-5){
\drawline(0,0)(0,-25)
{\thinlines\drawline(0,-8)(-5,-8)\put(-7.5,-8){\small\(j_3\)}}
\drawline(-5,-8)(-5,-25)
}
\put(20,-12){
\drawline(0,0)(0,-7.5)
{\thinlines
\drawline(0,-4)(10,-4) \put(10.8,-4){\small\(j_4\)}
\drawline(0,-13)(5,-13) \put(5.8,-16){\small\(j_6\)}
}
\drawline(10,-4)(10,-7.5)
{
\drawline(0,-7)(0,-18)
\drawline(10,-7)(10,-18)
\drawline(5,-13)(5,-18)}
}
}
\end{picture}
}

\subfloat[\(\Su_k\) and \(\Sud_k\)]{\label{subfig:seqparsub}
\begin{picture}(36,30)
\put(3,30){
\put(10,0){
\put(0.8,-2){\small\(j_0\)}
\drawline(0,0)(0,-9)
\put(-0.75,-9){\(\bullet\)}
\dottedline{0.5}(8,-8)(23,-8)
{\drawline[0](0,-8)(0,-19.1)}
\put(-0.75,-19.1){\(\bullet\)}
\dottedline{0.5}(3,-18.1)(23,-18.1)
\put(-0.75,-27.1){\(\bullet\)}
\dottedline{0.5}(3,-26.1)(7,-26.1)
\dottedline{0.5}(8,-29)(8,-8)
\dottedline{0.5}(23,-29)(23,-8)
\drawline(0,-18)(0,-30)
{
\thinlines\drawline(0,-12)(10,-12) 
\put(12,-15){\small\(\Su_1\)}
\drawline(0,-5)(-5,-5) 
\drawline(0,-22)(5,-22) 
\put(2,-22){\small\(\Su_2\)}
\put(2,-29){\small\(\Sud_2\)}
}
\drawline(5,-22)(5,-30) 
}
\put(5,-5){
\drawline(0,0)(0,-25)
{\thinlines\drawline(0,-8)(-5,-8)
}
\drawline(-5,-8)(-5,-25)
}
\put(20,-12){
\drawline(0,0)(0,-7.5)
{\thinlines
\drawline(0,-4)(10,-4) 
\put(3,-10){\small\(\Sud_1\)}
\drawline(0,-13)(5,-13) 
}
\drawline(10,-4)(10,-7.5)
{
\drawline(0,-7)(0,-18)
\drawline(10,-7)(10,-18)
\drawline(5,-13)(5,-18)}
}
}
\end{picture}
}

\subfloat[\(\Su'\) and \(\Sud'\)]{\label{subfig:seqparsubp}
\begin{picture}(36,30)
\put(3,30){
\put(10,0){
\put(0.8,-2){\small\(j_0\)}
\drawline(0,0)(0,-9)
\put(-0.75,-9){\(\bullet\)}
\dottedline{0.5}(8,-8)(23,-8)
{\drawline[0](0,-8)(0,-19.1)}
\put(-0.75,-19.1){\(\bullet\)}
\dottedline{0.5}(3,-18.1)(8,-18.1)
\put(-0.75,-27.1){\(\bullet\)}
\dottedline{0.5}(3,-26.1)(23,-26.1)
\dottedline{0.5}(8,-18.1)(8,-8)
\dottedline{0.5}(23,-29)(23,-8)
\drawline(0,-18)(0,-30)
{
\thinlines\drawline(0,-12)(10,-12) 
\drawline(0,-5)(-5,-5) 
\drawline(0,-22)(5,-22) 
\put(7,-29){\small\(\Sud'\)}
}
\drawline(5,-22)(5,-30) 
}
\put(5,-5){
\drawline(0,0)(0,-25)
{\thinlines\drawline(0,-8)(-5,-8)
}
\drawline(-5,-8)(-5,-25)
}
\put(20,-12){
\drawline(0,0)(0,-7.5)
{\thinlines
\drawline(0,-4)(10,-4) 
\put(3,-8){\small\(\Su'\)}
\drawline(0,-13)(5,-13) 
}
\drawline(10,-4)(10,-7.5)
{
\drawline(0,-7)(0,-18)
\drawline(10,-7)(10,-18)
\drawline(5,-13)(5,-18)}
}
}
\end{picture}
}

}\caption{Sequence}\label{fig:seq}
\end{figure}
}

\n{\classicfig}{
\begin{figure}\fbox{
\begin{picture}(36,30)
\put(3,30){
\put(10,0){
\put(0.8,-2){\small\(j_0\)}
\drawline(0,0)(0,-9)
\put(-0.75,-9){\(\bullet\)}
\put(1,-9){{\small\(\sinit\)}}
{\drawline[0](0,-8)(0,-19.1)}
\put(-0.75,-20.1){\(\bullet\)}
\put(1,-20.1){{\small\(\send\)}}
\drawline(0,-18)(0,-30)
{
\thinlines\drawline(0,-12)(10,-12) \put(10.8,-12){\small\(j_2\)}
\drawline(0,-5)(-5,-5) \put(-7.5,-5){\small\(j_1\)}
\drawline(0,-22)(5,-22) \put(5.8,-22){\small\(j_5\)}
}
\drawline(5,-22)(5,-30) 
}
\put(5,-5){
\drawline(0,0)(0,-25)
{\thinlines\drawline(0,-8)(-5,-8)\put(-7.5,-8){\small\(j_3\)}}
\drawline(-5,-8)(-5,-25)
}
\put(20,-12){
\drawline(0,0)(0,-7.5)
{\thinlines
\drawline(0,-4)(10,-4) \put(10.8,-4){\small\(j_4\)}
\drawline(0,-13)(5,-13) \put(5.8,-13){\small\(j_6\)}
}
\drawline(10,-4)(10,-7.5)
{
\drawline(0,-7)(0,-18)
\drawline(10,-7)(10,-18)
\drawline(5,-13)(5,-18)}
}
}
\end{picture}
}\caption{Sequence}
\end{figure}
}

\n{\figthw}{
\begin{figure}
\begin{tikzpicture}
\draw (1.8,0) ellipse (6pt and 6pt) node {\(\ell_1\)} ;
\draw (4.5,0) ellipse (6pt and 6pt) node {\(\ell_2\)};
\draw[->] (2,0) -- (4.3,0)  node[midway,above] {~~~\(\cguard(\cond)\)};
  \draw[->] (4.7,0.04) -- (6.2,0.4);
\draw[<-] (4.7,-0.04) -- (6.2,-0.4);
\draw (7,0) ellipse (40pt and 20pt) node {\(\Labs{\cmd}\)};
 \draw (11,0) ellipse (6pt and 6pt) node {\(\ell_3\)};
  \draw[->] (8,0) .. controls (12,-0.2) and (7,2.1)  .. (1.9,0.2) node[near end,above] {End of \(\cmd\)};
  \draw[->] (1.9,-0.2) .. controls (6,-1.2) and (8,-1.2)  .. (10.9,-0.2) node[midway,below] {\(\cguard(\neg\cond)\)};

\end{tikzpicture}
\caption{Labels of while}\label{fig:labwhile}
\end{figure}
}

\n{\figthun}{
\begin{figure}
\begin{tikzpicture}
\draw (-0.5,0) ellipse (6pt and 6pt) node {\(\ell_1\)} ;
\draw[->] (-0.3,0.04) -- (1.2,0.4);
\draw[<-] (-0.3,-0.04) -- (1.2,-0.4);
\draw (2,0) ellipse (40pt and 20pt) node {\(\Labs{\cmd_1}\)};
\draw (4.5,0) ellipse (6pt and 6pt) node {\(\ell_2\)};
\draw[<-] (4.3,0) -- (3,0);
\draw[->] (4.7,0.04) -- (6.2,0.4);
\draw[<-] (4.7,-0.04) -- (6.2,-0.4);
\draw (7,0) ellipse (40pt and 20pt) node {\(\Labs{\cmd_2}\)};
\draw (9.5,0) ellipse (6pt and 6pt) node {\(\ell_3\)};
\draw[<-] (9.3,0) -- (8,0);
\end{tikzpicture}
\caption{Labels of sequences}\label{fig:labseq}
\end{figure}
}

\n{\mytable}{%
\begin{floatingtable}[r]{
\begin{tabular}{l|c|c|c|c|}
& \loc & \progint & \multicolumn{2}{c|}{\progP} \\ \cline{3-5}
& & time & time & \begin{tabular}{l}false\\ alarms\end{tabular} \\
\hline
\messag & 65 & 0.05 & 0.20s & 0 \\
\hline
\embarque & 27 100 & - & 0.34s & 7 \\
\hline
\douze & 342 & - & 3.7s & 1 \\
\hline
\quinze & 414 & 3.8  & - & - \\
\hline
\end{tabular}} \caption{Benchmarks}\label{table:benchmarks}
\end{floatingtable}
}

\n{\tableabstraction}{
\begin{floatingtable}%
{\small
\begin{tabular}[r]{|p{4.82cm}|p{2.6cm}|}
 Concrete function & Abstract function\!\!\\
 \hline
 \(\func{(i,P,\m,\h)}{(i,P,\ecr_{lv:=e}(\m),\h)}\)& \(\aecr{lv:=e}: \Da \rightarrow \Da\) \!\!\\
 \hline
 \(\func{\concr{A},\concr{S}}{\cinter{\concr{A}}(\concr{S})}\)&\(\acinter{}:\DIa \times \Da \rightarrow \Da\)\!\!\\
 \hline
 \(\func{\concr{S}}{\{(i,P,\m,\h) , (i,P',\m',\h')\in\Tra  |  \newline (i,P,\m,\h) \in \concr{S} \wedge \m'\in\ecr_{lv:=e}(\concr{S}) \} }\) & \(\iaecr{lv:=e} :\newline \Da \rightarrow \DIa\)\!\!\\
 \hline
 \(\func{S}{\{(i,P,\m,h)\in S\mid  \evb(\m,\cond)=\vrai \}}\) & \(\force_{\cond}\!\! : \Da \rightarrow \Da \)\!\!\\
 \hline
\end{tabular}\vspace{-2mm}}\caption{Given abstractions\vspace{-2mm}}\label{fig:abstractions}
\end{floatingtable}
}

\n{\examples}{%
\begin{figure}%
{%
\nlnew%
\subfloat[\(\nlrefu{init1}\excmd_1,\lend\)]{\label{subfig:while}%
\(\begin{array}{l}%
\nl[init1] x:=0;
\\\nl[a1] \cwhile(\vrai)
\\{}\quad\{\nl[a2]\ccreate(\nl[a3] x:=x+1)\},\lend
\end{array}
\)
}
\subfloat[\(\nlrefu{init2}\excmd_2,\lend\)]{\label{subfig:create}
\(
\begin{array}{l}
\nl[init2] x:=0;\nl[b1] y:=0;
\\\nl[b2] \ccreate(\nl[b3] x=x+y);
\\\nl y:=3,\lend
\end{array}
\)
}\\
\subfloat[\(\li1\excmd_3,\lend\)]{\label{subfig:exemple}
\(\begin{array}{l}
 \li1 y:=0 ; \li2 z:=0;\\
\li3\ccreate(\li4 y:=y+z);\\
\li5 z:=3,\lend
\end{array}\)
} 
\subfloat[\(\nlrefu{init4}\excmd_4,\lend\)]{\label{fig:example:inter}\begin{tabular}{l}
 \(\nl[init4] y:=0 ; \nl z:=0;\)\\
\(\nl\ccreate(\nl \label{label:trois} y:=3);\)\\
\(\nl y:= 1; \label{label:un}   \nl \label{label:deux} z:=y,\lend\)
\end{tabular}}
}\caption{Program Examples}\label{fig:examples}
\end{figure}
}

\n{\myexamplebis}{
\begin{floatingfigure}{3.3cm}\hspace{-9mm}
\fbox{\hspace{-1mm}
\begin{tabular}{l}
 \(\nl y:=0 ; \nl z:=0;\)\\
\(\nl\ccreate(\nl \label{label:trois} y:=3);\)\\
\(\nl y:= 1; \label{label:un}   \nl \label{label:deux} z:=y,\lend\)
\end{tabular}\hspace{-2mm}
}
\caption{Example}\label{fig:example:inter}
\end{floatingfigure}
}

\n{\myexample}{
\begin{floatingfigure}{3.5cm}\hspace{-9mm}
\fbox{\hspace{-1mm}
\begin{tabular}{l}
 \(\nl y:=0 ; \nl z:=0;\)\\
\(\nl\ccreate(\nl y:=y+z);\)\\
\(\nl z:=3,\lend\)
\end{tabular}\hspace{-2mm}
}
\caption{Example}\label{fig:example}
\end{floatingfigure}
}

\n{\subformulareach}{(\restrict{\concr{G}}{\after(\sinit)}\cap\TR{\lab \stmt,\ell'})\cup \restrict{\concr{A}}{\compl{\after(\sinit)}}
}

\n{\formulareach}[1][\send]{\left\{(\sinit,#1)\left|
\begin{array}{l}
 (\sinit,#1) \in \big[
\subformulareach
\big]^{\star}
\\
 \wedge \thread(\sinit)=\thread(#1) \wedge \lbl(\sinit)=\ell
\end{array}
\right.\right\}}

\n{\formaldefsem}{%
\begin{definition}\label{def:concrsem}
\[\av{\lab \stmt,\ell'}\qc \egdef \langle \concr{S}',
 \concr{G}\cup\Se\cup\Su\cup\Sud,
 \concr{A}\cup\Su\cup\Sud   \rangle  \]%
\[\ssem{\lab \stmt,\ell'}\qc \egdef \nq \]%
  where:
 \begin{align*}
\col &= \formulareach
 \\
 \concr{S}' &= \{\send | \send\in\col\langle\concr{S}\rangle \wedge \lbl(\send)=\ell' \}  \\
 \Se &= \{(\sinter,\sinterb)\in \TR{\lab \stmt,\ell'} | \sinter\in \col\langle\concr{S}\rangle \}\\
\Su &= 
\{(\sinter,\sinterb)\in\TR{\lab \stmt,\ell'}| \exists \sinit\in\concr{S} : (\sinit,\sinter)\in \col;\Sche \wedge \sinter\in\after(\sinit) \}
\\
\ext{\sinit}{\send} &= \big[
(\restrict{\concr{G}}{\after(\sinit)}\cap\TR{\lab \stmt,\ell'})\cup \restrict{\concr{A}}{\compl{\after(\sinit)}} \cup \restrict{\concr{G}}{\after(\send)} 
\big]^{\star}\\
\Sud& = 
\left\{(\sinter,\sinterb)\left|\begin{array}{l}\exists \sinit,\send \in\concr{S}\times\concr{S'} : (\sinit,\send)\in\col\wedge  \\(\send,\sinter)\in \ext{s_0}{s_1} \wedge \sinter\in\after(\sinit)\smallsetminus \after(\send)\end{array}
\right.\right\}
 \end{align*}
  \end{definition}
}

\n{\auxiliarydefs}{\mytopfigureb[fig:opfunc]{Auxiliary definitions}{
 \(\thread(i,P,\m,\h) \egdef i\)\\
  \(\lbl(i,P,\m,\h) \egdef P(i)\)\\
  \( \after(i,P,\m,\h)\egdef\{(j,P',\m',\h\cdot\h')\in\St| j\in\desce_{\h'}(\{i\}) \} \)\\
 For \(X\subset \wp(\ids)\) :
  \begin{itemize}
   \item  
  \(\desce_{\epsilon}(X) \egdef X\)
   \item 
  and \(\desce_{(i,\ell,j)\cdot\h}(X)\egdef
 \begin{cases}
  \desce_{\h}(X\cup\{j\}) & \text{if }i\in X\\  
  \desce_{\h}(X) & \text{else}
 \end{cases}
  \)
  \end{itemize} 
 }
}

\n{\figrules}{
\begin{figure}[t]
{
\parbox{\textwidth}{\center
\vspace{\flgth}

\(\begin{array}{cc}\hspace{-0.25cm}
 \regle[\hspace{-0.1cm}assign]{\m'= \ecr_{lv:=e}(\m)}{%
 \li1 lv:=e,\ell_2 \vdash (\ell_1,\m) \rightlocarrow (\ell_2,\m') }
& 
\hspace{-0.15cm}\regle[\hspace{-0.06cm}guard]{\evb(\m,\cond)=\vrai}{%
\hspace{-0.1cm}\li1 \cguard(\cond),\ell_2 \vdash
(\ell_1,\m)\rightlocarrow(\ell_2,\m)}
\end{array}\)\\%
\(\begin{array}{cc}\hspace{-0.2cm}
\regle[while entry]
{\li1 \cguard(\cond),\ell_2\vdash t}
{\li1 \cwhile(\cond)\{\li2\cmd\},\ell_3\vdash t}
&
\regle[while exit]
{\li1 \cguard(\guardn),\ell_3\vdash t}
{\li1 \cwhile(\guardp)\{\li2\cmd\},\ell_3 \vdash t}
\end{array}\)\\%
\(\regle[then]
{\li1 \cguard(\guardp),\ell_2\vdash t}
{\li1 \cifte{\guardp}{\{\li2\cmd_1\}}\{\li3cmd_2\},\ell_4 \vdash t}\)\\
\(\regle[else]
{\li1 \cguard(\guardn),\ell_3\vdash t}
{\li1 \cifte{\guardp}{\{\li2\cmd_1\}}\{\li3cmd_2\},\ell_4 \vdash t}\)\\
\(\regle[parallel]{P(i)=\ell  \sepprem \li1 \stmt,\ell_2\vdash (\ell,\m)\rightlocarrow (\ell',\m') }{
\li1 \stmt,\ell_2\Vdash (i,P,\m,\h) \rightarrow (i,
\fx{P}{i}{\ell'}
,\m',\h)
}\)\\
\(\regle[spawn]{ P(i)=\ell_1 \sepprem \text{\(j\) is fresh 
in \((i,P,\m,\h)\)}\sepprem P'=\fx{ \fx{P}{i}{\ell_{3}} }{j}{\ell_2} }%
{%
\li1\crcreate(\ell_2),\ell_3 \Vdash
(i,P,\m,\h) \rightarrow (i,
P'
,\m,h\cdot (i,\ell_2,j))}\)\\
\(\regle[then body]{\li2\cmd,\ell_4\Vdash \tau }%
{%
\li1 \cifte{\guardp}{\{\li3\cmd_2\}}\{\li3cmd_2\},\ell_4 \Vdash \tau}
\)\\
\(\regle[else body]{\li2\cmd,\ell_4\Vdash \tau }%
{%
\li1 \cifte{\guardp}{\{\li2\cmd_1\}}\{\li3cmd_2\},\ell_4 \Vdash \tau}
\)\\
\(\begin{array}{ccc}
\regle[create]{\li1\crcreate(\ell_2),\ell_{3} \Vdash\tau}{\li1\ccreate(\li2\cmd),\ell_{3} \Vdash\tau}%
&\quad&
\regle[while body]{\li2\cmd,\ell_1\Vdash \tau }%
{%
\li1 \cwhile(\cond)\{\li2\cmd\},\ell_3 \Vdash \tau}
\end{array}\)\\
\(\begin{array}{ccc}
\regle[sequence 1]{\li{1}\cmd_1,\ell_{2}\Vdash \tau }%
{%
\li1 \cmd_1;\li2\cmd_2,\ell_3 \Vdash \tau}
&\quad&
\regle[sequence 2]{\li{2}\cmd_2,\ell_{3}\Vdash \tau }%
{%
\li1 \cmd_1;\li2\cmd_2,\ell_3 \Vdash \tau}
\end{array}\)\\
\(\begin{array}{ccc}
\regle[child]{\li2\cmd,\lend\Vdash \tau }%
{%
\li1 \ccreate(\li2\cmd),\ell_3 \Vdash \tau}
&\quad&
\regle[schedule]{
\text{\(P(j)\) is defined}\sepprem i\neq j}
{
\lab \stmt,\ell'\Vdash (i,P,\m,\h)\rightarrow(j,P,\m,\h)
}\end{array} \)
}}
\caption{Operational semantics rules}\label{figrules}
\end{figure}
}

\n{\pluspres}{\hspace{-0.2cm}}
\n{\figsyntax}{%
\begin{figure}[t]\fbox{\parbox{\textwidth}{\center%
\vspace{\flgth}%
\(
 \hspace{-0.5cm}\begin{array}{cclr}
 lv & ::= \pluspres&\pluspres & \text{left value}\\
  & |\pluspres &\pluspres x & \text{variable}\\
  & |\pluspres &\pluspres\pluspres {}^{\ast}e\hspace{-1cm} & \hspace{-1cm}\text{pointer deref}\\
e & ::= \pluspres&\pluspres & \hspace{-1cm}\text{expression}\\
 & |\pluspres & c & \text{constant}\\
 & |\pluspres &\pluspres lv & \text{left value}\\
 & |\pluspres &\pluspres o(e_1,e_2)\hspace{-0.3cm} & \text{operator}\\
 & |\pluspres &\pluspres \&x &\text{address}\\
\cond & ::= & &\pluspres \hspace{-1cm}\text{condition}\\
 & |\pluspres &\pluspres x & \text{variable}\\
 & |\pluspres &\pluspres \neg \cond\hspace{-0.3cm} & \text{negation}\\
 \end{array}
 \begin{array}{cclr}
 \cmd & ::= \pluspres& \pluspres&\pluspres \text{command}\\
 & |\pluspres &\pluspres \lab lv := e & \text{assignment}\\
 & |\pluspres &\pluspres \cmd_1;\cmd_2\hspace{-1cm} & \text{sequence}\\
 & |\pluspres &\pluspres \cifte{\cond}{\{\cmd_1\}}\{\cmd_2\} \hspace{-1.8cm}& \text{if}\\
 & |\pluspres &\pluspres \lab \cwhile(\cond)\{\cmd\}\hspace{-1cm}&\text{while}\\
 & |\pluspres &\pluspres \lab \ccreate(\cmd)&\text{new thread}\\
\stmt & ::= \pluspres& \pluspres& \text{statement}\\
 & |\pluspres &\pluspres \cmd,\ell' & \text{command}\\
 & |\pluspres &\pluspres \lab \cguard(\cond),\ell' & \text{guard}\\
 & |\pluspres &\pluspres \lab \crcreate(\ell''),\ell'&\text{new thread}\\
\end{array}
\)}}\caption{Syntax}\label{figsyntax}
\end{figure}}

\n{\figabstractsemantics}{%
\n{\milieu}{\abstr{Q}&\egdef&}%
\begin{figure}[t]\center\fbox{
\(\begin{array}{rcl}
\asem{\lab lv:=e} \milieu  \assign{lv:=e}\abstr{Q} \\ 
\asem{\li{1}\cmd_1;\li{2}\cmd_2} \milieu \asem{ \li{2}\cmd_2}\circ \asem{\li{1}\cmd_1} \abstr{Q} \\
\asem{\li1\cwhile(\guardp)\{\li2\cmd\}} \milieu \guard{\guardn} \loopwa\eomega \abstr{Q}\\
\text{with } \loopwa (\abstr{Q'}) &\egdef& \big(\asem{\cmd} \circ \guard{\guardp} \abstr{Q'}\big)\sqcup\abstr{Q'}\\
\asem{\li1\ccreate(\li2\cmd)}\milieu \glue{\abstr{Q}'}\circ \afp{\li2\cmd}\eomega \circ \afinit{\ell_2} (\abstr{Q})
\\
\text{with } \abstr{Q'} &\egdef& \aspawn{\ell_2}(\abstr{Q})
\end{array}\)}
\caption{Abstract semantics}\label{fig:abstract}
\end{figure}%
}

\n{\extendedsyntax}{%
\myfigure[mfeatures]{Extended syntax}{
\item[]%
\(\asem{\li0\cpar\{\li1\cmd_1|\li2\cmd_2\}}(\abstr{Q})\egdef\langle \abstr{C_1}\sqcap\abstr{C_2},\abstr{L},
\abstr{K'}
,\abstr{I}_1\sqcup\abstr{I}_2  \rangle\)
\\ with  
\(\qup{_1}= \afp{\li1\cmd_1,\lend} \circ \afinit{\ell_1} (\abstr{Q})  \)%
\\and \(\qup{_2}= \afp{\li2\cmd_2,\lend} \circ \afinit{\ell_2} (\abstr{Q})  \)%
\\and \(\abstr{K'}=\fx{\fx{\abstr{K}}{\ell_1}{\abstr{K_2}(\fin)\sqcup\abstr{K}(\ell_1)}}{\ell_2}{\abstr{K}_1(\fin)\sqcup\abstr{K}(\ell_2)}\)%
}%
}

\n{\basicsemfuncs}{
\begin{figure}[t]\begin{center}
 \fbox{\parbox{0.98\textwidth}{\vspace{-4mm}%
 \[\hspace{-0.2cm}\begin{array}{rcl}
  \cinter{\concr{A}}(\concr{S})&\egdef& \left\{s'\left| \exists s\in\concr{S}  :
   \begin{array}{l} 
   (s,s')\in(  \restrict{\concr{A}}{\compl{\after(s)}}\cup\Sche)^{\star}
   \\
   \wedge \thread(s)=\thread(s')
   \end{array}
     \right. \right\}
  \\
  \cextract{\ell}&\egdef&\left\{s'\left| \begin{array}{l}\exists s = (i,P,\m,\h\cdot(i,\ell,j))\in\St : \\  s'\in\after(s)\end{array}\right.\right\}
  \\
   \cschedule(\concr{S}) &\egdef& \left\{(j,P,\m,\h')\left| \exists i,\h : \begin{array}{l}(i,P,\m,\h')\in \concr{S}\\ \wedge \h'=\h\cdot(i,\ell,j)\end{array}\right.\right\}
   \\
\cinit_{\ell}(\qc)&\egdef  &
 \langle \cinter{\concr{A}\cup (\restrict{\concr{G}}{\cextract{\ell}})}\circ \cschedule(\concr{S}),\\ & & {}\quad\Sche,
  \concr{A}\cup(\restrict{\concr{G}}{\cextract{\ell}}) \rangle
 \\
 \ccomb_{\qc}(\concr{G}')&\egdef&  \langle \cinter{\concr{A}\cup\concr{G}' }(\concr{S}),\concr{G}\cup\concr{G}'
, \concr{A}\cup\concr{G}'  \rangle
\\
\cexe_{f,\concr{S},\concr{A}}(\concr{\concr{G}}) &\egdef& \concr{G}' \text{ with } \langle \concr{S}',\concr{G'},\concr{A}'\rangle = f \qc
\\
\fpd{f}\qc &\egdef& \cexe_{f,\concr{S},\concr{A}}\eomega (\concr{G})
   \end{array}\]\vspace{-4mm}
}}
\caption{Basic semantic functions}\label{concrfcts}
\end{center}
\end{figure}
}

\n{\figbasicabstractsem}{%
\begin{figure}[t]\center
 \fbox{
\vspace{-4mm}\begin{tabular}{rcl}
 $\assign{lv:=e} \qu$ &$\egdef$& $\langle \acinter{\abstr{I}}\circ\aecra{lv:=e}{\abstr{C}}, \abstr{L},\abstr{K''} ,\abstr{I}\rangle$\\
 \multicolumn{3}{l}{
 $\text{with } \abstr{K''}=\func{\ell}{\text{if } \ell\in\abstr{L} 
 \text{ then } \abstr{K}(\ell)\sqcup \iaecra{lv:=e}{\abstr{C}} \text{ else } \abstr{K}(\ell)    }$
 }
 \\
 $\guard{\guardp}\qu $&$\egdef$&$ \langle \acinter{\abstr{I}}\circ\force_{\guardp}(\abstr{C}),\abstr{L},\abstr{K},\abstr{I} \rangle$ 
 \\
 $\aspawn{\ell}\qu$&$\egdef$&$\langle \abstr{C},\abstr{L}\cup\{\ell\},\abstr{K},\abstr{I} \rangle$
 \\
 $\afinit{\ell} \qu $&$\egdef$&$ \langle \acinter{\abstr{I}\sqcup\abstr{K}(\ell) }(\abstr{C}), \abstr{L}, \func{\ell}{\bot}, \abstr{I}\sqcup\abstr{K}(\ell) \rangle$
 \\
 $\glue{\qu} (\abstr{K}') $&$\egdef$&$ \langle \acinter{\abstr{I}\sqcup\abstr{K'}(\fin)  }(\abstr{C}),\abstr{L},\abstr{K}\sqcup\abstr{K'},\abstr{I}\sqcup \abstr{K'}(\fin) \rangle$
 \\
 $\aexe_{\lab \cmd,\ell',\abstr{C},\abstr{L},\abstr{I}}({\abstr{K}}) $&$\egdef$& $\concr{\abstr{K}}'$\\
  $\text{with } \qup' $&$=$&$ \asem{\lab \cmd,\ell'} \qu$
  \\
  $\afp{\lab \cmd,\ell'}(\qu)$&$\egdef$&$\aexe_{\lab \cmd,\ell',\abstr{C},\abstr{L},\abstr{I}}\eomega({\abstr{K}})$
\end{tabular}\vspace{-4mm}
}
\caption{Basic abstract semantic functions}\label{fig:basicabstract}
\end{figure}
}

\let\newcommandbis\newcommand
\newcommand{\fleche}{\mapsto}
\newcommand{\func}[2]{\lambda #1. #2 } 
\newcommand{\fx}[3]{#1[#2 \fleche #3]} 
\newcommand{\nf}[3]{#1\smallsetminus[#2 \fleche #3]}  
\newcommand{\ndf}[2]{#1(#2)\uparrow} 

\newcommandbis{\notationdesfonctions}{
\paragraph{Notation :}  \( \fx{f}{z}{t} \) refer to \(f\) in which the value of \( z \) as change : it is now defined and is value is \( t \).

\( \nf{f}{z}{t} \) refer to \(f\) in which the value of \( z \) which was \(t\) has changed : it is now undefined 

\(\ndf{f}{x}\) will mean that \(f(x)\) is not defined
}

\newcommandbis{\notationdesfonctionslight}{%
\paragraph{Notation :}  \( \fx{f}{z}{t} \) refer to \(f\) in which the value of \( z \) as change : it is now defined and is value is \( t \).}

\newtheorem{definition}{Definition}
\newtheorem{proposition}{Proposition}
\newtheorem{theorem}{Theorem}
\newtheorem{lemma}{Lemma}
\newtheorem{claim}{Claim}

\begin{document}

 \title{From Single-thread to Multithreaded: An Efficient Static Analysis Algorithm}

\author{
  Jean-Loup Carr\'e
   \and 
  Charles Hymans
}

\maketitle

\begin{abstract}
A great variety of static analyses that compute safety properties of single-thread programs have now been developed. This paper presents a systematic method to extend a class of such static analyses, so that they handle programs with multiple POSIX-style threads. Starting from a pragmatic operational semantics, we build a denotational semantics that expresses reasoning \textit{\`a la} assume-guarantee. The final algorithm is then derived by abstract interpretation. It analyses each thread in turn, propagating interferences between threads, in addition to other semantic information. The combinatorial explosion, ensued from the explicit consideration of all interleavings, is thus avoided. The worst case complexity is only increased by a factor $n$ compared to the single-thread case, where $n$ is the number of instructions in the program. We have implemented prototype tools, demonstrating the practicality of the approach.
\end{abstract}

\section{Introduction}

Many static analyses have been developed to check safety properties of sequential programs  \cite{mine:LCTES06,AllamigeonGodardHymansSAS06,steensgaard96pointsto,CousotCousot04-WCC,mine:padoII} while
more and more software applications are multithreaded. Naive approaches to analyze such applications would run by exploring all possible interleavings, which is impractical. Some previous proposals avoid this combinatorial explosion (see Related Work). Our contribution is to show that \emph{every} static analysis framework for single-thread programs extends to one that analyzes multithreaded code with dynamic thread creation and with only a modest increase in complexity.
We ignore concurrency specific bugs, e.g., race conditions or deadlocks, as do some other authors \cite{DBLP:conf/concur/LammichM07}. If any, such bugs can be detected using orthogonal techniques \cite{locksmith,701315}.

\paragraph{Outline} We describe in Section \ref{syntax} a toy imperative language.
 This contains essential features of C with POSIX threads \cite{posix-but} with a thread creation primitive. 
The main feature of multithreaded code is that parallel threads may \emph{interfere}, i.e., side-effects of one thread may change the value of variables in other threads.
To take interference between threads into account, we model the behavior of a program by an infinite transition system: this is the operational semantics of our language, which we describe in Section \ref{subsection:evol}.
It is common practice in abstract interpretation to go from the concrete to the abstract semantics through an intermediate so-called collecting semantics \cite{CousotCousot92-1}. In our case a different but similar concept is needed, which we call \cname semantics, and which we introduce in Section \ref{section:nonst}.
This semantics will discover states, accumulate transitions encountered in the current thread and collect interferences from other threads. 
The main properties of this semantics---Proposition \ref{prop:basic} and Theorem \ref{theorem:denot}---are the technical core of this paper. These properties allow us to overapproximate the \cname semantics by a denotational semantics. Section \ref{abstract} then derives an abstract semantics from the \cname semantics through abstract interpretation. We discuss algorithmic issues, implementation, question of precision, and possible extensions in Section \ref{algosm}, and examine the complexity of our analysis technique in section \ref{section:complexity}, and conclude in Section \ref{section:conclusion}.

\paragraph{Related Work}
A great variety of static analyses that compute safety
properties of single-thread
programs have been developed, e.g., intervals \cite{CousotCousot04-WCC}, points-to-graph \cite{andersen94program,steensgaard96pointsto}, non-relational stores \cite{mine:LCTES06,AllamigeonGodardHymansSAS06} or relational stores such as octagons \cite{mine:padoII}. 

Our approach is similat to Rugina and Rinard \citerinard, in the sens that  we also use an abstract semantics that derives tuples containing information about current states, transitions of the current thread, and interference from other threads. While their main parallel primitive is \(\cpar\), which runs too threads ans waits for their completion before resuming computation, we are mostly interested in the more challenging thread creation primitive \(\ccreate\), which spawn a thread that can survive its father. In Section \ref{improvement}, we handle \(\cpar\) to show how they can be dealt with our techniques.

Some authors present generalizations of specific analyses to multithreaded code, e.g., Venet and Brat \cite{996869} and Lammich and M\"uller-Olm \cite{DBLP:conf/concur/LammichM07}, while our framework extends any single-threaded code analysis.

Our approach also has some similarities with Flanagan and Qadeer \cite{DBLP:conf/spin/FlanaganQ03}.
 They use a
model-checking approach to verify multi-threaded programs.
Their algorithm computes a guarantee condition for each thread; 
one can see our static analysis framework as computing a guarantee, too.
Furthermore, both analyses abstract away both number and ordering
of interferences from other threads.  Flanagan and Qadeer's approach
still keeps some concrete information, in the form of triples containing a
thread id, and concrete stores before and after transitions.  They
claim that their algorithm takes polynomial time in the size of
the computed set of triples.  However, such sets can have exponential
size in the number of global variables of the program.  When the
nesting depth of loops and thread creation statements is bounded,
our algorithm works in polynomial time.  Moreover, we demonstrate
that our analysis is still precise on realistic examples.
Finally, while Flanagan and Qadeer assume a given, static, set of
threads created at program start-up, we handle dynamic thread
creation.
The same restriction is required in Malkis \emph{et al.\@} \cite{DBLP:conf/sas/MalkisPR07}.

The 3VMC tool \cite{360206} has a more general scope. This is an extension of TVLA designed to do shape analysis and to detect specific multithreaded bugs. However, even without multithreading, TVLA already runs in doubly exponential time \cite{repsperso}.

Other papers focus on bugs that arise because of multithreading primitives. This is orthogonal to our work. See \cite{1040299,reduction75} for atomicity properties, Locksimth and Goblint tools \cite{locksmith,SPLST/Vojdani07,Vene03globalinvariants} for data-races and \cite{701315} for deadlock detection using geometric ideas.

\section{Syntax and Operational Semantics}
\label{syntax}

\figsyntax

\examples

\subsection{Simplified Language.} The syntax of our language is given in Fig.~\ref{figsyntax}. The syntax of the language is decomposed in two parts: commands (\(\cmd\)) and statements (\(\stmt\)). A statement \(\cmd,\ell'\) is a command with a return label where it should go after completion. E.g., in Fig \ref{subfig:while}, a thread at label \(\ell_3\) will execute \(\nr[a2]\ccreate(\nr[a3] x:=x+1),\nlrefb{a1}\).
Commands and statements are labeled, and we denote by \(\Lab\) the set of labels. Labels represent the control flow: the statement \(\lab \stmt,\ell'\) begins at label \(\ell\) and terminates at label \(\ell'\), e.g., in Fig \ref{subfig:create}, a thread at label \(\ell_2\) will execute the assignment \(x:=x+1\) and go to label \(\ell_3\). It is assumed that in a given command or statement each label appears only once. Furthermore, to represent the end of the execution, we assume a special label \(\lend\) which never appears in a command, but may appear as the return label of a statement. Intuitively, this label represents the termination of a thread: a thread in this label will not be able to execute any statement.

Notice that sequences \(\cmd_1;\cmd_2\) are not labeled. Indeed, the label of a sequence is implicitly the label of the first command, e.g., the program of Fig. \ref{subfig:create} is a sequence labeled by \(\nlref{init2}\). We write \(\lab \cmd\) when the label of \(\cmd\) is \(\ell\) and we write \(\lab \stmt,\ell'\) the statement \(\stmt\) labeled by \(\ell\) and \(\ell'\).
 A program is represented by a statement of the form \(\lab \cmd,\lend\). Other statements represent a partial execution of a program.
  The statements \(\ccreate\), \(\cwhile\) and \(\cif\) are not atomic, there are composed of several basic steps, e.g., to enter in a \(\cwhile\) loop. To model these basic steps, we introduce the statements \(\li1\crcreate(\ell_2),\ell_3\) and \(\li1\cguard(\cond),\ell_2\). Then, the semantics of  \(\ccreate\), \(\cwhile\) and \(\cif\) will be defined using the semantics of \(\li1\crcreate(\ell_2),\ell_3\) and \(\li1\cguard(\cond),\ell_2\).
 Local variables are irrelevant to our work. Then, all variables in our language are global. 
 Nevertheless, local variables have been implemented (See Section \ref{algosm}) as a stack.
 
This is a toy imperative language with dynamic thread creation. It can easily be extended to handle real-world languages like C or Ada, see Sections \ref{subsec:extensions} and \ref{algosm}.

\subsection{Description of the system.} To represent threads, we use a set \(\ids\) of \emph{thread identifiers}. During an execution of a program, each thread is represented by a different identifier.
We assume a distinguished identifier \(\main \in\ids\), and take it to denote the initial thread.

When a program is executed, 
threads go from a label to another one independently.
A \emph{control point} is a partial function \(P\) that maps thread identifiers to labels and that is defined in \(\main\). 
A control point associates each thread with its current label. The domain of \(P\) is the set of created threads, the other identifiers may be used after in the execution, for new threads. Let \(\mathbb{P}\) be the set of control points.
We write \(\dom(P)\) the domain of \(P\) and let \(\fx{P}{i}{\ell}\) be the partial function defined by   \(\fx{P}{i}{\ell}(j)\egdef
\begin{cases}
 \ell & \text{if }i=j\\
 P(j) &\text{if }i\in\dom(P)\smallsetminus\{j\}\\
 \text{undefined}& \text{else}\\
\end{cases}
\)

Furthermore, threads may create other threads at any time. A \emph{genealogy} of threads is a finite sequence of tuples \((i,\ell,j)\in\ids\times\Lab\times\ids\) such that 
\begin{inparaenum}[(a)]
 \item each two tuples \((i_1,\ell_1,j_1)\) and \((i_2,\ell_2,j_2)\) have distinct third component (i.e., \(j_1\neq j_2\)),
 \item \(\main\) is never the third component of a tuple.
\end{inparaenum}
Such a tuple \((i,\ell,j)\) means that thread \(i\) created thread \(j\) at label \(\ell\). We write \(j\) has been \emph{created} in \(\h\) to say that a uple \((i,\ell,j)\) appears in \(\h\). Let \(\Hist\) be the set of genealogies. We write \(\h\cdot \h'\) the \emph{concatenation} of the genealogies \(\h\) and \(\h'\).
The hypothesis (a) means that a thread is never created twice, the hypothesis (b) means that the thread  \(\main\) is never created: it already exists at the begining of the execution.

We let \(\stores\) be the set of \emph{stores}.
We leave the precise semantics of stores undefined for now, and only require two primitives 
 \(\ecr_{lv:=e}(\m)\) and \(\evb(\m,\cond)\). Given a store \(\m\), \(\ecr_{lv:=e}\) returns the store modified by the assignment \(lv:=e\). 
 The function \(\evb\) evaluates a condition \(\cond\) in a store \(\m\), returning \(\vrai\) or \(\faux\).

A uple \(\cstate{}\in\ids\times \mathbb{P}\times \stores \times \Hist\) is a \emph{state} if
 \begin{inparaenum}[(a)]
  \item \(i\in\dom(P)\),
  \item \(\dom(P)\) is the disjoint union between \(\{\main\}\) and the set of threads created in \(\h\).
 \end{inparaenum}
 Let \(\St\) be the set of states.
A state is a tuple \( (i,P,\m,\h)\) where \(i\) is the currently running thread, \(P\) states where we are in the control flow, \(\m\) is the current store and \(\h\) is the genealogy of thread creations. \(\dom(P)\) is the set of existing threads.
The hypothesis (a) means that the current thread exists, the hypothesis (b) means that the only threads that exist are the initial threads and the thread created in the past.

In the single-threaded case, only the store and the control point of the unique thread is needed. In the case of several threads, the control point of each thread is needed: this is \(P\).

There are two standard ways to model interferences between threads:
\begin{itemize}
  \item Either all threads are active, and at any time any threads can fire a transition,
 \item or, in each state there is an ``active thread'', a.k.a., a ``current thread'', and some so called schedule transitions can change the active thread.
\end{itemize}
Our model rests on latter choice: this allows us to keep track of a thread during execution.
Thread ids do not carry information as to how threads were created. This is the role of the \(\h\) component of states.

Given a program \(\li0\cmd,\lend\) the set \(\ini\) of initial states 
is the set of tuples \( (\main,P_0 , \m,\epsilon)\) where \(\dom(P_0)=\{\main\}\), \(P_0(\main)=\ell_0\), \(\m\) is an arbitrary store, and \(\epsilon\) is the empty word.

A \emph{transition} is a pair of states \(\tau=\big((i,P,\m,\h),(i',P',\m',\h\cdot\h')\big)\) such that \(\forall j\in\dom(P)\smallsetminus\{i\}, P(j)=P'(j)\) and if \((j,\ell,j')\) is a letter of \(\h'\), then \(j=i\) and \(P(i)=\ell\).

 We denote by \(\Tra\) the set of all transitions and we denote by \(\Sche\egdef\{((i,P,\m,\h),(j,P,\m,\h))\in\Tr\mid i\neq j \}\) the set of transitions that may appear in the conclusion of rule ``schedule'', respectively.
  A transition in \(\Sche\) only changes the identifier of the current thread. 
\figrules

\subsection{Evolution.} 
\label{subsection:evol}

To model interleavings, we use a small step semantics: each statement gives rise to an infinite transition system over states where edges \(s_1\rightarrow s_2 \) correspond to elementary computation steps from state \(s_1\) to \(s_2\). We define the judgment \( \li1 \stmt,\ell_2 \Vdash s_1 \rightarrow s_2\) to state that \(s_1\rightarrow s_2 \) is one of these global computation steps that arise when \(\cmd\) is executed, \returning{} to label \(\ell'\) on termination. 
To simplify semantic rules, we use an auxiliary judgment \(\li1 \stmt,\ell_2 \vdash (\ell,\m)\rightarrow(\ell',\m')\) to describe evolutions that are local to a given thread.

Judgments are derived using the rules of Fig.~\ref{figrules}.
 The rule ``parallel'' transforms local transitions into global transitions. ``While body'' and ``sequence'' rules are global because while loop and sequences may contain global subcommands, e.g., \(\li1\cwhile(x)\{\li2\ccreate(\li3x:=0)\}\).
In ``spawn'', the expression ``\(j\) is fresh in \((i,P,\m,\h)\)'' means that \(i\neq j\) and \(P(j)\) is not defined and \(j\) nevers appears in \(\h\), i.e., in \(\h\), there is no tuples \((i,\ell,i')\) with \(i\) or \(i'\) equal to \(j\). Intuitively, a fresh identifier is an identifier that has never been used (we keep track of used identifiers in \(\h\)).

We define the set of transitions generated by the statement \(\lab{}\stmt,\ell'\):
\[\TR{\lab{}\stmt,\ell'}=\{(s,s')\mid\lab \stmt,\ell'\Vdash s\rightarrow s' \}.\]

Notice that, unlike Flanagan and Qadeer \cite{DBLP:conf/spin/FlanaganQ03}, an arbitrary number of threads may be spawned, e.g., with the program \(\nlref{init1}\excmd_1,\lend\) of Fig.~\ref{subfig:while}. Therefore, \(\ids\) is infinite, an so are  \(\mathbb{P}\) and \(\TR{\lab{}\stmt,\ell'}\). Furthermore, \(\stores\) may be  infinite, e.g., if store maps variables to integers.
Therefore, we cannot have a complexity depending of cardinal of \(\TR{\lab{}\stmt,\ell'}\).

\paragraph{Example}
Let us consider stores that are maps from a unique variable to an integer. We write \([x=n]\) the store that maps \(x\) to the integer \(n\). The transitions generated by the statements extracted from Fig. \ref{subfig:while} are:
\begin{align*}
 \TR{\nr[init1] x:=0,\nlref{a1}}=&\{((i,P,[x=n],\h),(i,\fx{P}{i}{\nlrefb{a1}},[x=0],\h))\mid P(i)=\nlrefb{init1}\\& \wedge i\in\ids\wedge n\in\mathbb{Z}\}.\\
 \TR{\nr[a3] x:=x+1,\lend}=&
  \{((i,P,[x=n],\h),(i,\fx{P}{i}{\lend},[x=n+1],\h))\mid P(i)=\nlrefb{a3} \\&\wedge i\in\ids\wedge n\in\mathbb{Z}\}.
\end{align*}

\subsection{Properties of the language}\label{subsec:extensions}

Let \(\Labs{\lab \cmd,\lend}\) be the set of labels of the statement \(\lab \cmd,\lend\).

We also define by induction on commands, the set of labels of subthreads \(\Labss{\cdot}\) by \(\Labss{\li1\ccreate(\li2\cmd),\ell_3}=\Labs{\li2\cmd,\lend}\),\\
 \(\Labss{\li1\cmd_1,\li2\cmd_2,\ell_3} = 
 \Labss{\li1\cmd_1,\ell_2}
\cup 
 \Labss{\li2\cmd_2,\ell_3}
 \),\\
 \(\Labss{\li1\cifte{\guardp}{\{\li2\cmd_1\}}{\{\li3\cmd_2\}},\ell_4}=\\{}\quad\Labss{\li2\cmd_1,\ell_4}\cup\Labss{\li3\cmd_2,\ell_4}\),\\ \(\Labss{\li1\cwhile(\guardp)\{\li2\cmd\},\ell_3}=\Labss{\li2\cmd_1,\ell_1}\),\\
 and, for basic commands \(\Labss{\li1basic,\ell_2}=\emptyset\).

A statement generates only transitions from its labels and to its labels, this is formalized by the following lemma:

\begin{lemma}\label{lemma:A}
If \((s,s')\in\TR{\lab\stmt,\ell'}\smallsetminus \Sche\) then \(\lbl(s)\in\Labs{\lab\stmt,\ell'}\smallsetminus\{\ell'\}\) and \(\lbl(s')\in\Labs{\lab\stmt,\ell'}\) and \(\thread(s)=\thread(s')\).
\end{lemma}

As a consequence of Lemma \ref{lemma:A}, we have the following lemma :
\begin{lemma}
 \label{lemma:Abis}
 If \(\lbl(s)\notin \Labs{\lab \stmt ,\ell'}\smallsetminus\{\ell'\}\) then for all state \(s'\),  \((s,s')\notin\TR{\lab\stmt,\ell'}\smallsetminus \Sche\)
\end{lemma}

If, during the execution of a statement \(\lab\stmt,\ell'\), a thread creates another thred, then, the subthread is in a label of the command, furthermore, it is in \(\Labss{\lab\stmt,\ell'}\).
\begin{lemma}\label{lemma:Asub}
If \((s,s')=(\cstate{},\cstate{'})\in\TR{\lab\stmt,\ell'}\smallsetminus \Sche\) and \(j\in\dom(P')\smallsetminus\dom(P)\) then \(P'(j)\in \Labss{\lab\stmt,\ell'}\subset\Labs{\lab\stmt,\ell'}\).
\end{lemma}

\begin{lemma}\label{lemma:Asub2}
 If \((s,s')\in\TR{\lab\stmt,\ell'}\smallsetminus \Sche\) and \(\lbl(s)\in\Labss{\lab\stmt,\ell'}\smallsetminus\{\ell'\}\) then \(\lbl(s')\in\Labs{\lab\stmt,\ell'}\).
 
 Furthermore \(\ell\notin\Labs{\lab\stmt,\ell'}\) and \(\ell'\notin\Labs{\lab\stmt,\ell'}\).
\end{lemma}

Notice that in Fig. \ref{figrules} some statements are ``atomic''. We call these statements \emph{basic statements}.
Formally, a basic statement is a statement of the form \(\li1 lv:=e,\ell_2\), \(\li1\cguard(\cond),\ell_2\) or \(\li1 \cspawn(\ell_3),\ell_2 \).

On basic statement, we have a more precise lemma on labels:
\begin{lemma}\label{lemma:A'}Let \(\li1 basic,\ell_2\) be a basic statement.\\
 If \((s,s')=(\cstate{},\cstate{'})\in \TR{\li1 basic,\ell_2}\smallsetminus\Sche\) then \(\thread(s)=\thread(s')\) and \(\lbl(s)=\ell_1\) and \(\lbl(s')=\ell_2\).
\end{lemma}

\section{\Cnamep Semantics}
\label{section:nonst}
 
 \subsection{Basic Concepts}
 
  \auxiliarydefs
 
 To prepare the grounds for abstraction, we introduce an intermediate semantics, 
 called \cname semantics, 
 which associates a function on configurations with each statement.
 The aim of this semantics is to associate with each statement a transfer function that will be abstracted 
 (see Section \ref{abstract}) 
 as an abstract transfer function.

A \emph{concrete configuration} is a tuple \(\concr{Q}=\qc\) :
\begin{inparaenum}
 \item \(\concr{S}\) is the current state of the system during an execution,
 \item \(\concr{G}\), for \emph{guarantee}, represents what the current thread and its descendants can do
 \item and \(\concr{A}\), for \emph{assume}, represents what the other threads can do.
\end{inparaenum}

Formally, \(\concr{S}\) is a set of states, and \(\concr{G}\) and \(\concr{A}\) are sets of transitions containing \(\Sche\).
The set of concrete configurations is a complete lattice for the ordering \(\qcp{_1}\leqslant\qcp{_2} \Leftrightarrow \concr{S_1}\subset\concr{S_2}\wedge\concr{G_1}\subset\concr{G_2}\wedge\concr{A_1}\subset\concr{A_2}  \).
Proposition \ref{prop:coroc} will establish the link between operational and \cname semantics.

 \mondessinstates

Figure~\ref{fig:states} illustrates the execution of a whole program. Each vertical line represents the execution of a thread from top to bottom, and each horizontal line represents the creation of a thread. At the beginning (top of the figure), there is only the thread \(\main=j_0\). 

During execution, each thread may execute transitions. At state \(\sinit\), \(\thread(\sinit)\) denotes the \emph{currently running thread} (or \emph{current thread}), see Fig.~\ref{fig:opfunc}. On Fig.~\ref{fig:states}, the current thread of \(\sinit\) is \(j_0\) and the current thread of \(\sinter\) is \(j_2\).

During the program execution given in Fig.~\ref{fig:states}, \(j_0\) creates \(j_1\). We say that \(j_1\) is a \emph{child} of \(j_0\) and \(j_0\) is the \emph{parent} of \(j_1\). Furthermore, \(j_1\) creates \(j_3\). We then  introduce the concept of \emph{descendant}: the thread \(j_3\) is a descendant of \(j_0\) because it has been created by \(j_1\) which has been created by \(j_0\). More precisely, descendants depend on genealogies. Consider the state \(\sinit=(j_0,\Pinit,\minit,\hinit)\) with \(\hinit=[(j_0,\ell_1,j_1)]\): the set of descendants of \(j_0\) from \(\hinit\) (written \(\desce_{\hinit}(\{j_0\})\), see Fig.~\ref{fig:opfunc}) is just \(\{j_0,j_1\}\).
 The set of descendants of a given thread increases during the execution of the program. In Fig.~\ref{fig:states}, the genealogy of \(\sinter\) is of the form \(\hinit\cdot \hinterf\) for some \(\hinterf\), here \(\hinterf=[(j_0,\ell_2,j_2),(j_1,\ell_3,j_3),(j_2,\ell_4,j_4)]\). When the execution of the program reaches the state \(\sinter\), the set of descendants of \(j_0\) from \(\hinit\cdot \hinterf\) is \(\desce_{\hinit\cdot\hinterf}(\{j_0\})=\{j_0,j_1,j_2,j_3,j_4\}\).
 
In a genealogy, there are two important pieces of information. First, there is a tree structure: a thread creates children that may creates children and so on... Second, there is a global time, e.g., in \(\hinterf\), the thread \(j_2\) has been created before the thread \(j_3\).

\begin{lemma}\label{lemma:Descendant}
 Let \(\h\cdot\h'\) a genealogy and \(i\), \(j\) which are not created in \(\h'\).
 Therefore, either \(\desce_{\h'}(\{j\}) \subset \desce_{\h\cdot\h'}(\{i\}) \) or \(\desce_{\h'}(\{j\}) \cap \desce_{\h\cdot\h'}(\{i\}) = \emptyset\).
 \end{lemma}
\begin{proof}
 We prove this lemma by induction on \(\h'\).
 If \(\h'=\epsilon\), then \(\desce_{\epsilon}(\{j\}) = \{j\} \).
 
 Let us consider the case \(\h'=\h''\cdot(i',\ell,j')\).
 By induction hypothesis  either \(\desce_{\h''}(\{j\}) \subset \desce_{\h\cdot\h''}(\{i\}) \) or \(\desce_{\h''}(\{j\}) \cap \desce_{\h\cdot\h''}(\{i\}) = \emptyset\).
 
 In the first case, if \(i'\in \desce_{\h''}(\{j\})\), therefore \(j'\in \desce_{\h''\cdot(i',\ell',j')}(\{j\})\) and \(j'\in \desce_{\h\cdot\h''\cdot(i',\ell',j')}(\{i\})\), else \(j'\notin \desce_{\h''\cdot(i',\ell',j')}(\{j\})\).
 
 In the second case, let us consider the subcase \(i'\in \desce_{\h''}(\{j\})\). Therefore \(i'\notin \desce_{\h\cdot\h''}(\{i\})\). In addition to this, \(j\) is not created in \(\h\cdot\h''\) (a thread cannot be created twice in a genealogy), therefore \(j\notin \desce_{\h\cdot\h''}(\{i\})\).
 Hence \(j'\in \desce_{\h''\cdot(i',\ell',j')}(\{j\})\) and \(j'\notin \desce_{\h\cdot\h''\cdot(i',\ell',j')}(\{i\})\).
 
 The subcase \(i'\in \desce_{\h\cdot\h''}(\{i\})\) is similar. Let us consider the subcase \(i'\notin\desce_{\h''}(\{j\})\cup \desce_{\h\cdot\h''}(\{i\})\).
 Therefore \(\desce_{\h\cdot\h''\cdot(i',\ell',j')}(\{i\}) = \desce_{\h\cdot\h''}(\{i\})  \) and \(\desce_{\h''\cdot(i',\ell',j')}(\{j\}) = \desce_{\h''}(\{j\})\).
\end{proof}

 We also need to consider sub-genealogies such as \(\hinterf\). In this partial genealogy, \(j_1\) has not been created by \(j_0\). Hence \(\desce_{\hinterf}(\{j_0\})=\{j_0,j_2,j_4\}\). Notice that \(j_3\notin \desce_{\hinterf}(\{j_0\})\) even though the creation of \(j_3\) is in the genealogy \(\hinterf\).

During an execution, after having encountered a state \(\sinit=(j_0,\Pinit,\minit,\hinit)\) we distinguish two kinds of descendants of \(j_0\):
\begin{inparaenum}[(i)]
 \item \label{enum:past}those which already exist in state \(\sinit\) (except \(j_0\) itself) and their descendants,
 \item \label{enum:future}\(j_0\) and its other descendants.
\end{inparaenum}
Each thread of kind (\ref{enum:past}) has been created by a statement executed by \(j_0\).
We call \(\after(\sinit)\) the states from which a thread of kind (\ref{enum:future}) can execute a transition.
In Fig.~\ref{fig:states}, the thick lines describe all the states encountered while executing the program that fall into \(\after(\sinit)\).

The following lemma explicits some properties of \(\after\):
\begin{lemma}\label{lemma:F+}
 Let \(T\) a set of transitions.
 Let \((s_0,s_1)\in T^{\star}\) therefore:
 \begin{enumerate}
  \item If \(\thread(s_0)=\thread(s_1)\) then \(s_1\in\after(s_0)\) 
  \item If \(s_1\in\after(s_0)\) then \(\after(s_1)\subset\after(s_0)\)
 \end{enumerate}
\end{lemma}
\begin{proof}
Let \(\cstate{_0}=s_0\) and \(\cstate{_1}=s_1\). By definition of transitions, there exists \(\h'_1\) such that \(\h_1=\h_0\cdot\h'_1\).
Because \(i_0\in\desce_{\epsilon}(\{i_0\})\),  \(i_0\in\desce_{\h'_1}(\{i_0\})\). Therefore, if \(\thread(s)=\thread(s')\), i.e., \(i_1=i_0\), then \(s_1\in\after(s_0)\) (By definition of \(\after\)).

 Let us assume that \(s_1\in\after(s_0)\).
 Let \(s_2=\cstate{_2}\in\after(s_1)\). Therefore, there exists \(\h'_2\) such that \(\h_2=\h_1\cdot\h'_2=\h_0\cdot\h'_1\cdot\h'_2\) and \(i_2\in\desce_{\h'_2}(\{i_1\})\).
 Because \(s_1\in\after(s_0)\), by definition, \(i_1\in\desce_{\h'_1}(\{i_0\})\). Therefore \(i_1\in\desce_{\h'_2}(\{i_1\})\cap\desce_{\h'_1\cdot\h'_2}(\{i_0\})\). According to Lemma \ref{lemma:Descendant}, \(\desce_{\h'_2}(\{i_1\})\subset\desce_{\h'_1\cdot\h'_2}(\{i_0\})\).
  Hence  \(i_2\in\desce_{\h'_1\cdot\h'_2}(\{i_0\})\) and therefore \(s_2\in\after(s_0)\).
\end{proof}

When a schedule transition is executed, the current thread change. The futur descendants of the past current thread and the new current thread are diffents. This is formalized by the following lemma:

\begin{lemma} \label{lemma:F-}
 If \((s_1,s_2)\in \Sche\) then \(\after(s_1)\cap\after(s_2)=\emptyset\).
\end{lemma}
\begin{proof}
Let \(\cstate{_1}=s_1\) and \(i_2=\thread(s_2)\). Therefore \((i_2,P_1,\m_1,\h_1)=s_2\).
Let \(s=\cstate{}\in \after(s_1)\cap\after(s_2)\).

By definition of \(\after\), there exists \(\h'\) such that \(\h=\h_1\cdot\h'\), \(i\in \desce_{\h'}(\{i_1\})\) and \(i\in \desce_{\h'}(\{i_2\})\).
Furthermore \(i_1\) and \(i_2\) are in \(\dom(P_1)\). Therefore \(i_1\) and \(i_2\) are either created in \(\h_1\), or are \(\main\). Hence, \(i_1\) and \(i_2\) cannot be created in \(\h'\). Therefore, \(i_2\notin\desce_{\h'}(\{i_1\})\) and therefore \(\desce_{\h'}(\{i_2\})\subset\desce_{\epsilon\cdot\h'}(\{i_1\})\). Using Lemma \ref{lemma:Descendant} we conclude that \(\desce_{\h'}(\{i_1\})\cap\desce_{\h'}(\{i_2\}) = \emptyset\).
This is a contradiction with \(i\in \desce_{\h'}(\{i_1\})\) and \(i\in \desce_{\h'}(\{i_2\})\).
\end{proof}

During the execution of a set of transition \(T\) that do not create threads, the set of descendants does not increase:
\begin{lemma} \label{lemma:Ndesce}
 Let \(T\) a set of transitions such that:\\ for all \((s,s')=(\cstate{},\cstate')\in T, \h=\h'\).\\
 Let \(s_0=\cstate{_0}\),  \(s=(i,P,\m,\h_0\cdot\h)\) and \(s=(i',P',\m',\h_0\cdot\h\cdot\h')\).
 
 If \((s,s')\in( \restrict{\concr{A}}{\compl{\after(s_0)}}\cup T)^{\star}\) then \(\desce_{\h\cdot\h'}\{i_0\}=\desce_{\h}\{i_0\}\).
\end{lemma}
\begin{proof}
 Let \(s_1,\ldots,s_n\) a sequence of states such that \(s_1=s\), for all \(k\in\{1,\ldots,n-1\}\), \((s_k,s_{k+1)}\in \restrict{\concr{A}}{\compl{\after(s_0)}}\cup T)^{\star}\), and \(s_n=s'\).
 
 Let \((i_k,P_k,\m_k,\h_0\cdot\h\cdot\h_k)=s_k\).
 
 If \(\h_k\neq\h_{k+1}\) then, \((s_k,s_{k_1})\in\restrict{\concr{A}}{\compl{\after(s_0)}}\) and then \(i_k\notin \desce_{\h\cdot\h_k}\{i_0\} \) and then \(\desce_{\h\cdot\h_k}\{i\}=\desce_{\h\cdot\h_{k+1}}\{i_0\}\).
 
 Therefore, in all cases \(\desce_{\h\cdot\h_k}\{i\}=\desce_{\h\cdot\h_{k+1}}\{i\}\) and then, by straightforward induction, \(\desce_{\h\cdot\h'}\{i\}=\desce_{\h}\{i\}\).
 \end{proof}

\begin{lemma} \label{lemma:Ndescew}
 Let \(T\) a set of transitions such that:\\ for all \((s,s')=(\cstate{},\cstate')\in T, \h=\h'\).\\
 Let \(s=(i,P,\m,\h)\) and \(s=(i',P',\m',\h\cdot\h')\).
 
 If \((s,s')\in( \restrict{\concr{A}}{\compl{\after(s_0)}}\cup T)^{\star}\) then \(\desce_{\h'}\{i\}=\{i\}\).
\end{lemma}
\begin{proof}
 Apply Lemma \ref{lemma:Ndesce} with \(s_0=s\).
\end{proof}

These lemmas has a consequence on \(\after\):

\begin{lemma}\label{lemma:E'}
Let \(T\) a set of transitions such that:\\ for all \((s,s')=(\cstate{},\cstate')\in T, \h=\h'\).\\
If \((s_0,s_1)\in( \restrict{\concr{A}}{\compl{\after(s_0)}}\cup T)^{\star}\) and \(s_1\in\after(s_0)\) then
\(\thread(s_1) = \thread(s_0)\).
\end{lemma}
\begin{proof}
 Let \(\cstate{_0} =s_0 \) and \((i_1,P_1,\m,\h_0\cdot\h_1)=s_1\).
 By Lemma \ref{lemma:Ndescew} \(\desce_{\h_1}\{i_0\}=\{i_0\}\) and by definition of \(\after\), \(i_1\in \desce_{\h_1}\{i_0\}\).
\end{proof}

\begin{lemma}\label{lemma:H}
Let \(T_1\) a set of transitions such that:\\ for all \((s,s')=(\cstate{},\cstate')\in T, \h=\h'\).\\
Let \(T_2\) a set of transitions.

Let \(s_0,s_1,s\) three states such that  \((s_0,s_1)\in T_1^{\star}\), \(\thread(s_0)=\thread(s_1)\) and \((s_1,s)\in T^{\star}\).

If \(s\in\after(s_0)\) then \(s\in\after(s_1)\).
\end{lemma}
\begin{proof}
 Let \(\cstate{_0} =s_0 \), \((i_1,P_1,\m,\h_0\cdot\h_1)=s_1\) and \((i,P,\m,\h_0\cdot\h_1\cdot\h)=s\).
 By Lemma \ref{lemma:Ndescew} \(\desce_{\h_1}\{i_0\}=\{i_0\}\) and by definition of \(\after\), \(i_1\in \desce_{\h_1}\{i_0\}\).
 Therefore \(\desce_{\h_1\cdot\h}(\{i_0\}) = \desce_{\h}(\desce_{\h_1}(\{i_0\})) = \desce_{\h}(\{i_0\})\).
 
 Because \(s\in\after(s_0)\), \(i\desce_{\h_1\cdot\h}(\{i_0\})\), therefore \(i\desce_{\h}(\{i_0\})\). Hence \(s\in\after(s_1)\).
\end{proof}

\subsection{Definition of the \Cname Semantics}

 \mondessinlocdef

Let us recall some classical definitions.
For any binary relation \(R\) on states let \(\restrict{R}{S}=\{(s,s')\in R \mid s\in S\}\) be the \emph{restriction} of \(R\) to \(S\) and \(R\langle S \rangle=\{s'\mid\exists s\in S : (s,s')\in R\}\) be the \emph{application} of \(R\) on \(S\). \(R ; R' =\{(s,s'') \mid\exists s'\in\St : (s,s')\in R \wedge (s',s'')\in R' \}\) is the \emph{composition} of \(R\) and \(R'\). Let \(R^{\star} = \bigcup_{k\in\mathbb{N}}R^{k}\) where \(R^{0}=\{(s,s) \mid s\in\St\}\) and \(R^{k+1}=R; R^{k}\). Finally, for any set of states \(S\), let \(\compl{S}=\St\smallsetminus S\) be the \emph{complement} of \(\concr{S}\).

The definition of the \cname semantics \(\osem{\lab \stmt,\ell'}\) of a statement \(\lab \stmt,\ell'\) requires some intermediate relations and sets. The formal definition is given by the following definition:
 
 \formaldefsem
 
 Let us read together, on some special cases shown in Fig.~\ref{fig:illusconcr}. This will explain the rather intimidating of Definition \ref{def:concrsem} step by step, introducing the necessary complications as they come along.

The statement is executed between states \(\sinit=(j_0,P,\m,\h)\) and \(\send=(j_0,P',\m',\h\cdot \h')\).

Figure~\subref{figsub:singlethread} describes the single-thread case: there is no thread interaction during the execution of \(\lab \stmt,\ell'\). The thread \(j_5\) is spawned after the execution of the statement. 
E.g., in Fig.~\ref{subfig:create}, \(\nr[b1] y:=0;\nlref{b2}\).

In this simple case, a state \(\sinter\) is reachable from \(\sinit\) if and only if there exists a path from \(\sinit\) to \(\sinter\) using only transitions done by the unique thread (these transitions should be in the guarantee \(\concr{G}\)) and that are generated by the statement. \(\concr{S'}\) represents the final states reachable from \(\concr{S}\). Finally, in this case:
 \begin{align*}
 \col&=\{(\sinit,\send)\in\big[\concr{G}\cap \TRi{\lab \stmt,\ell'} \big]^{\star}| \lbl(\sinit)=\ell\}\\
 \concr{S}' &= \{\send \mid \send\in\col(\concr{S}) \wedge \lbl(\send)=\ell' \} \\
 \Se&=\{(\sinter,\sinterb)\in \TRi{\lab \stmt,\ell'} \mid \sinter\in \col(\concr{S}) \}\\
 \av{\lab \stmt,\ell'}&\langle \concr{S}, \concr{G} , \Sche \rangle = \langle \concr{S}',
 \concr{G}\cup\Se,
 \Sche   \rangle \vspace{-1.5mm}
 \Su&=\Sud=\emptyset
 \end{align*}

Figure~\subref{figsub:interference} is more complex: \(j_0\) interferes with threads \(j_1\) and \(j_3\).
These interferences are assumed to be in \(\concr{A}\). Some states can be reached only with such interference transitions. E.g, consider the statement \(\li{_{\ref{label:un}}} y:= 1; \li{_{\ref{label:deux}}} z:=y,\lend\) in Fig.~\ref{fig:example:inter}: 
at the end of this statement, the value of \(z\) may be \(3\), because the statement \(\li{_{\ref{label:trois}}}y:=3,\lend\) may be executed when the thread \(\main\) is at label \(\ell_{\ref{label:deux}}\). 
Therefore, to avoid missing some reachable states, transitions of \(\concr{A}\) are taken into account in the definition of \(\col\).
 In Fig.~\subref{figsub:interference}, the statement \(\lab \stmt,\ell'\) is executed by descendants of \(j_0\) of kind (\ref{enum:future}) (i.e., \(\after(\sinit)\)), and the interferences come from \(j_1\) and \(j_3\) which are descendants of kind (\ref{enum:past}) (i.e., in \(\compl{\after(\sinit)}\)).
Finally, we find the complete formula of Definition~\ref{def:concrsem}: \[\col = \formulareach. \]

In Fig.~\subref{figsub:par}, when \(j_0\) executes the statement \(\lab \stmt,\ell'\) it creates subthreads (\(j_2\) and \(j_4\)) which execute transitions in parallel of the statement. The guarantee \(\concr{G}\) is not supposed to contain only transitions executed by the current thread but also these transitions. These transitions, represented by thick lines in Fig.~\subref{figsub:par}, are collected into the set \(\Su\).
Consider such a transition, it is executed in parallel of the statement, i.e., from a state of \(\Sche\circ\col(\{s_0\})\). Furthermore, this transition came from the statement, and not from an earlier thread, hence from \(\after(\sinit)\).
\[\Su = 
\{(\sinter,\sinterb)\in\TRi{\lab \stmt,\ell'}\mid \exists \sinit\in\concr{S} : ( \sinit,\sinter)\in \Sche\circ \col \wedge \sinter\in\after(\sinit) \}.
\]

The threads created by \(j_0\) when it executes the statement \(\lab \stmt,\ell'\) may survive when this statement \returns{} in \(\send\), as shown in Fig.~\subref{figsub:sub}. Such a thread \(i\) (here, \(i\) is \(j_4\) or \(j_5\) or \(j_6\)) can execute transitions that are not in \(\Su\). \(\Sud\) collects these transitions. The creation of \(i\) results of a \(\ccreate\) statement executed between \(\sinit\) and \(\send\). Hence, such a transition \((s,s')\) is executed from a state in \(\after(\sinit)\smallsetminus\after(\send)\). The path from \(\send\) to \(s\) is comprised of transitions 
 in \((\restrict{\concr{G}}{\after(\sinit)}\cap\TRi{\lab \stmt,\ell'})\cup \restrict{\concr{A}}{\compl{\after(\sinit)}}\) (similarly to \(\col\)) and of transitions of \(j_0\) or \(j_5\) under the dotted line, i.e., transitions in \(\restrict{\concr{G}}{\after(\send)}\). 

\subsection{Properties of the \Cname Semantics }

\basicsemfuncs

To prepare for our static analysis we provide a compositional analysis of the \cname semantics in Theorem \ref{theorem:denot} below. To this end, we introduce a set of helper functions, see Fig.~\ref{concrfcts}. We define, for any extensive\footnote{A function \(f\) of domain \(D\) is \emph{extensive} if and only if for every set \(X\subset D\), \(X \subset f(X)\)} function \(f\), \(f\eomega(X)\egdef\bigcup_{n\in\mathbb{N}}f^n(X)\).
 
The function \(\cinter{\concr{A}}(\concr{S})\) returns states that are reachable from \(\concr{S}\) by applying interferences in \(\concr{A}\).
Notice that these interferences do not change the label of the current thread:
\begin{lemma}\label{lemma:B}
 Let \(s=\cstate{}\) and \(s'=\cstate'\).
 If \((s,s') \in ( \restrict{\concr{A}}{\compl{\after(s)}}\cup\Sche )^{\star}\) then \(P(i)=P'(i)\), i.e., \(\lbl(s)=P'(\thread(s))\). 
 
 If furthermore \(\thread(s)=\thread(s')\) then \(\lbl(s)=\lbl(s')\).
\end{lemma}
\begin{proof}
 There exists a sequence of states \(s_0\), \ldots, \(s_n\) such that \(s_0=s\) and \(s_n=s'\) and for all \(k\in\{0,\ldots,n-1\}\), \((s_k,s_{k+1})\in \restrict{\concr{A}}{\compl{\after(s)}}\cup\Sche\).
 
 Let \(\cstate{_k}=s_k\). Let us prove by induction that \(P_k(i)=P(i)\).
If \((s_k,s_{k+1})\in\Sche\) and \(P_k(i)=P(i)\) then \(P_{k+1}(i)=P(i)\).
If \((s_k,s_{k+1})\in \restrict{\concr{A}}{\compl{\after(s)}}\) and \(P_k(i)=P(i)\) then \(s_k\notin\after(s_k)\) and then \(i_k\neq i\) and then \(P_{k+1}(i)=P_{k}(i)=P(i)\).
\end{proof}

The function 
\(\cextract{\ell}\) computes the set of states that may be reached after having created a thread at label \(\ell\); \(\cschedule\) applies a schedule transition to the last child of the current thread. The function \(\cinit_{\ell}\) computes a configuration for the last child created at \(\ell\), taking into account interferences with its parent using \(\cextract{\ell}\); notice that we need here the genealogies to define \(\cextract{\ell}\) and then to have Theorem \ref{theorem:denot}. 
The function \(\cexe\) computes a part of the guarantee (an under-approximation), given the semantics of a command represented as a function \(f\) from configuration to configuration. And \(\fp{}\) iterates \(\cexe\) to compute the whole guarantee.

During the execution of a statement \(\lab\stmt,\ell'\), some interference transition may be fired at any time.
Nevertheless, the labels of the thread(s) executing the statement are still in a label of the statement:
\begin{lemma}
 \label{lemma:F}
 If \((s_0,s)\in (\TR{\lab\stmt,\ell'}\cup\restrict{\concr{A}}{\compl{\after(s_0)}})^{\ast}\), \(\lbl(s_0)\in\Labs{\lab\stmt,\ell'}\) and \(s\in\after(s_0)\) then \(\lbl(s)\in \Labs{\lab\stmt,\ell'}\).
 
 Futhermore, if \(\lbl(s)=\ell'\) or \(\lbl(s)=\ell\) then \(\thread(s_0)=\thread(s)\).
\end{lemma}
\begin{proof}
  There exists a path \(s_1,\ldots,s_n\) such that \(s_n=s\) and for all \(k\in\{0,\ldots,n-1\}\), \((s_k,s_{k-1})\in \TR{\lab\stmt,\ell'}\cup\restrict{\concr{A}}{\compl{\after(s_0)}}\).
 Let \(\cstate{_0}=s_0\) and for  \(k\geqslant 1\),
 let \((i_k,P_k,\m_k,\h_0\cdot\h_k)=s_k\). 
 
 Let us prove by induction on \(k\) that \(P_{k}(i)\in\Labs{\lab\stmt,\ell'}\) and for all \(j\in\desce_{\h_k}(\{i_0\})\smallsetminus\{i_0\}\), \(P_{k}(j)\in\Labss{\lab\stmt,\ell'}\).
 
 Let us assume that \(k\) satisfy the induction property, and let us show that \(k+1\) satifies the induction property.
 
 In the case \((s_k,s_{k+1})\in\restrict{\concr{A}}{\compl{\after(s_0)}}\), \(i_k\notin\desce_{\h_k}(\{i_0\})\) and then for all \(j=\desce_{\h_k}(\{i_0\})=\desce_{\h_{k+1}}(\{i_0\})\), \(P_k(j)=P_{k+1}(j)\).
 
 In the case \((s_k,s_{k+1})\in\TR{\lab\stmt,\ell'}\) and \(i_k=i_0\), by Lemma \ref{lemma:A},  \(P_{k+1}(i_k)\in\Labs{\lab\stmt,\ell'}\). Furthermore, if \(j\in\desce_{\h_k}(\{i_0\}) \) then \(P_k(j)=P_{k+1}(j)\).
 If \(j\in\desce_{\h_{k+1}}(\{i_0\})\smallsetminus\desce_{\h_k}(\{i_0\})\), then \(j\in\dom(P_{k+1})\smallsetminus\dom(P_k)\) and by Lemma \ref{lemma:Asub}, \(P_{k+1}(j)\in \Labss{\lab\stmt,\ell'}\).
 
 In the case\((s_k,s_{k+1})\in\TR{\lab\stmt,\ell'}\) and \(i_k=i_0\), we conclude similarly by Lemma \ref{lemma:Asub2}.
If \(s\in\after(s_0)\), then \(i_n\in\desce_{\h_n}(\{i_0\})\) and therefore \(\lbl(s)\in \Labs{\lab\stmt,\ell'}\).

If \(\lbl(s)=\ell'\) or \(\lbl(s)=\ell\), then, because by Lemma \ref{lemma:Asub2}, \(\ell\) and \(\ell'\) are not in \(\Labss{\lab \stmt,\ell'}\), we have \(\thread(s_0)=\thread(s)\).
\end{proof}

The following lemma summarizes the consequences on \(\col\) of Lemmas \ref{lemma:F+} and \ref{lemma:F}:

\begin{lemma}\label{lemma:R} Let \(\nq = \ssem{\lab\stmt,\ell'}\qc\).

If \((s_0,s)\in\col\) therefore \(s\in\after(s_0)\), \(\after(s)\subset\after(s_0)\) and \(\lbl(s)\in\Labs{\lab\stmt,\ell'}\). 
\end{lemma}
\begin{proof}
 \((s_0,s)\in\big[\subformulareach\big]^{\star} \), then by Lemma \ref{lemma:F+}, \(s\in\after(s_0)\) and  \(\after(s)\subset\after(s_0)\). Furthermore, by Lemma \ref{lemma:F}, \(\lbl(s)\in\Labs{\lab\stmt,\ell'}\). 
\end{proof}

The following proposition show that \(\fpd{}\) collect all transitions generated by a statement.

\begin{proposition}[Soundness of \(\fpd{}\)]\label{proposition:guarantee}
 Let \(\qc\) a concrete configuration, \(\lab \stmt,\ell'\) a statement and  
 \(\concr{G}_{\infty} = \fp{\av{\lab \stmt,\ell'}}\qc\).
 Let \(s_0\in\concr{S}\) and \(s\in\after(s_0)\) such that \((s,s')\in \TRi{\lab \stmt,\ell'}\).
 
 If \((s_0,s)\in  \big[
\restrict{(\TRi{\lab \stmt,\ell'})}{\after(s_0)}\cup \restrict{\concr{A}}{\compl{\after(s_0)}}
\big]^{\star}\) 
then \((s,s') \in \concr{G}_{\infty}\)
\end{proposition}

\begin{proof}
Let \(\qcp{_k} = \cexe_{\osem{\lab \stmt ,\ell'},\concr{S},\concr{A}}^k \concr{G}\)\\
and \(\nqp{_{k}} = \osem{\lab \stmt ,\ell'} \langle \concr{S} , \concr{G}_k , \concr{A}  \rangle \) \\
and \(T= \TRi{\lab \stmt,\ell'}\)

Let \(s_0,\ldots,s_{n+1}\) a path such that \(s_n=s\), \(s_{n+1}=s'\) and for all \(k\), \((s_k,s_{k+1}) \in \big[\restrict{T}{\after(s_0)}\cup \restrict{\concr{A}}{\compl{\after(s_0)}}
\big]^{\star}\).
Let \(m\) an arbitrary integer.
Then, let \(k_0\) the smallest \(k\) (if it exists) such that \((s_{k},s_{k+1}) \in \restrict{T}{\after(s_0)}\smallsetminus \concr{G}_{m} \).
Then, by definition, \((s_{k},s_{k+1})\in \Se_{m}\cup\Su_{m}\subset \concr{G}_{m+1}\subset\concr{G}_{\infty}\).
\end{proof}

\subsection{Basic Statements}

Basic statement have common properties, therefore, we will study them at the same time.
Proposition \ref{prop:basic} explain how to overapproximate the semantics of a basic statement.
It will be used in the abstract semantics.

An execution path of a basic statement can be decomposed in interferences, then one transition of the basic statement, and then, some other interferences.
The following lemma show this.
This lemma will allow us to prove Proposition \ref{prop:basic}.

\begin{lemma}\label{lemma:C}Let \(\li1 basic,\ell_2\) be a basic statement,\\ 
and \(\nq=\ssem{\li1 basic,\ell_2} \qc \).
Let  \((s_0,s) \in \col\)
then:
\begin{itemize}
 \item either \(s \in \cinter{\concr{A}}(\{s_0\})\) and \(\lbl(s)=\ell_1\),
 \item or \(s \in\cinter{\concr{A}}(\TR{\li1 basic,\ell_2} \smallsetminus\Sche \langle \cinter{\concr{A}}(\{s_0\}) \rangle)\) \\and \(\lbl(s)=\ell_2\)
\end{itemize}
\end{lemma}
\begin{proof}
Let us consider the case \((s_0,s) \in (\restrict{\concr{A}}{\compl{\after(s_0)}}\cup\Sche)^{\star}\).
By definition of \(\col\), \(\thread(s_0)=\thread(s)\). Therefore \(s \in \cinter{\concr{A}}(\{s_0\})\). 
By Lemma \ref{lemma:B}, \(\lbl(s_0)=\lbl(s)\), hence, \(\lbl(s)=\ell_1\).

Let us consider the case \((s_0,s) \notin (\restrict{\concr{A}}{\compl{\after(s_0)}}\cup\Sche)^{\star}\)
Because \((s_0,s)\in\col\), \((s_0,s) \in [(\restrict{\concr{G}}{{\after(s_0)}}\cap\TR{\li1 basic,\ell_2})\restrict{\concr{A}}{\compl{\after(s_0)}}]^{\star}\).
So \((s_0,s) \in (\restrict{\concr{A}}{\compl{\after(s_0)}}\cup\Sche)^{\star};[\restrict{\concr{G}}{{\after(s_0)}}\cap\TR{\li1 basic,\ell_2}\smallsetminus\Sche];[(\restrict{\concr{G}}{{\after(s_0)}}\cap\TR{\li1 basic,\ell_2})\restrict{\concr{A}}{\compl{\after(s_0)}}]^{\star}\).

Let \(s_1, s_2, s_3,\ldots,s_n\) a sequence of states such that \((s_1,s_2)\in(\restrict{\concr{A}}{\compl{\after(s_0)}}\cup\Sche)^{\star} \) and \((s_2,s_3)\in \restrict{\concr{G}}{{\after(s_0)}}\cap\TR{\li1 basic,\ell_2}\smallsetminus\Sche\) and for all \(k\in\{3,\ldots,n\}\), \((s_k,s_{k+1})\in (\restrict{\concr{G}}{{\after(s_0)}}\cap\TR{\li1 basic,\ell_2})\restrict{\concr{A}}{\compl{\after(s_0)}}\).

Notice that \((s_1,s_2)\in \restrict{\concr{G}}{{\after(s_0)}}\) and therefore \(s_1\in\after(s_0)\).
By Lemma \ref{lemma:E'}, \(\thread(s_0)=\thread(s_1)\). Therefore \(s_1 \in \cinter{\concr{A}}(\{s_0\})\).
 
By Lemma \ref{lemma:A'}, \(\lbl(s_2)=\ell_2\).

Let \(k_0\) the smallest (if it exists) \(k\geqslant 2\) such that \((s_{k},s_{k+1})\in \TR{\li1 basic,\ell_2} \smallsetminus\Sche\).
Therefore \((s_2,s_{k_0})\in (\restrict{\concr{A}}{\compl{\after(s_0)}}\cup\Sche)^{\star}\).
By Lemma \ref{lemma:B}, \(\lbl(s_{k_0})=\lbl(s_2)=\ell_2\). According to Lemma \ref{lemma:A'}, this is a contradiction. Therefore, for all \(k\in\{3,\ldots,n\}\), \((s_k,s_{k+1})\in \Sche\cup\restrict{\concr{A}}{\compl{\after(s_0)}}\).

By Lemma \ref{lemma:A}, \(\thread(s_1)=\thread(s_2)\), hence \(\thread(s_2)=\thread(s)\). Therefore \(s_2\in\cinter{\concr{A}}(\{s_2\})\)
\end{proof}

Now, we introduce some claims on the semantics of basic statements.
Claims \ref{claim:SuV} and \ref{claim:SudV} say that when a basic statement is executed, only one thread is executed. Notice that \(\cspawn\) creates a subthread, but does not execute it.
The Claim \ref{claim:SeV}  caracterizes the transitions done by the current thread. The Claim \ref{claim:BS} gives an overapproximation of \(\concr{S'}\), the set of states reached at the end of the execution of a basic statement.

\begin{claim}\label{claim:SuV}
 Let \(\li1 basic,\ell_2\) a basic statement and \(\nq=\ssem{\li1 basic,\ell_2}\qc\).
 Therefore, \(\Su=\emptyset\).
 \end{claim}
\begin{proof}
Let \((s,s')\in\Su\).
Therefore, \((s,s')\in\col;\Sche\langle \concr{S} \rangle\).
So, there exists \(s_0\in\concr{S_0}\) and \(s_1\) such that \((s_0,s_1)\in\col\), \((s_1,s)\in\Sche\) and \(s\in\after(s_0)\).
Hence, by Lemma \ref{lemma:F+}, \(\thread(s)=\thread(s')\).
Given that \((s,s_1)\in\col\), \(\thread(s)=\thread(s_1)\).
But, because \((s_1,s)\in\Sche\), \(\thread(s)\neq\thread(s_1)\).
There is a contradiction. Hence \(\Su=\emptyset\).
\end{proof}

\begin{claim}\label{claim:SudV}
 Let \(\li1 basic,\ell_2\) a basic statement and \(\nq=\ssem{\li1 basic,\ell_2}\qc\).
 Therefore, \(\Sud=\emptyset\).
 \end{claim}
\begin{proof}
 Let \((s,s')\in\Sud\). There exists \(s_0\in\concr{S}\) and \(s_1\) such that \((s_0,s_1)\in\col\), \((s_2,s)\in\ext{s_0}{s_1}\) and \(s\in\after(s_0)\smallsetminus\after(s_1)\).
 
Let \(\cstate{_0}=s_0\) and \((i_1,P_1,\m_1,\h_0\cdot\h_1)=s_1\).
Because \((s_0,s_1)\in\col\), \(\thread(s_0)=\thread(s_1)\).
Let \(j\in\desce_{\h_1}(\{i_0\})\). Let \(s'_1=(j,P_1,\m_1,\h_0\cdot\h_1)\).
Therefore \(s'_1\in \after(s_0)\) and \((s_0,s'_1)\in (\TR{\li1 basic,\ell_2}\cup\restrict{\concr{A}}{\compl{\after(s_0)}})^{\star};(\Sche^{\star})\).
By lemma \ref{lemma:E'}, \(j=\thread(s'_1)=\thread(s_0)=i_0\).
Hence \(\desce_{\h_1}(\{i_0\})=\{i_0\}\).

Let \((i,P,\m,\h_0\cdot\h_1\cdot\h)=s\). By definition of \(\desce\) and a straightforward induction on \(\h\), \(\desce_{\h_1\cdot\h}(\{i_0\})=\desce_{\h}(\{i_0\})\).

Because \(s\in\after(s_0)\), then \(i\in \desce_{\h_1\cdot\h}(\{i_0\})\). Therefore \(i=i_0\). By Lemma \ref{lemma:F+}, \(s\in\after(s_1)\). This is contradictory with \(s\in\after(s_0)\smallsetminus\after(s_1)\).
Hence \(\Sud=\emptyset\).
\end{proof}

\begin{claim}\label{claim:SeV}
 Let \(\li1 basic,\ell_2\) a basic statement and \(\nq=\ssem{\li1 basic,\ell_2}\qc\).\\
 Therefore, \(\Su\subset \{(s,s')\in\TR{\li1 basic,\ell_2}\mid s\in \cinter{\concr{A}}(\concr{S}) \}\cup\Sche\).
 \end{claim}
\begin{proof}
Let \((s,s')\in \Se\smallsetminus\Sche\).
Then \((s,s')\in \TR{\li1 basic,\ell_2} \) and \(s\in\col\langle\concr{S} \rangle \).
Then, there exists \(\sinit\in\concr{S}\) such that \((s_0,s)\in\col\).
 Because \((s_0,s)\in \TR{\li1 basic,\ell_2}\smallsetminus\Sche \), by Lemma \ref{lemma:A'}, \(\lbl(s)=\ell_1\neq\ell_2\).
 By Lemma \ref{lemma:C}, \(s \in \cinter{\concr{A}}(\{s_0\}) \subset \cinter{\concr{A}}(\concr{S})\). 
 Because \(\thread(s_0)=\thread(s)\), \((s,s')\in \Se\).
\end{proof}

\begin{claim}
 \label{claim:BS}
 Let \(\li1 basic,\ell_2\) a basic statement, \(\qcp'=\osem{\li1 basic,\ell_2}\qc\) and \(\nq=\ssem{\li1 basic,\ell_2}\qc\).\\
 Therefore, \(\concr{S}'\subset\cinter{\concr{A}} \big(\TR{\li1 basic,\ell_2}\smallsetminus \Sche \langle  \cinter{\concr{A}}(\concr{S})\rangle\big)\).
\end{claim}
\begin{proof}
Let \(s\in\concr{S}'\). Therefore, \(\lbl(s)=\ell_2\) and there exists \(s_0\in\concr{S}\) such that \((s_0,s)\in\col\).

Because \(\lbl(s)=\ell_2\neq\ell_1\), according to Lemma \ref{lemma:C}, \(s \in\cinter{\concr{A}}(\TR{\li1 basic,\ell_2} \smallsetminus\Sche \langle \cinter{\concr{A}}(\{s_0\}) \rangle) \subset\cinter{\concr{A}}(\TR{\li1 basic,\ell_2} \smallsetminus\Sche \langle \cinter{\concr{A}}(\concr{S}) \rangle) \)
\end{proof}

\n{\Gnew}{\concr{G}_{\text{new}}}

\begin{proposition}[Basic statements]\label{prop:basic}
 Let \(\li1 basic,\ell_2\) be a basic statement, then:
  \[\osem{\li1 basic,\ell_2}\qc \leqslant \langle \concr{S''}, \concr{G}\cup\Gnew ,\concr{A} \rangle\]
 where \(\concr{S}'' = \cinter{\concr{A}} \big(\TR{\li1 basic,\ell_2}\smallsetminus \Sche \langle  \cinter{\concr{A}}(\concr{S})\rangle\big)\)
\\ and \(\Gnew = \{(s,s')\in\TR{\li1 basic,\ell_2}\mid s\in \cinter{\concr{A}}(\concr{S})  \} \)
\end{proposition}

\begin{proof}
This proposition is a straightforward consequence of Claims \ref{claim:SuV}, \ref{claim:SudV}, \ref{claim:SeV} and \ref{claim:BS}.
\end{proof}

\subsection{Overapproximation of the \Cname Semantics}

The next theorem shows how the \cname semantics can be over-approximated by a denotational semantics, and is the key point in defining the abstract semantics.

\begin{theorem}\label{theorem:denot}
\begin{enumerate}[1.\!]
 \item\label{fcomposition}%
\(\av{\li1\cmd_1;\li2\cmd_2,\ell_3}(\concr{Q}) \leqslant \av{\li2\cmd_2,\ell_3}\circ\av{\li1\cmd_1,\ell_2}(\concr{Q})\) 
\item\label{fif}\(\av{\li1\cifte{(\cond}{\{\li2\cmd_1\}}\{\li4\cmd_2\},\ell_3}(\concr{Q}) 
\leqslant \\
\av{\li2\cmd_1,\ell_3}\circ\av{\li1\cguard(\cond),\ell_2}(\concr{Q}) 
\sqcup
\av{\li4\cmd_2,\ell_3}\circ\av{\li1\cguard(\neg\cond),\ell_4}(\concr{Q}) \)
\item\label{fwhile}%
 \(\av{\li1 \cwhile(\guardp)\{\li2\cmd\},\ell_3 }(\concr{Q})\leqslant \av{\li1 \cguard(\guardn),\ell_3} \circ \loopw\eomega (\concr{Q})\)\\
with \(\loopw(\concr{Q}') = \big( \av{\li2\cmd,\ell_1} \circ 
\av{\li1  \cguard(\guardp),\ell_2} (\concr{Q}') \big) \sqcup \concr{Q'}\)
\item\label{fcreate}%
\(\av{\li1\ccreate(\li2\cmd),\ell_3}(\concr{Q}) \leqslant  
\ccomb_{\concr{Q}'}\circ\fp{\av{\li2\cmd,\lend}}\circ \cinit_{\ell_2} (\concr{Q}')\)\\
with  \(\concr{Q}' = \av{\li1\crcreate(\ell_2),\ell_3}(\concr{Q})\)
\end{enumerate}
\end{theorem}

While points \ref{fcomposition} and \ref{fwhile} are as expected, the overapproximation of semantics of \(\li1\ccreate(\li2\cmd),\ell_3\) (point \ref{fcreate}) computes interferences which will arise from executing the child and its descendants with \(\fp{}\) and then combines this result with the configuration of the current thread. 
This theorem will be proved later.

The following proposition consider a statement \(\lab\stmt,\ell'\) set of transition \(T\). The only constraint on \(T\) is on the use of labels of \(\lab\stmt,\ell'\). 

The proposition consider an execution of the statement from a state \(s_0\) to a state \(s_1\), and, after, an execution \(s_2,\ldots,s_n\) of other commands.
The labels of \(\lab\stmt,\ell'\) mays only be used :
\begin{itemize}
 \item for interferences,
 \item or by the statement,
 \item after having applied the statement, i.e., after \(s_1\).
\end{itemize}.
This Proposition ensures us that any transition executed by a thread created during the execution of \(\lab\stmt,\ell'\) (i.e., between \(s_0\) and \(s_1\)) is a transition generated by the statement \(\lab\stmt,\ell'\).

\begin{proposition}\label{prop:apres1}
 Let \(\lab\stmt,\ell'\) a statement, \\\(\nq=\ssem{\lab\stmt,\ell'}\qc\).
  Let \((s_0,s_1)\in\col\) and \(T\) a set of transitions such that for all \((s,s')\in T\), if \(\lbl(s)\in\Labs{\lab\stmt,\ell'}\) then \((s,s')\in \TR{\lab\stmt,\ell'}\) or \(s\in\after(s_1)\cup\compl{\after(s_0)}\). 
 
 Let \(s_2,\ldots,s_n\) a sequence of states such that for all \(k\in\{1,\ldots,n-1\}\), \((s_{k},s_{k+1})\in T\).
 Therefore, if \(s_k\in\after(s_0)\) then either \(s_k\in\after(s_1)\) or \((s_{k},s_{k+1})\in \TR{\lab\stmt,\ell'}\) 
\end{proposition}%
\begin{proof}%
Let 
for all \(k\geqslant 1\), let \((i_k,P_k,\m_k,\h_0\cdot\h_k)=s_k\).

 Let us show by induction on \(k\geqslant 1\) that for all \(j\), if \(j\in\desce_{\h_0\cdot\h_k}(\{i_1\})\smallsetminus \desce_{\h_k}(\{i_1\})\) then \(P_k(j)\in\Labs{\lab\stmt,\ell'}\).
 
Let \(j_0\in\desce_{\h_0\cdot\h_k}(\{i_1\})\smallsetminus \desce_{\h_0}(\{i_1\})\) and \(s'_1=(j_0,P_1,\m_1,\h_0\cdot\h_1)\). 
Therefore \(s'_1\in \after(s_0)\). 
Given that\((s_0,s'_1)\in \col;\Sche\),  by Lemma \ref{lemma:R}, \(P_1(j_1)=\lbl(s'_1)\in \Labs{\lab\stmt,\ell'}\). 
 
By induction hypothesis, for all \(j\), if \(j\in\desce_{\h_0\cdot\h_{k-1}}(\{i_1\})\smallsetminus \desce_{\h_{k-1}}(\{i_1\})\) then \(P_{k-1}(j)\in\Labs{\lab\stmt,\ell'}\).

 Let \(j\in \desce_{\h_0\cdot\h_k}(\{i_1\})\smallsetminus \desce_{\h_k}(\{i_1\})\).
 
 If \(\thread(s_{k-1})=j\), therefore, \(s_{k-1}\in\after(s_0)\smallsetminus\after(s_1)\). Furthermore, by induction hypothesis, \(P_{k-1}(j)=\lbl(s_{k-1})\in \Labs{\lab\stmt,\ell'} \). By definition of \(T\), \((s_{k-1},s_k)\in \TR{\lab\stmt,\ell'}\). By Lemma \ref{lemma:A}, \(P_{k}(j)=\lbl(s_k)\in \Labs{\lab\stmt,\ell'}\).
 
 If \(j\in\dom(P_{k})\smallsetminus\dom(P_{k-1})\), then, \(\thread(s_{k-1})\in \desce_{\h_0\cdot\h_{k-1}}(\{i_1\})\smallsetminus \desce_{\h_{k-1}}(\{i_1\})\). Hence, as above, \((s_{k-1},s_k)\in \TR{\lab\stmt,\ell'}\). Hence, according to Lemma \ref{lemma:Asub}, \(P_{k}(j)=\lbl(s_k)\in \Labs{\lab\stmt,\ell'}\).
 
 Else, by definition of a transition, \(P_{k-1}(j)=P_{k}(j)\).
 
 Let \(k\) such that \(s_k\in\after(s_0)\), hence, either \(s_k\in\after(s_1)\), or \(s_k\notin\after(s_1)\). In the last case \(i_k\in \desce_{\h_0\cdot\h_{k-1}}(\{i_1\})\smallsetminus \desce_{\h_{k-1}}(\{i_1\}) \), and therefore \(\lbl(s_k)\in \Labs{\lab\stmt,\ell'} \). Hence, by definition of \(T\), \((s_k,s_{k+1})\in\TR{\lab\stmt,\ell'}\).
\end{proof}

\subsubsection{Proof of Property \ref{fcomposition} of Theorem \ref{theorem:denot}} 

\begin{lemma}
 \label{lemma:G}
 \(\TR{\li1\cmd_1;\li2\cmd_2,\ell_3}=\TR{\li1\cmd_1,\ell_2}\cup\TR{\li2\cmd_2,\ell_3}\)
\end{lemma}

In this section, we consider an initial configuration : \(\concr{Q}_0 =\qcp{_0}\) and a sequence \(\li1\cmd_1;\li2\cmd_2,\ell_3\).
We write \(\Tr_1=\TR{\li1\cmd_1,\ell_2}\) and \(\Tr_2=\TR{\li2\cmd_2,\ell_3}\) and \(\Tr=\TR{\li1\cmd_1;\li2\cmd_2,\ell_3}\)

Define:\\
 \(\concr{Q}' =\qcp{'}=\av{\li1\cmd_1;\li2\cmd_2,\ell_3}(\concr{Q}_0)\)\\ \(\concr{K} =\nqp{}=\ssem{\li1\cmd_1;\li2\cmd_2,\ell_3}(\concr{Q}_0)\)\\
 \(\concr{Q}_1 =\qcp{_1}=\av{\li1\cmd_1,\ell_2}(\concr{Q}_0)\)\\ \(\concr{K}_1 =\nqp{_1}=\ssem{\li1\cmd_1,\ell_2}(\concr{Q}_0)\)\\
 \(\concr{Q}_2 =\qcp{_2}=\av{\li2\cmd_2,\ell_3}(\concr{Q}_1)\)\\
 \(\concr{K}_2 =\nqp{_2}=\ssem{\li2\cmd_2,\ell_3}(\concr{Q}_1)\)

\begin{lemma}
 \label{lemma:G+}
 If \((s,s')\in\Tr\) and \(\lbl(s)\in \Labs{\li1\cmd_1,\ell_2}\smallsetminus\{\ell_2\}\) then \((s,s')\in\Tr_1\).
 
 If \((s,s')\in\Tr\) and \(\lbl(s)\in \Labs{\li2\cmd_2,\ell_3}\) then \((s,s')\in\Tr_2\).
\end{lemma}
\begin{proof}
Let us consider that \(\lbl(s)\in \Labs{\li1\cmd_1,\ell_2}\smallsetminus\{\ell_2\}\). Hence because labels of \(\li1cmd_1;\li2\cmd_2,\ell_3\) are pairwise distinct, \(\lbl(s)\notin\Labs{\li2\cmd_3,\ell_3}\).
 By Lemma \ref{lemma:Abis}, \((s,s')\notin\Tr_2\). Hence, by Lemma \ref{lemma:G}, \((s,s')\notin\Tr_1\)
 
 The case \(\lbl(s)\in \Labs{\li2\cmd_2,\ell_3}\) is similar. 
\end{proof}

\begin{lemma}\label{lemma:Cbis}
Using the above notations, for every \((s_0,s)\in\col\) such that \(s_0\in\concr{S}_0\),
 \begin{itemize}
  \item either \((s_0,s)\in\col_1\) and \(\lbl(s)\neq\ell_2\)
  \item or there exists \(s_1\in\concr{S}_1\) such that \((s_0,s_1)\in\col_1\), \((s_1,s)\in\col_2\)
 \end{itemize}
 \end{lemma}
 
 \begin{proof}
 Let \((s_0,s)\in\col\).
 Either \((s_0,s)\in\col_1\) or \((s_0,s)\notin\col_1\).
 
 In the first case, either \(\lbl(s)\neq\ell_2\), or \(\lbl(s)=\ell_2\). If \(\lbl(s)=\ell_2\), then, by definition, \(s\in\concr{S}_1\). By definition, \((s,s)\in\col_2\) and \((s,s)\in\ext[_1]{s_0}{s} \). We just have to choose \(s_1=s\).
 
 In the second case, \((s_0,s)\notin\col_1\). Let \(T_0 = (\restrict{\concr{G_0}}{\after(s_0)}\cap\Tr_1)\cup\restrict{\concr{A_0}}{\compl{\after(s_0)}}\).
 Since \((s,s')\in\col'\) , \(\thread(s_0)=\thread(s)\) and \(\lbl(s_0)=\ell_1\). Furthermore \((s_0,s)\notin \col_1\), so \((s_0,s)\notin T_0^{\star} \).
 Since \((s,s')\in\col'\subset [(\restrict{\concr{G_0}}{\after(s_0)}\cap\Tr)\cup\restrict{\concr{A_0}}{\compl{\after(s_0)}}]^{\star}  \), \(\Tr=\Tr_1\cup\Tr_2\) (using Lemma \ref{lemma:G})  and \(\Tr_1\supset\Sche\) , therefore \((s_0,s) \in [ T_0 \cup  (\restrict{\concr{G_0}}{\after(s_0)}\cap\Tr_2\smallsetminus \Sche)]^{\star}\).

 Recall \((s_0,s)\notin T^{\star}\), hence \((s_0,s)\in T_0^{\star};(\restrict{\concr{G_0}}{\after(s_0)}\cap\Tr_2\smallsetminus \Sche); [ T_0 \cup  (\restrict{\concr{G_0}}{\after(s_0)}\cap\Tr_2)]^{\star}\).
 Therefore, there exists \(s_1\), \(s_2\) such that:
 \begin{itemize}
  \item \((s_0,s_1)\in T_0^{\star}\)
  \item \((s_1,s_2)\in\restrict{\concr{G_0}}{\after(s_0)}\cap\Tr_2\smallsetminus \Sche \)
  \item \((s_2,s)\in [T_0 \cup  (\restrict{\concr{G_0}}{\after(s_0)}\cap\Tr_2)]^{\star}\)
 \end{itemize}
 
 Since \(s_0\in\concr{S}_0\), \(\lbl(s_0)=\ell_1 \in\Labs{\li1\cmd_1,\ell_2} \).
Since \((s_1,s_2)\in\restrict{\concr{G_0}}{\after(s_0)}\), \(s_1\in\after(s_0)\). Furthemore \((s_0,s_1)\in T_0^{\star} \subset \Tr_1\cup\restrict{\concr{A_0}}{\compl{\after(s_0)}} \), so, according to Lemma \ref{lemma:F}, \(\lbl(s_1)\in\Labs{\li1\cmd_1,\ell_2}\).

Given that \((s_1,s_2)\in\Tr_2\smallsetminus\Sche\), according to Lemma \ref{lemma:Abis}, \(\lbl(s_1)\in\Labs{\li2\cmd_2,\ell_3}\).
Hence \(\lbl(s_1)\in\Labs{\li2\cmd_2,\ell_3}\cap \Labs{\li1\cmd_1,\ell_2}\). 
Because the labels of \(\li1\cmd_1;\li2\cmd_2,\ell_3\) are pairwise distincts, \(\lbl(s_1)=\ell_2 \). Using Lemma \ref{lemma:F}, we conclude that \(\thread(s_0)=\thread(s_1)\).

Given that \(\thread(s_0)=\thread(s)\) and \(\lbl(s_0)=\ell_1\) and \( (s_0,s_1)\in T_0^{\star} \), we conclude that \((s_0,s_1)\in\col_1\).
Furthermore \(\lbl(s_1)=\ell_2\) and \(s_0\in\concr{S_0}\), therefore \(s_1\in\concr{S}_1\).

\((s_1,s)\in [T_0 \cup  (\restrict{\concr{G_0}}{\after(s_0)}\cap\Tr_2)]^{\star}\). Therefore, by proposition \ref{prop:apres1}, \((s_1,s)\in [T_0 \cup  (\restrict{\concr{G_0}}{\after(s_1)}\cap\Tr_2)]^{\star}\subset \ext[_1]{s_0}{s_1}\).
 
Recall that 
\((s_2,s)\in [T_0 \cup  (\restrict{\concr{G_0}}{\after(s_0)}\cap\Tr_2)]^{\star}\), then there exists \(s_3,\ldots,s_n\) such that for all \(k\in\{3,\ldots,n-1\}\), \((s_k,s_{k+1})\in T_0 \cup  (\restrict{\concr{G_0}}{\after(s_0)}\cap\Tr_2)\).
By definition, if \((s_k,s_{k+1})\in \restrict{\concr{G_0}}{\after(s_0)}\cap\Tr_1\), then \((s_k,s_{k+1})\in\Sud_1\).
 
 We show by induction on \(k\) that if \((s_k,s_{k+1})\in \restrict{\concr{G_0}}{\after(s_0)}\cap\Tr_1\smallsetminus\Sche\), then \(s_k\notin\after(s_1)\).
 By induction hypothesis, \((s_2,s_k)\in\restrict{(\restrict{\concr{G_0}}{\after(s_0)}\cap\Tr_1)}{\compl{\after(s_1)}}\cup\restrict{\concr{A_0}}{\compl{\after(s_0)}}\cup (\restrict{\concr{G_0}}{\after(s_0)}\cap\Tr_2)]^{\star}\).
 Therefore, by Lemma \ref{lemma:F}, if \(s_k\in\after(s_2)\), then \(\lbl(s_k)\in\Labs{\li2cmd_2,\ell_3}\). Therefore, because labels are pairwise distinct, if \(s_k\in\after(s_2)\), then \(\lbl(s_k)\notin\Labs{\li1cmd_1,\ell_2}\smallsetminus\{\ell_2\}\). Therefore, by Lemma \ref{lemma:Abis}, if \(s_k\in\after(s_2)\), then \((s_k,s_{k+1})\notin\Tr_1\).
 
 Hence, \((s_1,s)\in [\restrict{{\Sud_1}}{\compl{\after(s_1)}}\cup\restrict{\concr{A_0}}{\compl{\after(s_0)}}\cup (\restrict{\concr{G_0}}{\after(s_0)}\cap\Tr_2)]^{\star}\).
 By Lemma \ref{lemma:F+}, \(\after(s_1)\subset\after(s_0)\), hence \((s_1,s)\in [\restrict{(\Sud_1\cup\concr{A_0})}{\compl{\after(s_0)}}\cup (\restrict{\concr{G_0}}{\after(s_0)}\cap\Tr_2)]^{\star}\subset[\restrict{\concr{{A_1}}}{\compl{\after(s_0)}}\cup (\restrict{\concr{G_0}}{\after(s_0)}\cap\Tr_2)]^{\star} \).
 Therefore \((s_1,s)\in\col_2\).
\end{proof}
 
 \begin{lemma}\label{lemma:Cter}
Using the above notations, for every \((s_0,s)\in\col\) such that \(s_0\in\concr{S}_0\) and \(s'\in\concr{S'}\), there exists \(s_1\in\concr{S}_1\) such that \((s_0,s_1)\in\col_1\), \((s_1,s)\in\col_2\) and \((s_1,s)\in \ext[_1]{s_0}{s_1}\).
 \end{lemma}
 
\begin{proof}
If \((s_0,s)\in\col_1\), then, according to Lemma \ref{lemma:R}, \(\lbl(s)\in\Labs{\li1cmd_1,\ell_2}\). In this case \(\lbl(s)\neq\ell_3\). This is not possible because \(s\in\concr{S'}\).
 
 Therefore, according to Lemma \ref{lemma:Cbis} there exists \(s_1\in\concr{S}_1\) such that \((s_0,s_1)\in\col_1\), \((s_1,s)\in\col_2\) and \((s_1,s)\in \ext[_1]{s_0}{s_1}\)
\end{proof}

\begin{lemma}\label{lemma:Ext1}
 Using the notations of this section, let \(s_0\in\concr{S_0},s_1\in\concr{S_1},s_2\in\concr{S}_2,s\) such that \((s_0,s_1)\in\col_1\), \((s_1,s_2)\in\col_2\cap\ext[_1]{s_0}{s_1}\) and \((s_2,s)\in\ext{s_0}{s_2}\).
 Therefore \((s_1,s)\in \ext[_1]{s_0}{s_1}\).
\end{lemma}
\begin{proof}
Notice that, by Lemma \ref{lemma:F+}, \(\after(s_2)\subset\after(s_1)\subset\after(s_0)\).

Recall that:

 \(\ext{s_0}{s_2}=\big[
(\restrict{\concr{G_0}}{\after(s_0)}\cap\Tr)\cup \restrict{\concr{A_0}}{\compl{\after(s_0)}} \cup \restrict{\concr{G_0}}{\after(s_2)} 
\big]^{\star}\)

 \(\ext[_1]{s_0}{s_1}=\big[
(\restrict{\concr{G_0}}{\after(s_0)}\cap\Tr_1)\cup \restrict{\concr{A_0}}{\compl{\after(s_0)}} \cup \restrict{\concr{G_0}}{\after(s_1)} 
\big]^{\star}\)

By Lemma \ref{lemma:G},  \(\ext{s_0}{s_2}=\big[
(\restrict{\concr{G_0}}{\after(s_0)}\cap\Tr_1)\cup(\restrict{\concr{G_0}}{\after(s_0)}\cap\Tr_2)\cup \restrict{\concr{A_0}}{\compl{\after(s_0)}} \cup \restrict{\concr{G_0}}{\after(s_2)} 
\big]^{\star}\). Let \(T=(\restrict{\concr{G_0}}{\after(s_0)}\cap\Tr_2)\cup\restrict{\concr{G_0}}{\after(s_2)}\).
Therefore, because \(\after(s_2)\subset\after(s_0)\), \(\ext{s_0}{s_2}=\big[
(\restrict{\concr{G_0}}{\after(s_0)}\cap\Tr_1)\cup \restrict{\concr{A_0}}{\compl{\after(s_0)}} \cup \restrict{T}{\after(s_0)}\big]^{\star}\).

By Proposition \ref{prop:apres1}, \((s_2,s)\in\big[
(\restrict{\concr{G_0}}{\after(s_0)}\cap\Tr_1)\cup \restrict{\concr{A_0}}{\compl{\after(s_0)}} \cup \restrict{T}{\after(s_1)}\big]^{\star}\). Because \(\after(s_2)\subset\after(s_1)\subset\after(s_0)\), \(\restrict{T}{\after(s_1)}  = (\restrict{\concr{G_0}}{\after(s_1)}\cap\Tr_2)\cup\restrict{\concr{G_0}}{\after(s_2)}\).
Hence \((s_2,s) \in \ext[_1]{s_0}{s_1}\).
Hence \((s_1,s)\in\ext[_1]{s_0}{s_1};\ext[_1]{s_0}{s_1}=\ext[_1]{s_0}{s_1} \).
\end{proof}

\begin{lemma}\label{lemma:Ext2}
 Using the notations of this section, let \(s_0\in\concr{S_0},s_1\in\concr{S_1},s_2\in\concr{S}_2,s\) such that \((s_0,s_1)\in\col_1\), \((s_1,s_2)\in\col_2\cap\ext[_1]{s_0}{s_1}\) and \((s_2,s)\in\ext{s_0}{s_2}\).
 Therefore \((s_2,s)\in\ext[_2]{s_1}{s_2}\).
\end{lemma}
\begin{proof}
Notice that, by Lemma \ref{lemma:F+}, \(\after(s_2)\subset\after(s_1)\subset\after(s_0)\).

Recall that 

 \(\ext{s_0}{s_2}=\big[
(\restrict{\concr{G_0}}{\after(s_0)}\cap\Tr)\cup \restrict{\concr{A_0}}{\compl{\after(s_0)}} \cup \restrict{\concr{G_0}}{\after(s_2)} 
\big]^{\star}\)

 \(\ext[_2]{s_1}{s_2}=\big[
(\restrict{\concr{G_1}}{\after(s_1)}\cap\Tr_2)\cup \restrict{\concr{A_1}}{\compl{\after(s_1)}} \cup \restrict{\concr{G_1}}{\after(s_2)} 
\big]^{\star}\)

Since \((s_2,s)\in\ext{s_0}{s_2}\), \(\concr{A_0}\subset\concr{A_1}\), \(\concr{G_0}\subset\concr{A_1}\), and \(\after(s_1)\subset\after(s_0)\)
there exists \(s_3,\ldots,s_n\) such that \(s_n=s\) and for all \(k\in  \{3,\ldots,n-1\}\), 
\((s_k,s_{k+1})\in (\restrict{\concr{G_1}}{\after(s_0)}\cap\Tr) \cup \restrict{\concr{A_1}}{\compl{\after(s_1)}} \cup \restrict{\concr{G_1}}{\after(s_2)}\).

Due to Lemma \ref{lemma:G}, for all \(k\in  \{3,\ldots,n-1\}\), 
\((s_k,s_{k+1})\in (\restrict{\concr{G_1}}{\after(s_0)}\cap\Tr_1) \cup (\restrict{\concr{G_1}}{\after(s_0)}\cap\Tr_2)\cup \restrict{\concr{A_1}}{\compl{\after(s_1)}} \cup \restrict{\concr{G_1}}{\after(s_2)}\).

Because \((s_1,s_2)\in\col_2\), \((s_1,s_2)\in\big[ (\restrict{\concr{G_1}}{\after(s_1)}\cap\Tr_2) \restrict{\concr{A_1}}{\compl{\after(s_1)}} \big]^{\star}
\subset
\big[(\restrict{\concr{G_1}}{\after(s_0)}\cap\Tr_2) \cup (\restrict{\concr{G_1}}{\after(s_0)}\cap\Tr_2)\cup \restrict{\concr{A_1}}{\compl{\after(s_1)}} \cup \restrict{\concr{G_1}}{\after(s_2)}\big]^{\star}\).

Hence, by Proposition \ref{prop:apres1} applied on the statement \(\li1\cmd_1,\ell_2\), for all \(k\in  \{3,\ldots,n-1\}\), 
\((s_k,s_{k+1})\in (\restrict{\concr{G_1}}{\after(s_0)}\cap\Tr_1)\cup (\restrict{\concr{G_1}}{\after(s_1)}\cap\Tr_2) \cup \restrict{\concr{A_1}}{\compl{\after(s_1)}} \cup \restrict{\concr{G_1}}{\after(s_2)} \).

Given that \((\restrict{\concr{G_1}}{\after(s_0)}\cap\Tr_1)=(\restrict{\concr{G_1}}{\after(s_0)\smallsetminus\after(s_0)}\cap\Tr_1)\cup(\restrict{\concr{G_1}}{\after(s_1)}\cap\Tr_1)\) and \(\restrict{\concr{G_1}}{\after(s_2)}\cap\Tr_1 \subset \restrict{\concr{G_1}}{\after(s_2)}\), by Proposition \ref{prop:apres1} applied on the statement \(\li2\cmd_2,\ell_3\), we conclude that for all \(k\in  \{3,\ldots,n-1\}\), 
\((s_k,s_{k+1})\in (\restrict{\concr{G_1}}{\after(s_0)\smallsetminus\after(s_1)}\cap\Tr_1)\cup (\restrict{\concr{G_1}}{\after(s_1)}\cap\Tr_2) \cup \restrict{\concr{A_1}}{\compl{\after(s_1)}} \cup \restrict{\concr{G_1}}{\after(s_2)} \).
Let \(k_0\) such that \((s_{k_0},s_{k_0+1})\in (\restrict{\concr{G_1}}{\after(s_0)\smallsetminus\after(s_1)}\cap\Tr_1)\smallsetminus \restrict{\concr{G_1}}{\after(s_2)}\). By Lemma \ref{lemma:Ext1}, \((s_1,s_{k_0})\in\ext[_1]{s_0}{s_1}\).
Therefore \((s_{k_0},s_{k_0+1})\in\Sud_1\).

Hence \((s_2,s)\in \big[
\restrict{{\Sud_1}}{\after(s_0)\smallsetminus\after(s_1)}\cup (\restrict{\concr{G_1}}{\after(s_1)}\cap\Tr_2) \cup \restrict{\concr{A_1}}{\compl{\after(s_0)}} \cup \restrict{\concr{G_1}}{\after(s_2)}
\big]^{\star}\).
Because \(\restrict{{\Sud_1}}{\after(s_0)\smallsetminus\after(s_1)}\subset\restrict{\concr{A}}{\compl{\after{s_1}}}\), we conclude that \((s_2,s)\in\ext[_2]{s_1}{s_2}\).
\end{proof}

To prove the Property \ref{fcomposition} of Theorem \ref{theorem:denot}, we have to prove that \(\concr{Q}_2\geqslant \concr{Q}'\). We claim that 
\begin{inparaenum}[(a)]
 \item \(\concr{S}'\subset \concr{S}_2\)
 \item \(\Se'\subset \Se_1\cup\Se_2\)
 \item \(\Su'\subset\Su_1\cup\Su_2\cup \Sud_1 \)
 \item \(\Sud'\subset \Sud_1\cup\Sud_2 \)
\end{inparaenum}.
Using this claims and the definition of the semantics \(\osem{\cdot}\), we conclude that \(\concr{Q}_2\geqslant \concr{Q}'\).

Now, we prove these claims:

\begin{claim}\label{lemma:seq:s}
 Using the notations of this section, \(\concr{S}'\subset \concr{S}_2\).
\end{claim}

\begin{proof}
 Let \(s\in\concr{S}'\),
so there exists \(s_0\in\concr{S}\) such that \((s_0,s)\in\col'\) and \(\lbl(s)=\ell_3\).
According to Lemma \ref{lemma:Cter} there exists \(s_1\in\concr{S}_1\) such that \((s_1,s)\in\col_2\). Therefore \(s\in\concr{S}_2\).
\end{proof}

\begin{claim}\label{lemma:seq:par}
  Using the notations of this section, \(\Se'\subset \Se_1\cup\Se_2\).
\end{claim}

\begin{proof}
 Let \((s,s')\in\Se'\). So \((s,s')\in\Tr\), and there exists \(s_0\in\concr{S}\) such that \((s_0,s)\in\col'\).

According to Lemma \ref{lemma:Cbis} either \((s_0,s)\in\col_1\) and \(\lbl(s)\neq \ell_2\), or there exists \(s_1\in\concr{S}_1\) such that \((s_0,s_1)\in\col_1\) and \((s_1,s)\in\col_2\).

In the first case, according to Lemma \ref{lemma:R}, \(\lbl(s)\in\Labs{\li1\cmd_1,\ell_2}\). 
Since \(\lbl(s)\neq \ell_2\) and by Lemma \ref{lemma:G+}, \((s,s')\in\Tr_1\).
Hence, by definition, \((s,s')\in\Se_1\)

In the second case, by Lemma \ref{lemma:F}, \(\lbl(s')\in\Labs{\li2\cmd_2,\ell_3}\).
Since \((s,s')\in\Tr\), by Lemma \ref{lemma:G+} \((s,s')\in\Tr_2\).
Given that \(s\in\col\langle\concr{S}_1\rangle\) and \((s,s')\in\Tr_2\), we conclude that \((s,s')\in\Se_2\).
\end{proof}

\begin{claim}\label{lemma:seq:par2}
 Using the notations of this section \(\Su'\subset\Su_1\cup\Su_2\cup \Sud_1 \).
\end{claim}%
\begin{proof}%
 Let \((s,s')\in\Su'\).
Therefore, \((s,s')\in\Tr\) and there exists \(s_0\in\concr{S_0}\) and \(s_2\) such that \((s_0,s_2)\in\col'\), \((s_2,s)\in\Sche\) and \(s\in\after(s_0)\).
According to Lemma \ref{lemma:Cbis} there are two cases:

First case: \((s_0,s_2)\in\col_1\) and \(\lbl(s_2)\neq\ell_2\).
Then, using the fact that \(\Sche\subset\Tr_1\), \((s_0,s)\in(\Tr_1\cup\restrict{{\concr{A}_0}}{\compl{\after(s_0)}})^{\star}\).
Because \(s\in\after(s_0)\), by Lemma \ref{lemma:F}, \(\lbl(s)\in\Labs{\li1\cmd_1,\ell_2}\smallsetminus \{\ell_2\}\). 
Hence, according to Lemma \ref{lemma:G+}, \((s,s')\in\Tr_1\).
 We conclude that \((s,s')\in\Su_1\).

Second case:
There exists \(s_1\in\concr{S}_1\) such that \((s_0,s_1)\in\col_1\), \((s_1,s_2)\in\col_2\) and
\((s_1,s_2)\in\ext[_1]{s_0}{s_1}\). 
Hence \((s_1,s)\in\ext[_1]{s_0}{s_1};\Sche=\ext[_1]{s_0}{s_1}\).

If \(s\in\after(s_1)\), then, because \((s_1,s)\in\col_2;\Sche\), by Lemma \ref{lemma:F}, \(\lbl(s)\in\Labs{\li2\cmd_2,\ell_3}\). So, in this case, by Lemma \ref{lemma:G+}, \((s,s')\in\Tr_2\) and then  \((s,s')\in\Su_2\).

Let us consider the case \(s\notin\after(s_1)\).
Given that \((s_0,s_1)\in\col\), \((s_1,s)\in\ext[_1]{s_1}{s_2}\), so by Proposition \ref{prop:apres1}, \((s,s')\in\Tr_1\).
Hence, \((s,s')\in\Sud_1\).
 \end{proof}

\begin{claim}\label{lemma:seq:sub}
 Using the notations of this section \(\Sud'\subset \Sud_1\cup\Sud_2 \).
\end{claim}
\begin{proof}
 Let \((s,s')\in\Sud'\). Then, there exists \(s_0\) and \(s_2\) such that 
\((s_0,s_2)\in\col'\) and 
\((s_2,s)\in \ext{s_0}{s_2}\).
According to Lemma \ref{lemma:Cter}, there exists \(s_1\in\concr{S}_1\) such that
\((s_0,s_1)\in\col_1\) and \((s_1,s_2)\in\col_2\) and \((s_1,s_2)\in \ext[_1]{s_0}{s_1}\).

By Lemma \ref{lemma:Ext1} and Lemma \ref{lemma:Ext2}, \((s_1,s)\in\ext[_1]{s_0}{s_1}\) and \((s_2,s)\in\ext[_2]{s_1}{s_2}\).

Let us consider the case \(s\notin\after(s_1)\).
Because \(s\in\after(s_0)\), then \(s\in\after(s_0)\smallsetminus\after(s_1)\).
Furthermore, given that \((s_0,s_1)\in\col_1\) and \((s_1,s)\in\col_2\), by Proposition \ref{prop:apres1}, \((s,s')\in\Tr_1\).
We conclude that \((s,s')\in\Sud_1\).

Let us consider the case \(s\in\after(s_1)\). Because \(s\in\after(s_0)\smallsetminus\after(s_2)\), \(s\in\after(s_1)\smallsetminus\after(s_2)\).
By Lemma \ref{lemma:F}, \(\lbl(s)\in\Labs{\li2\cmd_2,\ell_2}\). Hence, by Lemma \ref{lemma:G+}, \((s,s')\in\Tr_2\) and therefore, \((s,s')\in\Sud_2\).
\end{proof}

\subsubsection{Proof of Property \ref{fif} of Theorem \ref{theorem:denot}} 
In this section, we consider a command \(\li1 \cifte{\guardp}{\{\li2\cmd_1\}}{\{\li3\cmd_2\}},\ell_4\) and an initial configuration \(\concr{Q}_0=\qcp{_0}\)\\

Let \(\qcp'=\osem{\li1 \cifte{\guardp}{\{\li2\cmd\}}{\{\li3\cmd\}},\ell_4}\qc\).\\
Let \(\qcp{_{{+}}}=\osem{\li1 \cguard{\guardp},\ell_2}\qc\).\\
Let \(\qcp{_1}=\osem{\li2 \cmd_1,\ell_4}\qcp{_{{+}}}\).\\
Let \(\qcp{_{\neg}}=\osem{\li1 \cguard{\guardn},\ell_3}\qc\).\\
Let \(\qcp{_2}=\osem{\li3 \cmd_1,\ell_4}\qcp{_{\neg}}\).\\
Let \(\Tr=\TR{\li1 \cifte{\guardp}{\{\li2\cmd\}}{\{\li3\cmd\}},\ell_4}\).

\begin{lemma}\label{lemma:Gif}
 \(\TR{\li1 \cifte{\guardp}{\{\li2\cmd\}}{\{\li3\cmd\}},\ell_4} =
     \TR{\li1 \cguard{\guardp},\ell_2}
 \cup\TR{\li2 \cmd_1,\ell_4}
 \cup\TR{\li1 \cguard{\guardn},\ell_3}
 \cup\TR{\li3 \cmd_1,\ell_4}\). 
\end{lemma}

\begin{lemma}\label{lemma:Cif}
 If \((s_0,s)\in\col\) and \(s_0\in\concr{S}_0\), then, one of the three folowing properties hold:
 \begin{enumerate}
  \item \(s\in\cinter{\concr{A_0}}(\{s_0\})\), 
  \item or there exists \(s_1\in\concr{S_{{+}}}\) such that \((s_1,s)\in\col_{1}\cap\ext[_{{+}}]{s_0}{s_1}\)
  \item or there exists \(s_1\in\concr{S_{{\neg}}}\) such that \((s_1,s)\in\col_{2}\cap\ext[_{{\neg}}]{s_0}{s_1}\)
 \end{enumerate}
\end{lemma}
\begin{proof}
Let us consider the case \(s\notin\cinter{\concr{A_0}}(\{s_0\})\).
Because \((s_0,s)\in\col\), \((s_0,s)\in[(\restrict{\concr{G_0}}{\after(s_0)}\cap\Tr)\cup\restrict{\concr{A_0}}{\compl{\after(s_0)}}]^{\star}\).

Therefore, there exists \(s'_0\) and \(s_1\) such that \((s_0,s'_0)\in( \restrict{\concr{A_0}}{\compl{\after(s_0)}} \cup\Sche)^{\star} \), \((s'_0,s_1)\in \restrict{\concr{G_0}}{\after(s_0)}\cap\Tr\) and \((s_1,s)\in [(\restrict{\concr{G_0}}{\after(s_0)}\cap\Tr)\cup\restrict{\concr{A_0}}{\compl{\after(s_0)}}]^{\star} \).
Because \((s'_0,s_1)\in \restrict{\concr{G_0}}{\after(s_0)}\cap\Tr\), \(s\in\after(s_0)\). By Lemma \ref{lemma:E'}, \(\thread(s_0)=\thread(s'_0)\). 
By Lemma \ref{lemma:B}, \(\lbl(s_0)=\lbl(s'_0)=\ell_1\). Therefore, due to Lemmas \ref{lemma:A} and \ref{lemma:Gif}, \((s'_0,s_1)\in \TR{\li1 \cguard{\guardp},\ell_2}\cup\TR{\li1 \cguard{\guardn},\ell_3}\).
Either \((s'_0,s_1)\in \TR{\li1 \cguard{\guardp},\ell_2}\) or \((s'_0,s_1)\in \TR{\li1 \cguard{\guardn},\ell_3}\).

In the first case, by Lemma \ref{lemma:A}, \( \thread(s_0)=\thread(s_1)\) and \(\lbl(s_1)=\ell_2\).
Therefore, \((s_0,s_1)\in\col_{{+}}\) and \(s_1\in\concr{S_{{+}}}\).
There exists a sequence \(s_2,s_n\) such that \(s_n=s\) and \(\forall k\in\{1,\ldots n-1\}\), \((s_k,s_{k+1})\in(\restrict{\concr{G_0}}{\after(s_0)}\cap\Tr)\cup\restrict{\concr{A_0}}{\compl{\after(s_0)}} \).

Let us prove by induction on \(k\), that \(\forall k\in\{1,\ldots n\}\), \((s_k,s_{k+1})\in(\restrict{\concr{G_0}}{\after(s_1)}\cap\TR{\li2\cmd,\ell_4})\cup\restrict{\concr{A_0}}{\compl{\after(s_0)}} \).
Let us consider the case \((s_k,s_{k+1})\in \restrict{\concr{G_0}}{\after(s_0)}\cap\Tr\). By induction hypothesis \((s_1,s_k)\in [(\restrict{\concr{G_0}}{\after(s_1)}\cap\TR{\li2\cmd,\ell_4})\cup\restrict{\concr{A_0}}{\compl{\after(s_0)}}]^{\star} \).
Hence, by Proposition \ref{prop:apres1}, either \((s_k,s_{k+1})\in\TR{\li1\cguard(\guardp),\ell_2}\) or \(s_k\in\after(s_1)\).
If  \((s_k,s_{k+1})\in\TR{\li1\cguard(\guardp),\ell_2}\) and \(s_k\in\after(s_1)\) then \((s_k,s_{k+1})\in\Sud_{{+}}\). This is contradictory with Claim \ref{claim:SudV}.
Therefore \(s_k\in\after(s_1)\).
By Lemma \ref{lemma:F}, \(\lbl(s_k)\in\Labs{\li2\cmd_1,\ell_4}\). Hence, by Lemmas \ref{lemma:A} and \ref{lemma:Gif}, \((s_1,s_k)\in\TR{\li2\cmd,\ell_4}\).

We conclude that \((s_1,s)\in[ (\restrict{\concr{G_0}}{\after(s_1)}\cap\TR{\li2\cmd,\ell_4})\cup\restrict{\concr{A_0}}{\compl{\after(s_0)}}]^{\ast}\subset\col_1\cap \ext[_{{+}}]{s_0}{s_1}\).

The second case is similar.
\end{proof}

\begin{claim}\label{claim:ifS}
 \(\concr{S'}\subset\concr{S}_1\cup\concr{S}_2\)
\end{claim}
\begin{proof}
 Let \(s\in\concr{S'}\). Therefore there exists \(s_0\in\concr{S}_0\) such that \((s_0,s)\in\col\) and \(\lbl(s)=\ell_4\neq\ell_1\). Hence, due to Lemma \ref{lemma:B}, \(s\notin\cinter{\concr{A_0}}\{s_0\}\).
 
 According to Lemma \ref{lemma:Cif}, there exists \(s_1\) such that either \begin{inparaenum}[(1)]\item
 \(s_1\in\concr{S_{{+}}}\) and \((s_1,s)\in\col_{1}\cap\ext[_{{+}}]{s_0}{s_1}\),
 \item or, \(s_1\in\concr{S_{{\neg}}}\) and \((s_1,s)\in\col_{2}\cap\ext[_{{\neg}}]{s_0}{s_1}\).\end{inparaenum}

In the first case, by definition, \(s\in\concr{S_1}\) and in the second case \(s\in\concr{S_2}\)
\end{proof}

\begin{claim}\label{claim:ifSe}\(\Se\subset\Se_{{+}}\cup\Se_1\cup\Se_{{\neg}}\cup\Se_2\).
\end{claim}
\begin{proof}
 Let \((s,s') \in \Se \).
 Then, there exists \(s_0\in\concr{S_0})\) such that \((s_0,s)\in\col\).
 
 Let us consider the case \(s\in\cinter{\concr{A_0}}(\{s_0\})\). By Lemma \ref{lemma:B}, \(\lbl(s)=\ell_1\).
 Hence, by Lemmas \ref{lemma:A} and \ref{lemma:Gif}, \((s,s') \in \TR{\li1 \cguard{\guardp},\ell_2}\cup\TR{\li1 \cguard{\guardn},\ell_3}\).
 Hence, \((s,s')\in \Se_{{+}}\cup\Se_{{\neg}}\).
 
 According to Lemma \ref{lemma:Cif}, if \(s\notin\cinter{\concr{A_0}}(\{s_0\})\), then, there exists \(s_1\) such that either \begin{inparaenum}[(1)]\item
 \(s_1\in\concr{S_{{+}}}\) and \((s_1,s)\in\col_{1}\cap\ext[_{{+}}]{s_0}{s_1}\),
 \item or, \(s_1\in\concr{S_{{\neg}}}\) and \((s_1,s)\in\col_{2}\cap\ext[_{{\neg}}]{s_0}{s_1}\).\end{inparaenum}
 
 In the first case, by Lemma \ref{lemma:F}, \(\lbl(s_k)\in\Labs{\li2\cmd_1,\ell_4}\). Hence, by Lemmas \ref{lemma:A} and \ref{lemma:Gif}, \((s_1,s_k)\in\TR{\li2\cmd,\ell_4}\) and therefore \((s,s')\in\Se_1\).
 
In the second case, we similarly conclude that  \((s,s')\in\Se_2\).
\end{proof}

\begin{claim}\label{claim:ifSu}\(\Su\subset\Su_1\cup\Su_2\).
\end{claim}
\begin{proof}
 Let \((s,s')\in\Su\).
 Therefore, there exists \(s_0\in\concr{S_0}\) and \(s_2\) such that \((s_0,s_2)\in\col\) and \((s_2,s)\in\Sche\) and \(s\in\after(s_0)\). Notice that \(\thread(s_0)=\thread(s_2)\neq\thread(s)\).
 
 Assume by contradiction that \(s_2\in\cinter(\{s_0\})\). Hence, due to Lema \ref{lemma:E'}, \(\thread(s)=\thread(s_0)\). This is contradictory.
 
 Therefore, according to Lemma \ref{lemma:Cif}, there exists \(s_1\) such that either \begin{inparaenum}[(1)]\item
 \(s_1\in\concr{S_{{+}}}\) and \((s_1,s)\in\col_{1}\cap\ext[_{{+}}]{s_0}{s_1}\),
 \item or, \(s_1\in\concr{S_{{\neg}}}\) and \((s_1,s)\in\col_{2}\cap\ext[_{{\neg}}]{s_0}{s_1}\).\end{inparaenum}
  In the two cases, by Lemma \ref{lemma:H}, \(s\in\after(s_1)\).
 
 In the first case, by Lemma \ref{lemma:F}, \(\lbl(s)\in\Labs{\li2\cmd_1,\ell_4}\) and therefore, by Lemmas \ref{lemma:Gif} and \ref{lemma:A}, \((s,s')\in\TR{\li2\cmd_1,\ell_4}\).
 Hence, \((s,s')\in\Su_1\)
 
 In the second case, we similarly conclude that  \((s,s')\in\Su_2\).
 \end{proof}
 
\begin{claim}\label{claim:ifSud}\(\Sud\subset\Sud_1\cup\Sud_2\).
\end{claim}
\begin{proof}
 Let \((s,s')\in\Sud\).
 Therefore, there exists \(s_0\in\concr{S_0}\) and \(s_2\in\concr{S'}\) such that \((s_0,s_2)\in\col\) and \((s_2,s)\in\ext{s_0}{s_2}\) and \(s\in\after(s_0)\smallsetminus\after(s_2)\). Notice that \(\thread(s_0)=\thread(s_2)\neq\thread(s)\).
 
 Assume by contradiction that \(s_2\in\cinter(\{s_0\})\). Hence, due to Lemma \ref{lemma:B}, \(\lbl(s_2)=\ell_1\). This is contradictory with \(s_2\in\concr{S'}\).
 
 Therefore, according to Lemma \ref{lemma:Cif}, there exists \(s_1\) such that either \begin{inparaenum}[(1)]\item
 \(s_1\in\concr{S_{{+}}}\) and \((s_1,s)\in\col_{1}\cap\ext[_{{+}}]{s_0}{s_1}\),
 \item or, \(s_1\in\concr{S_{{\neg}}}\) and \((s_1,s)\in\col_{2}\cap\ext[_{{\neg}}]{s_0}{s_1}\).\end{inparaenum}
  In the two cases, by Lemma \ref{lemma:H}, \(s\in\after(s_1)\).
 
 In the first case, because \(s\notin\after(s_2)\), by Proposition \ref{prop:apres1}, \((s,s')\in\TR{\li1\cmd_1,\ell_2}\). Hence, \((s,s')\in\Sud_1\)

 In the second case, we similarly conclude that  \((s,s')\in\Sud_2\).
 \end{proof}
 
 Property \ref{fif} of Theorem \ref{theorem:denot} is a straightforward consequence of Claims \ref{claim:ifS}, \ref{claim:ifSe},  \ref{claim:ifSu}, \ref{claim:ifSud}.

\subsubsection{Proof of Property \ref{fwhile} of Theorem \ref{theorem:denot}}

In this section, we consider a command \(\li1 \cwhile(\guardp)\{\li2\cmd\},\ell_3\) and an initial configuration \(\concr{Q}_0=\qcp{_0}\).\\
Let \(\concr{Q}'=\qcp{'}=\av{\li1 \cwhile(\guardp)\{\li2\cmd\},\ell_3 }\concr{Q}_0\).\\
Let \(\concr{Q}_{\omega}=\qcp{_{\omega}}=\loopw\eomega(\concr{Q}_0)\).\\
Let \(\concr{Q}''=\qcp{''}=\av{\li1 \cwhile(\guardp)\{\li2\cmd\},\ell_3 }\concr{Q}_{\omega}\).\\
Let \(\concr{K}=\nqp{}=\ssem{\li1 \cwhile(\guardp)\{\li2\cmd\},\ell_3 }\concr{Q}_{\omega}\).\\
Let \(\concr{Q}_{+}=\qcp{_{+}}=\av{\li1\cguard(\guardp),\ell_2}(\concr{Q}_{\omega})\).\\
Let \(\concr{K}_{+}=\nqp{_{+}}=\ssem{\li1\cguard(\guardp),\ell_2}(\concr{Q}_{\omega})\).\\
Let \(\concr{K}_{\cmd}=\nqp{_{\cmd}}=\ssem{  \li2\cmd,\ell_1}(\concr{Q}_{+})\).\\
Let \(\concr{Q}_{\neg}=\qcp{_{\neg}}= \av{\li1 \cguard(\guardn),\ell_3}\concr{Q}_{\omega}\).\\
Let \(\concr{K}_{\neg}=\nqp{_{\neg}}= \ssem{\li1 \cguard(\guardn),\ell_3}\concr{Q}_{\omega}\).\\
Let \(\Tr = \TR{\li1 \cwhile(\guardp)\{\li2\cmd\},\ell_3}\).

\begin{lemma}\label{lemma:Gw}
 \[\TR{\li1 \cwhile(\guardp)\{\li2\cmd\},\ell_3}=\TR{\li1 \cguard(\guardn),\ell_3}\cup\TR{\li1\cguard(\guardp),\ell_2}\cup\TR{\li2\cmd,\ell_1}\]
\end{lemma}

Notice that, by definition, \(\concr{Q}_0\leqslant \concr{Q}_{\omega}\)

\begin{lemma}\label{lemma:apresw} We use the above notations.
 Let \(s_0,s_1,\ldots,s_n,\ldots,s_m\) a sequence of states such that \((s_0,s_m)\in\col_{\omega}\), \((s_0,s_n)\in\col_{\omega}\) , \(s_n\in\concr{S}_{\omega} \) and for all \(k\in\{0,\ldots,m-1\}\), \((s_k,s_{k+1})\in (\restrict{{\concr{G}_{\omega}}}{\after(s_0)}\cap\Tr)\cup\restrict{\concr{A_{\omega}}}{\compl{\after(s_0)}}\).
 
 Therefore, \((s_n,s_{m})\in\col_{\omega}\) .
\end{lemma}

\begin{proof}
For all \(k\), 
 \((s_k,s_{k+1})\in
 (\restrict{{\concr{G}_{\omega}}}{\after(s_n)}\cap\Tr)\cup (\restrict{{\concr{G}_{\omega}}}{\after(s_0)\smallsetminus\after(s_n)}\cap\Tr)\cup\restrict{\concr{A_{\omega}}}{\compl{\after(s_0)}}\).
 
 Let \(k_0\geqslant n\) such that \((s_{k_0},s_{k_0+1})\in
 (\restrict{{\concr{G}_{\omega}}}{\after(s_0)\smallsetminus\after(s_n)}\cap\Tr)\).
 Notice that \((s_n,s_{k_0})\in\ext[_{\omega}]{s_0}{s_n}\) and \(s_{k_0}\in \after(s_0)\smallsetminus\after(s_n)\).
 Hence, \((s_{k_0},s_{k_0+1}) \in \Sud_{\omega}\subset \concr{A_{\omega}}\).
 Therefore \((s_{k_0},s_{k_0+1}) \in \restrict{\concr{A_{\omega}}}{\compl{\after(s_1)}}\).
 
 In addition to this, according to Lemma \ref{lemma:R}, \(\after(s_n)\subset\after(s_0)\), so, for all \(k\geqslant n\), 
 \((s_k,s_{k+1})\in
 (\restrict{{\concr{G}_{\omega}}}{\after(s_n)}\cap\Tr)\cup\restrict{\concr{A_{\omega}}}{\compl{\after(s_0)}}\).
 \end{proof}

\begin{lemma}\label{lemma:Cw}
 Using the notations of this section, if \(s\in\col\langle\concr{S}_{0}\rangle\), then, there exists \(s_0\in\concr{S}_{\omega}\) such that:
 \begin{enumerate}
  \item either \((s_0,s)\in\col_{\neg}\),
  \item or there exists \(s_1\in\concr{S}_{+}\) such that \((s_0,s_1)\in\col_{+}\) and \((s_1,s)\in \col_{\cmd}\) and \(\lbl(s)\neq\ell_1\).
 \end{enumerate}
\end{lemma}

\begin{proof}
 Let \(s\in\col\langle\concr{S}_0\rangle\).
 We consider a sequence \(s_0,\ldots,s_n\) of minimal length such that the following properties hold:
 \begin{inparaenum}[(1)]
  \item \(s_n=s\), 
  \item \(s_0\in\concr{S}_{\omega}\), 
  \item for all \(k\in\{0,\ldots,n-1\}\), \((s_k,s_{k+1}) \in (\restrict{{\concr{G}_{\omega}}}{\after(s_0)} \cap \Tr)\cup \restrict{{\concr{A}_{\omega}}}{\compl{\after(s_0)}}\)
 \end{inparaenum}. 
A such sequence exists because \(\concr{S}_0\subset\concr{S}_{\omega}\).

If for all \(k\in\{0,\ldots,n-1\}\), \((s_k,s_{k+1})\in\Sche\cup\restrict{{\concr{A}_{\omega}}}{\compl{\after(s_0)}} \) then \((s_0,s)\in\col_{+}\cap \col_{\neg} \subset \col_{\neg}\).

Let us assume, from now, that there exists \(k\in\{0,\ldots,n-1\}\) such that  \((s_k,s_{k+1})\in\restrict{{\concr{G}_{\omega}}}{\after(s_0)} \cap \TR{\li1 \cwhile(\guardp)\{\li2\cmd\},\ell_3}\smallsetminus\Sche\).
Let \(k_0\) the smallest such \(k\).

Therefore \((s_{k_0},s_{k_0+1})\in \restrict{{\concr{G}_{\omega}}}{\after(s_0)}\), so, \(s_{k_{0}}\in\after(s_0)\). According to Lemma \ref{lemma:E'}, \(\thread(s_0)=\thread(s_{k_0})\).
By Lemma \ref{lemma:B}, \(\lbl(s_0)=\lbl(s_{k_0})\).
But \(\lbl(s_0)=\ell_1\), therefore, by Lemma \ref{lemma:Abis}, \((s_{k_0},s_{k_0+1})\notin\TR{\li2\cmd,\ell_1}\).
Therefore, by Lemma \ref{lemma:Gw}, 
either \((s_{k_0},s_{k_0+1})\in\TR{\li1\cguard(\guardn),\ell_3}\) or \((s_{k_0},s_{k_0+1})\in\TR{\li1\cguard(\guardp),\ell_2}\).

In the first case, by Lemma \ref{lemma:A'}, \(\lbl(s_{k_0+1})=\ell_3\). Let us prove by induction on \(k\) that for all \(k>k_0\), \((s_k,s_{k+1})\in \restrict{\concr{A_{\omega}}}{\compl{\after(s_0)}} \cup \Sche\).
By induction hypothesis \((s_{k_0},s_{k})\in[\restrict{\concr{A_{\omega}}}{\compl{\after(s_0)}} \cup \Sche]^{\star}\). Let us consider the case \((s_k,s_{k+1})\in \restrict{\concr{G_{\omega}}}{\after(s_0)}\cap \Tr\). Therefore \(s_k\in\after(s_0)\), then by Lemma \ref{lemma:E'}, \(\thread(s_k)=\thread(s_{k_0+1})\). By Lemma \ref{lemma:B}, \(\lbl(s_k)=\lbl(s_{k_0+1})=\ell_3\). So, by Lemma \ref{lemma:Abis}, \((s_k,s_{k+1})\in\Sche\).
Hence \((s_0,s)\in \col_{\neg}\).

In the second case, \((s_0,s_{k_0+1})\in\col_{{+}}\) and therefore, by Lemma \ref{lemma:A'}, \(s_{k_0+1}\in\concr{S}_{{+}}\).
Either there exists \(k_1>k_0\) such that \((s_{k_1},s_{k_1+1})\in \restrict{\concr{G}}{\after(s_0)}\cap(\TR{\li1\cguard(\guardn),\ell_3}\cup\TR{\li1\cguard(\guardp),\ell_3})\) or there does not exists a such \(k_1\).

Assume by contradiction that \(k_1\) exists, therefore, by Lemma \ref{lemma:A'}, \(\lbl(s_{k_0})=\ell_1\). According to Lemma \ref{lemma:F}, \(\thread(s)=\thread(s_0)\). Hence, \((s_0,s_{k_1})\in\col_{\omega}\).
So, by Lemma \ref{lemma:apresw}, \((s_{k_1},s_{n})\in\col_{\omega}\). This is contradictory with the minimality of the path \(s_1,\ldots,s_n\). Therefore \(k_1\) does not exists. 

Hence, 
for all \(k>k_0\), \((s_{k},s_{k+1})\in (\restrict{\concr{G_{\omega}}}{\after(s_0)}\cap\TR{\li2\cmd,\ell_1})\cup\restrict{\concr{A_{\omega}}}{\compl{\after(s_0)}}\).
According to proposition \ref{prop:apres1}, for all \(k>k_0\), \((s_{k},s_{k+1})\in (\restrict{\concr{G_{\omega}}}{\after(s_1)}\cap\TR{\li2\cmd,\ell_1})\cup\restrict{\concr{A_{\omega}}}{\compl{\after(s_0)}}\).
Therefore, \((s_{k_0},s)\in\col_{\omega}\)
\end{proof}

\begin{claim}\label{claim:WS}
 Using the notation of this section \(\concr{S'}\subset\concr{S}_{\neg}\).
 \end{claim}
\begin{proof}
 Let \(s\in\concr{S'}\), therefore, \(s\in\col\langle \concr{S}_0 \rangle\). 
 Furthermore, \(\lbl(s)=\ell_3\). Hence, according to Lemma \ref{lemma:R}, for all \(s_1\), \((s_1,s)\notin \col_{\omega}\).
 Therefore, according to Lemma \ref{lemma:Cw}, there exists \(s_0\in\concr{S_{\omega}}\) such that \((s_0,s)\in\col_{\neg}\). Hence \(s\in\concr{S}_{\neg}\). 
\end{proof}

\begin{claim}\label{claim:WSe}
 \(\Se\subset\Se_{\neg}\cup\Se_{{+}}\cup\Se_{\cmd}\)
\end{claim}
\begin{proof}
 Let \((s,s')\in\Se\). According to Lemma \ref{lemma:Gw}, \((s,s')\in \TR{\li1 \cguard(\guardn),\ell_3}\cup\TR{\li1\cguard(\guardp),\ell_2}\cup\TR{\li2\cmd,\ell_1} \).
 
 Let us consider the case \((s,s')\in \TR{\li1 \cguard(\guardn),\ell_3}\cup\TR{\li1\cguard(\guardp),\ell_2}\). Due to Lemma \ref{lemma:A'}, \(\lbl(s)=\ell_1\)
 Hence, according to Lemma \ref{lemma:Cw}, either \((s_0,s)\in\col_{\neg}\) or there exists \(s_1\in\concr{S}_{+}\) such that \((s_1,s)\in \col_{\cmd}\) (contradiction with Lemma \ref{lemma:R} and \(\lbl(s)=\ell_1\)).
 According to Lemma \ref{lemma:C}, either \(\lbl(s)=\ell_2\neq\ell_1\) (contradiction) or \(s\in\cinter{\concr{A}_0}(\concr{S_0})\subset\col_{\neg}\langle\concr{S_\omega}\rangle\cap\col_{{+}}\langle\concr{S_\omega}\rangle\). Therefore either \((s,s')\in \Se_{\neg}\) or \((s,s')\in \Se_{{+}}\).

 Let us consider the case \((s,s')\in \TR{\li2\cmd,\ell_1} \). Therefore, according to Lemma \ref{lemma:A}, \(\lbl(s)\in \Labs{\li2\cmd,\ell_1}\smallsetminus\{\ell_1\} \).
If \(s''\in\col_{\neg}\langle\concr{S_{\omega}}\rangle\), then, by Lemma \ref{lemma:C}, \(\lbl(s'')\in\{\ell_1,\ell_3\}\). Hence, \(s\notin\col_{\neg}\langle\concr{S_{\omega}}\rangle\).
So, by Lemma \ref{lemma:Cw}, there exists \(s\in\concr{S}_0\) and \(s_1\in\concr{S}_{{+}}\) such that \((s_0,s_1)\in\col_{{+}}\) and \((s_1,s)\in\col_{\cmd}\). According to Proposition \ref{prop:apres1}, \((s,s')\in\after(s_1)\) and therefore \((s,s')\in\Se_{\cmd}\).
\end{proof}

\begin{claim}\label{claim:WSu}
 \(\Su\subset\Su_{\cmd}\)
\end{claim}
\begin{proof}
 Let \((s,s')\in\Su\). There exists \(s_0\) and \(s_2\) such that \((s_0,s_2)\in \col_{\omega}\).
 By Lemma \ref{lemma:C}, either \((s_0,s_2)\in\col_{\neg}\) or there exists \(s_1\in\concr{S}_{{+}}\) such that \((s_0,s_1)\in\col_{{+}}\) and \((s_1,s_2)\in\col_{\cmd}\) and \(\lbl(s_2)\neq\ell_2\).
  
  In the first case, because \(s\in\after(s_0)\), by Lemma \ref{lemma:E'}, \(\thread(s)=\thread(s_0)\). But, by definition of \(\Sche\) and \(\col_{{\neg}}\), \(\thread(s_2)\neq\thread(s)\) and \(\thread(s_0)=\thread(s_2)\). This is contradictory.
  
  In the second case, by Proposition \ref{prop:apres1}, \(s\in\after{s_1}\). Because \(\thread(s)\neq\thread(s_0)=\thread(s_2)\), by Lemma \ref{lemma:F}, \(\lbl(s)\in\Labs{\li2\cmd,\ell_1}\smallsetminus\{\ell_2\}\). Therefore, by Lemmas \ref{lemma:Gw} and \ref{lemma:A'}, \((s,s')\in\TR{\li2\cmd,\ell_1}\). Hence \((s,s')\in \Su_{\cmd}\)
\end{proof}

\begin{claim}\label{claim:WSud}
 \(\Sud\subset\Sud_{\neg}\)
\end{claim}
\begin{proof}
 Let  \((s,s')\in\Sud\). Therefore, there exists \(s_0\concr{S}_{\omega}\) and \(s_1\in\concr{S}'\) such that \((s_0,s_1)\in\col\) and \((s_1,s)\in\ext{s_0}{s_1}\).
 
 Notice that \(\lbl(s_1)=\ell_3\), therefore, according to Lemma \ref{lemma:R}, \(s_1\notin\col_{{+}};\col_{\cmd}\langle\concr{S_\omega}\rangle\).
 hence, by Lemma \ref{lemma:Cw}, \((s_0,s_1)\in\col_{{\neg}}\).
 
 \((s_1,s)\in\ext{s_0}{s_1}\subset (\restrict{\concr{G_{\omega}}}{\after(s_0)}\cap\Tr)\cup\restrict{\concr{A_{\omega}}}{\compl{\after(s_0)}}\cup\restrict{\concr{G_{\omega}}}{\after(s_1)}\).
 By Proposition \ref{prop:apres1},  \((s_1,s)\in(\restrict{\concr{G_{\omega}}}{\after(s_0)}\cap\TR{\li1\cguard(\guardn),\ell_2})\cup (\restrict{\concr{G_{\omega}}}{\after(s_1)}\cap\Tr\smallsetminus\TR{\li1\cguard(\guardn),\ell_2})\cup\restrict{\concr{A_{\omega}}}{\compl{\after(s_0)}}\cup\restrict{\concr{G_{\omega}}}{\after(s_1)}=\ext[_{\neg}]{s_1}{s_2}\). 
\end{proof}

Property \ref{fwhile} of Theorem \ref{theorem:denot} is a straightforward consequence of Claims \ref{claim:WS}, \ref{claim:WSe}, \ref{claim:WSu} and \ref{claim:WSud}.

\subsubsection{Proof of Property \ref{fcreate} of Theorem \ref{theorem:denot}} 

Let \(\concr{Q}_0 =\qcp{_0}\) a configuration.\\
Let \(\concr{Q}'=\qcp{'}=\av{\li1\ccreate(\li2\cmd),\ell_3 }(\concr{Q}_0)\)\\
Let \(\concr{K}=\nqp{}=\ssem{\li1\ccreate(\li2\cmd),\ell_3 }(\concr{Q}_0)\)\\
Let \(\concr{Q}_1=\qcp{_1}=\av{\li1\crcreate(\ell_2),\ell_3}(\concr{Q}_0)\)\\
Let \(\concr{K}_1=\nqp{_1}=\ssem{\li1\crcreate(\ell_2),\ell_3}(\concr{Q}_0)\)\\
Let \(\concr{Q}_2=\qcp{_2}= \cinit_{\ell_2}(\concr{Q}_1)\)\\
Let \(\concr{G}_{\infty}=\afp{\li2 \cmd,\lend}(\concr{Q}_2)\)\\
Let \(\concr{K}_3=\nqp{_3} ={\ssem{\li2 \cmd,\lend}}\langle \concr{S}_2,\concr{G}_{\infty},\concr{A}_2\rangle\)\\
Let \(\concr{Q}_3=\qcp{_3}= \ccomb_{\concr{Q}_0} (\concr{G}_{\infty})\)
Let \( \Tr=\TR{\li1\ccreate(\li2\cmd),\ell_3}\)

\begin{lemma}\label{lemma:Gcr}
 \(\TR{\li1\ccreate(\li2\cmd),\ell_3} = \TR{\li1\cspawn(\ell_2),\ell_3}\cup\TR{\li2\cmd,\lend}\)
\end{lemma}

\begin{lemma}\label{lemma:apresC0}
Let \(T\) a set of transitions.
 Let \(s_0\), \(s_1\), \(s_2\), \(s\) and \(s'\) such that \((s_0,s_1)\in\col_1\), \(s_2\in\cschedule\{s_1\}\), \(\lbl(s_1)=\ell_3\), \((s_2,s)\in T^{\star}\) and \(s\in\after(s_0)\).
 
 Therefore, \(s\in\after(s_1) \cup \after(s_2)\).
\end{lemma}
\begin{proof}
According to Lemma \ref{lemma:C}, there exists \(s'_0\) and \(s'_1\) such that, \(s'_0\in\cinter{\concr{A_0}}\{s_0\}\), \((s'_0,s'_1)\in \TR{\li1\cspawn(\ell_2),\ell_3}\smallsetminus\Sche\), and  \(s_1\in\cinter{\concr{A_0}}\{s'_1\}\).

By Lemmas \ref{lemma:E'} and \ref{lemma:A}, \(\thread(s_0)=\thread(s'_0)=\thread(s'_1)=\thread(s_1)\).

Let \(i_0=\thread(s_0)\) and \(i=\thread(s)\).

Let \(\h_0\), \(\h'_0\), \(j\), \(\h_1\) and \(\h\) such that, respectively, the genealogy of 
\(s_0\), \(s'_0\), \(s''_0\), \(s_1\), \(s_2\), \(s\) is \(\h_0\), \(\h_0\cdot\h'_0\), \(\h_0\cdot\h'_0\cdot(i_0,\ell_2,j)\), \(\h_0\cdot\h'_0\cdot(i_0,\ell_2,j)\cdot\h_1\), \(\h_0\cdot\h'_0\cdot(i_0,\ell_2,j)\cdot\h_1\),  \(\h_0\cdot\h'_0\cdot(i_0,\ell_2,j)\cdot\h_1\cdot\h\).
Notice that \(s_1\) and \(s_2\) have the same genealogy.

Because \((s_0,s'_{0})\in [\restrict{\concr{A_0}}{\compl{\after(s_0)}} \cup\Sche]^{\ast}\), by Lemma \ref{lemma:Ndescew}, \(\desce_{\h'_0}\{i_0\}= \{i_0\}\).

 Because \((s''_1,s_1)\in [\restrict{\concr{A_0}}{\compl{\after(s_0)}} \cup\Sche]^{\ast}\), by Lemma \ref{lemma:Ndescew}, \(\desce_{(i_0,\ell_2,j)\cdot\h_1}\{i_0\}= \desce_{(i_0,\ell_2,j)}\{i_0\}=\{i_0,j\}\).
 
 By definition of \(\desce\), \(\desce_{\h'_0\cdot(i_0,\ell_2,j)\cdot\h_1\cdot\h}(\{i_0\}) = \desce_{\h}[\desce_{(i_0,\ell_2,j)\cdot\h_1}(\desce_{\h'_0}\{i_0\})]=\desce_{\h}\{i_0,j\}\)
 By definition of \(\desce\), \(\desce_{\h'_0\cdot(i_0,\ell_2,j)\cdot\h_1\cdot\h}(\{i_0\}) = \desce_{\h}(\{i_0\})\cup\desce_{\h}(\{j\})\).
  
 Because \(s\in\after(s_0)\), \(i\in\desce_{\h'_0\cdot(i_0,\ell_2,j)\cdot\h_2\cdot\h}(\{i_0\})\). Therefore either \(i\in\desce_{\h}(\{i_0\}) \) or \(i\in\desce_{\h}(\{j\}) \).
 If \(i\in\desce_{\h}(\{i_0\}) \) then \(s\in\after(s_1)\).
 If \(i\in\desce_{\h}(\{j\}) \) then \(s\in\after(s_2)\).
\end{proof}

\begin{lemma}\label{lemma:apresC}
 Let \(s_0\), \(s_1\), \(s_2\), \(s\) and \(s'\) such that \((s_0,s_1)\in\col_1\), \(s_2\in\cschedule\{s_1\}\), \(\lbl(s_1)=\ell_3\), \((s_2,s)\in\restrict{(\concr{G_0}\cup\concr{A_0})}{\compl{\after(s_1)}}^{\star}\) and \((s,s') \in\restrict{\concr{G_0}}{\after(s_0)}\cap\Tr \).
 
 Therefore, \(s\in\after(s_2)\) (i.e.,  \((s,s') \in\restrict{\concr{G_0}}{\after(s_2)}\cap\Tr \)).
\end{lemma}
\begin{proof}
 Due to Lemma \ref{lemma:apresC0}, \(s\in\after(s_1) \cup \after(s_2)\).
 Assume by contradiction that \(s\in\after(s_1)\). Therefore, by Lemma \ref{lemma:E'}, \(\thread(s)=\thread(s_1)\) and by Lemma \ref{lemma:B}, \(\lbl(s)=\lbl(s_1)=\ell_3\).
 This is contradictory with Lemma \ref{lemma:A} which implies \(\lbl(s)\neq\ell_3\).
\end{proof}

\begin{lemma}\label{lemma:Ccr}
 If \((s_0,s)\in\col\) then:
 \begin{itemize}
  \item either \(s\in\cinter{\concr{A_0}}(s_0)\) and \(\lbl(s)=\ell_1\)
  \item or there exists \(s_1,s_2,s_3\) such that \((s_0,s_1)\in\col_1\), \((s_1,s_2)\in\Sche\), \((s_2,s_3)\in\col_3\cap\ext[_1]{s_0}{s_1}\), \((s_3,s)\in\Sche\) and \(s_2\in\cschedule\{s_1\}\). Furthermore \(\lbl(s_1)=\lbl(s)=\ell_3\) and \(s\in\cinter{\concr{G_0}\cup\concr{A_0}}\{s_1\}\).
 \end{itemize}
\end{lemma}
\begin{proof}
 
 If \((s_0,s)\in [\restrict{\concr{A_0}}{\compl{\after(s_0)}} \cup\Sche]^{\ast}\) then \(s\in\cinter{\concr{A_0}}(s_0)\) and by Lemma \ref{lemma:B}, \(\lbl(s)=\ell_1\).

Then, let us consider the other case: \((s_0,s)\notin [\restrict{\concr{A_0}}{\compl{\after(s_0)}} \cup\Sche]^{\ast}\). Therefore, there exists \(s'_0\) and \(s_1\) such that \((s_0,s'_{0})\in [\restrict{\concr{A_0}}{\compl{\after(s_0)}} \cup\Sche]^{\ast}\), \((s'_{0},s_{1})\in (\restrict{\concr{G_0}}{\after(s_0)}\cap\Tr)\) and  \((s_1,s)\in [(\restrict{\concr{G_0}}{\after(s_0)}\cap\Tr)\cup\restrict{\concr{A_0}}{\compl{\after(s_0)}}]^{\star}\).

 Due to Lemma \ref{lemma:E'}, because \(s'_{0}\in\after(s_0)\), \(\thread(s'_{0})=\thread(s_0)\).
 According to Lemma \ref{lemma:A'}, \(\thread(s_{1})=\thread(s'_{0})=\thread(s_0)\) and \(\lbl(s_{1})=\ell_3\). Therefore \((s_0,s_{1})\in\col_1\).
 
 Let \(\cstate{_1}=s_1\). Let \(\h'_1\) and \(j\) such that \(\h'_1\cdot(i,\ell_2,j)=\h_1\).
 Let \(s_2=(j,P_1,\m_1,\h_1)\). Therefore, \(s_2\in\cschedule\{s_1\}\) and \((s_1,s_2)\in\Sche\).
 Let \((i,P,\m,\h)=s\) and \(s_3=(j,P,\m,\h)\). Therefore, \((s_3,s)\in\Sche\).
 
 Given that \(\Sche\subset\concr{A}_0\cap\concr{G}_0\cap\Tr\), we conclude that \((s_2,s_3)\in [(\restrict{\concr{G_0}}{\after(s_0)}\cap\Tr)\cup\restrict{\concr{A_0}}{\compl{\after(s_0)}}]^{\star}\).
 Using Lemma \ref{lemma:apresC} and a straightforward induction, \((s_2,s_3)\in [(\restrict{\concr{G_0}}{\after(s_2)}\cap\Tr)\cup\restrict{\concr{A_0}}{\compl{\after(s_0)}}]^{\star}\).
 Then \((s_2,s_3)\in\ext[_1]{s_0}{s_1}\).
 Furthermore by Lemma \ref{lemma:F+}, \(\after(s_2)\subset\after(s_0)\). Hence \((s_2,s_3)\in [(\restrict{\concr{G_0}}{\after(s_2)}\cap\Tr)\cup\restrict{\concr{A_0}}{\compl{\after(s_2)}}]^{\star}\). Therefore, by Proposition \ref{proposition:guarantee}, \((s_2,s_3)\in\col_3\).
 \end{proof}

\begin{claim}
 \(\concr{S'}\subset\cinter{\concr{G_0}\cup\concr{A_0}}(\concr{S_1})\).
\end{claim}
\begin{proof}
 Let \(s\in\concr{S'}\). Therefore there exists \(s_0\in\concr{S_0}\) such that \((s_0,s)\in\col\) and \(\lbl(s)=\ell_3\neq\ell_1\).
 According to Lemma \ref{lemma:Ccr} there exists \(s_1\) such that \((s_0,s_1)\in\col_1\), \(\lbl(s_1)=\ell_3\) and \(s\in\cinter{\concr{G_0}\cup\concr{A_0}}\{s_1\}\).
 Therefore \(s_1\in\concr{S}_1\) and \(s\in\cinter{\concr{G_0}\cup\concr{A_0}}(\concr{S_1})\).
\end{proof}

\begin{claim}
 \(\Se\subset\Se_1\).
\end{claim}
\begin{proof}
Let \((s,s')\in\Se\). According to Lemma \ref{lemma:A}, \(\lbl(s)\neq\ell_3\).
There exists \(s_0\in\concr{S_0}\) such that \((s_0,s)\in \col\).
Therefore, according to lemma \ref{lemma:Ccr}, \(s\in\cinter{\concr{A_0}}\{s_0\}\). Therefore \((s_0,s)\in\col_1\) and, by Lemma \ref{lemma:B}, \(\lbl(s)=\ell_1\). Due to Lemmas \ref{lemma:Abis} and \ref{lemma:Gcr}, \((s,s')\in\TR{\li1\cspawn(\ell_2),\ell_3}\).
Hence \((s,s')\in\Se_1\).
\end{proof}

\begin{claim}
 \(\Su\subset\Se_3\cup\Su_3\).
\end{claim}
\begin{proof}
Let \((s,s')\in \Su\).
Therefore, there exists \(s_0\in\concr{S}_0\) such that \((s_0,s)\in\col;\Sche\) and \(s\in\after(s_0)\).
Notice that by definition of \(\Sche\), \(\thread(s_0)\neq\thread(s)\).

Assume by contradiction, that \(s\in\Sche\langle\cinter{\concr{A_0}}\{s_0\}\rangle\).
Due to Lemma \ref{lemma:E'}, \(\thread(s_0)=\thread(s)\). This is contradictory.

Hence, by Lemma \ref{lemma:Ccr}, there exists \(s_1, s_2,s_3\) such that \((s_0,s_1)\in\col_1\), \((s_1,s_2)\in\Sche\), \((s_2,s_3)\in\col_3\), \((s_3,s)\in\Sche\), \(s_2\in\cschedule\{s_1\}\),  and  \(\lbl(s_1)=\lbl(s)=\ell_3\).

 Hence, \(s_1\in\concr{S_1}\), \(s_2\in\concr{S}_2\). 

 According to Lemma \ref{lemma:F-} \(\after(s_1)\cap\after(s_2)=\emptyset\). Given that \((s_2,s)\in\col;\Sche;\Sche\), \((s_2,s)\in \restrict{(\concr{G_0}\cup\concr{A_0})}{\compl{\after(s_1)}}^{\star}\).
 Hence, du to Lemma \ref{lemma:apresw}, \(s\in\after(s_2)\).
 
 If \(\thread(s)=\thread(s_2)\), then \((s_2,s)\in\col_3\) and \((s,s')\in\Se_3\).
 If \(\thread(s)\neq\thread(s_2)\), then \((s,s')\in\Su_3\).
\end{proof}

\begin{claim}
 \(\Sud\subset\Se_3\cup\Su_3\).
\end{claim}
\begin{proof}
 Let \((s,s')\in \Sud\). There exists \(s_0,s_4\) such that \((s_0,s_4)\in\col\) and \((s_4,s)\in\ext{s_0}{s_4}\) and \(s_4\in\concr{S'}\).
 By Lemma \ref{lemma:Ccr}, there exists \(s_1,s_2,s_3\) such that \((s_0,s_1)\in\col_1\), \(s_2\in\cschedule_{\concr{A}}(\{s_1\})\), \((s_2,s_3)\in \col_3\cap\ext[_1]{s_0}{s_1}\) and \((s_3,s_4)\in\Sche\).
  
Furthermore, \(s\in\after(s_0)\smallsetminus\after(s_4)\). Due to Lemma \ref{lemma:apresC0}, either
 \(s\in\after(s_1)\smallsetminus\after(s_4)\) or \(s\in\after(s_2) \smallsetminus\after(s_4)\).

Assume by contradiction that \(s\in\after(s_1)\smallsetminus\after(s_4)\). Therefore \((s,s')\in \Sud_1\). But, by Claim \ref{claim:SudV}, \(\Sud_1=\emptyset\).
Therefore \(s\in\after(s_2) \smallsetminus\after(s_4)\).

Let \(\cstate{}=s\) and \(s_{5}=(\thread(s_2),P_5,\m_5,\h_5)\).

Given that \((s_4,s)\in\ext{s_0}{s_4}\), \((s_4,s)\in[(\restrict{\concr{G_0}}{\after(s_0)}\cap\Tr)\cup\restrict{{\concr{A}_2}}{\compl{\after(s_0)}}]^{\ast}\) and by Lemma \ref{lemma:apresC0}, \((s_4,s)\in[(\restrict{\concr{G_0}}{\after(s_1)\cup\after(s_2)}\cap\Tr)\cup\restrict{{\concr{A}_2}}{\compl{\after(s_0)}}]^{\ast}\).

By definition of \(\cpost\), \(\after(s_1)\subset\cextract{\ell_2}\). Furthermore by Lemma \ref{lemma:F-}, \(\after(s_1)\cap\after(s_2)=\emptyset\). Therefore \(\after(s_1)\subset\cextract{\ell_2}\smallsetminus\after(s_2)\).
Hence, \((s_4,s)\in [(\restrict{\concr{G_0}}{\after(s_2)}\cap\Tr)\cup\restrict{{\concr{A}_2}}{\compl{\after(s_0)}}\cup\restrict{\concr{G_0}}{\cextract{\ell_2}\smallsetminus\after(s_2)}]^{\ast}\).
By Lemma \ref{lemma:F+}, \(\after(s_2)\subset\after(s)\), therefore \((s_4,s)\in [(\restrict{\concr{G_0}}{\after(s_2)}\cap\Tr)\cup\restrict{({\concr{A}_2}\cup\restrict{\concr{G_0}}{\cextract{\ell_2}})}{\compl{\after(s_0)}}]^{\ast}\).
By Proposition \ref{proposition:guarantee}, \((s_4,s)\in [(\restrict{\concr{G_\infty}}{\after(s_2)}\cap\Tr)\cup\restrict{({\concr{A}_2}\cup\restrict{\concr{G_0}}{\cextract{\ell_2}})}{\compl{\after(s_0)}}]^{\ast}\).

Let \(\cstate{}=s\) and \(s_{5}=(\thread(s_2),P,\m,\h)\).
Therefore, \((s_2,s_5)\in\col_3\).

If \(i=\thread(s_2)\), then \(s_5=s\) and  \((s,s')\in\Se_3\).
If \(i\neq\thread(s_2)\), then \((s_5,s)\in\Sche\) and \((s,s')\in\Su_3\).
\end{proof}
 
\subsection{Overapproximation of the Execution of a Program}

\begin{lemma}\label{lemma:init}
 For all \(P\) and \(\m\), \(\after((\main,P,\m,\epsilon))=\St\).
 
 In particular, if \(\ini\) is the set of initial states of a program and \(s\in\ini\), then \(\after(s)=\St\).
\end{lemma}

The following proposition shows the connection between the operational and the \cname semantics.

\n{\labprog}{\lab }
\begin{proposition}[Connection with the operational semantics]\label{prop:coroc}
 Consider a program \(\labprog \cmd,\lend\) and its set of initial states \(\ini\).  
 Let:\vspace{-2mm}%
 \[\qcp'\egdef\osem{\labprog \cmd,\lend}\langle \ini,\concr{G}_{\infty},\Sche\rangle\]\vspace{-6.5mm}%
 \[\text{with } \concr{G}_{\infty} = \fp{\av{\labprog \cmd,\lend}}\confinit \vspace{-2mm}\]%
 \noindent Then:\vspace{-4mm}
 \begin{eqnarray*}
    \concr{S'} &=&\{(\main,P,\m,\h)\in\TRi{\labprog \cmd,\lend}^{\star} \langle\ini\rangle\mid  P(\main)=\lend\} \\ 
  \concr{G'}&=&\concr{G}_{\infty}=\{(s,s')\in\TRi{\labprog \cmd,\lend} \mid   s\in\TRi{\labprog \cmd,\lend}^{\star} \langle\ini\rangle \}\cup\Sche\\
  \concr{A'}&=&
   \{(s,s')\in\TRi{\labprog \cmd,\lend} \mid  
   s\in\TRi{\labprog \cmd,\lend}^{\star} \langle\ini\rangle \wedge \thread(s)\neq\main \}\\
   & &
   \cup\Sche
 \end{eqnarray*}
\end{proposition}
\begin{proof}
 We only have to prove that \(\col=\{s\in\TRi{\labprog \cmd,\lend}^{\star}\langle\ini\rangle\mid \thread(s)=\main\}\).
\end{proof}
\begin{proof}
 Let \(s_1\in \{s\in\TRi{\labprog \cmd,\lend}^{\star}\langle\ini\rangle\mid \thread(s)=\main\}\).

There exists \(s_0\in\concr{S}\) such that \((s_0,s)\in\TRi{\labprog \cmd,\lend}^{\star} \)
By proposition \ref{proposition:guarantee}, \((s_0,s)\in\concr{G}_{\infty}\cap\TRi{\labprog \cmd,\lend}^{\star}\)

By Lemma \ref{lemma:init},
 \((s_0,s) (\restrict{{\concr{G}_{\infty}}}{\after(s_0)}\cap\TRi{\labprog \cmd,\lend}^{\star})\cup\restrict{\Sche}{\compl{\after(s_0)}} \).
Hence \((s_0,s)\).

It is straightforward to check that \(\col\subset\{s\in\TRi{\labprog \cmd,\lend}^{\star}\langle\ini\rangle\mid \thread(s)=\main\}\).
\end{proof}

Recall that \(\TRi{\labprog \cmd,\lend}^{\star} (\ini)\) is the set of states that occur on paths starting from \(\ini\).
\(\concr{S'}\) represents all final states reachable by the whole program from an initial state. \(\concr{G'}\) represents all transitions that may be done during any execution of the program and \(\concr{A}'\) represents transitions of children of \(\main\).

\section{Abstract Semantics}
\label{abstract}

\subsection{Abstraction}
\tableabstraction

Recall from the theory of abstract interpretation \cite{CousotCousot04-WCC} that  
a \emph{Galois connection} \cite{ContinuousLatticesandDomains} between a concrete complete lattice \(X\) and an abstract complete lattice \(Y\) is a pair of monotonic functions \(\alpha : X \rightarrow  Y\) and \(\gamma : Y \rightarrow  X\) 
such that \(\forall x\in X , \forall y\in Y,  \alpha(x)\leqslant y\Leftrightarrow x\leqslant \gamma(y)\); \(\alpha\) is called the \emph{abstraction} function and \(\gamma\) the \emph{concretization} function.  
Product lattices are ordered by the product ordering and sets of functions from \(X\) to a lattice \(L\) are ordered by the pointwise ordering \( f\leqslant g \Leftrightarrow \forall x \in X, f(x)\leqslant g(x)\).
A monotonic function \(f^{\sharp}\) is an \emph{abstraction} of a monotonic function \(f^{\flat}\) if and only if \(\alpha\circ f^{\flat}\circ\gamma  \leqslant f^{\sharp}\). 
It is a classical result \cite{ContinuousLatticesandDomains} that an adjoint uniquely determines the other in a Galois connection; therefore, we sometimes omit the abstraction function (lower adjoint) or the concretization function (upper adjoint).

Our concrete lattices are the powersets \(\D\) and \(\DI\) ordered by inclusion.
Remember, our goal is to adapt any given single-thread analysis in a multithreaded setting. Accordingly, we are given an abstract complete lattice \(\Da\) of abstract states and an abstract complete lattice \(\DIa\) of abstract transitions.
These concrete and abstract lattices are linked by two Galois connections, respectively \(\alpha\Dval,\gamma\Dval\) and \(\alpha\DIval,\gamma\DIval\). 
We assume that abstractions of states and  transitions depend only on stores and that all the transitions that leave the store unchanged are in \(\gamt(\bot)\). This assumption allows us to abstract \(\cguard\) and \(\crcreate\) as the least abstract transition \(\bot\).

We also assume we are given the abstract operators of Table \ref{fig:abstractions}, which are correct abstraction of the corresponding concrete functions.
We assume \(\fin\in\Lab\) a special label which is never used in statements. Furthermore, we define \(\cextract{\fin}\egdef\St\).

We define a Galois connection between 
 \( \wp(\St) \) and \( \wp(\Lab) \):
\(\alphl(\concr{S}) = \{\ell \in\Lab\mid  \concr{S} \cap \cextract{\ell}\neq \emptyset \} \) and 
\(\gaml(\abstr{L}) = \bigcap_{\ell\in\Lab\smallsetminus\abstr{L}} \compl{\cextract{\ell}} \) (by convention, this set is \(\St\) when \(\abstr{L} = \Lab\)). The set \(\alphl(\concr{S})\) represents the set of labels that may have been encountered before reaching this point of the program.

Note that we have two distinct ways of abstracting states \(\cstate{}\), either by using \(\alpha\Dval\), which only depends  on the store \(\m\), or by using \(\alphl\) which only depends  on the genealogy \(\h\) and the current thread \(i\). The latter is specific to the multithreaded case, and is used to infer information about possible interferences.

Just as \(\alpha\Dval\) was not enough to abstract states in the multithreaded setting, \(\alpha\DIval\) is not enough, and lose the information that a given transition is or not in a given \(\cextract{\ell}\). This information is needed because \(\restrict{\concr{G}}{\cextract{\ell}}\) is used in Theorem \ref{theorem:denot} and Fig.~\ref{concrfcts}.
Let us introduce the following Galois connection between the concrete lattice \(\DI\) and the abstract lattice \(\DIa^{\Lab}\), the product of \( |\Lab| \) copies of \(\DI\), to this end:
\(\alphk(\concr{G}) = \func{\ell}{\alpha\DIval(\restrict{\concr{G}}{\cextract{\ell}})}\)\\
\(\gamk(\abstr{K}) = \{(s,s')\in\Tra\mid  \forall \ell\in\Lab, s\in\cextract{\ell} \Rightarrow (s,s') \in\gamma\DIval(\abstr{K}(\ell))  \} \).\\
\(\abstr{K}=\alphk(\concr{G})\) is an abstraction of the ``guarantee condition'': \(\abstr{K}(\fin)\) represents the whole set \(\concr{G}\), and \(\abstr{K}(\ell)\) represents the interferences of a child with its parent, i.e., abstracts  \(\restrict{\concr{G}}{\cextract{\ell}}\).

\emph{Abstract configurations} are tuples \(\qu \in \Da \times \wp(\Lab) \times  \DIa^{\Lab} \times \DIa \) such that \(\acinter{\abstr{I}}\abstr{C}=\abstr{C}\) and \(\fin\in\abstr{L}\). 
The meaning of each component of an abstract configuration is given by the Galois connection \(\alphconf,\gamconf\):
 \begin{align*}
  \alphconf \qc &\egdef \langle\acinter{\alpha\DIval(\concr{A})}(\alpha\Dval (\concr{S})), \alphl(\concr{S}),\alphk(\concr{G}),\alpha\DIval(\concr{A})  \rangle\\
  \gamconf \qc &\egdef \langle\gamma\Dval (\abstr{C})\cap\gaml(\abstr{L}),
\gamk(\abstr{K}),\gamma\DIval(\abstr{I})  \rangle
 \end{align*}
\(\abstr{C}\) abstracts the possible current stores \(\concr{S}\). \(\abstr{L}\) abstracts the labels encountered so far in the execution. \(\abstr{I}\) is an abstraction of interferences \(\concr{A}\).

\subsection{Applications: Non-Relational Stores and Gen/Kill Analyses}\label{sec:nonrelational}

As an application, we show
some concrete and abstract stores that can be used in practice. We define a Galois connection \(\alpha\galstore,\gamma\galstore\) between concrete and abstract stores and encode both \emph{abstract states} and \emph{abstract transitions} as abstract stores, i.e., \(\Da=\DIa\). Abstract states are concretized by: \[\gamma\Dval(\am)\egdef\{(i,P,\m,\h) \mid  \m\in\gamma\galstore(\am) \}.\]

\paragraph{Non-relational store} Such a store is a map from the set of variables \(\va\) to some set \(\mathcal{V}^{\flat}\) of \emph{concrete values}, and abstract stores are maps from \(\va\) to some complete lattice \(\mathcal{V}^{\sharp}\) of \emph{abstract values}. Given a Galois connection \(\alpha\Vval,\gamma\Vval\) between \(\mathcal{V}^{\flat}\) and \(\mathcal{V}^{\sharp}\), the following is a classical, so called non-relational abstraction of stores:
\[\alpha\galstore (\m) \egdef \func{x}{\alpha\Dval(\m(x))}\text{ and }\gamma\galstore (\am) \egdef \{\m \mid  \forall x, \m(x) \in \gamma\Dval (x)\}.\]

Let \(\valeur{\abstr{C}}{e}\) and \(\addr{\abstr{C}}{lv}\) be the abstract value of the expression \(e\) and the set of variables that may be represented by \(lv\), respectively, in the context \(\abstr{C}\). 
 \begin{eqnarray*}
 \gamma\DIval(\am) &\egdef& 
\big\{((i,P,\m,h),(i',P',\m',h'))\mid \forall x, \m'(x) \in \gamma\Dval(\am (x))\cup\{\sigma(x)\}\big\}
\\
\aecra{x:=e}{\abstr{C}}&\egdef&\fx{\abstr{C}}{x}{\valeur{\abstr{C}}{e}}\\
\aecra{lv:=e}{\abstr{C}}&\egdef& \bigcup_{x\in \addr{\abstr{C}}{lv} } \aecra{x:=e}{\abstr{C}}\\
  \iaecra{lv:=e}{\abstr{C}} &\egdef&\func{x}{ \text{if } x\in\addr{\abstr{C}}{lv} \text{ then } \valeur{\abstr{C}}{e}\text{ else }\bot }\\
  \acinter{\abstr{I}}(\abstr{C})&\egdef&\abstr{I}\sqcup\abstr{C}\\
  \force_{x}(\m)&\egdef&\fx{\m}{x}{\vrai^{\sharp}}\text{ and } \force_{\neg x}(\m)=\fx{\m}{x}{\faux^{\sharp}}
 \end{eqnarray*}

\paragraph{Gen/kill analyses} In such analyses \cite{DBLP:conf/concur/LammichM07}, stores are sets, e.g., sets of initialized variables, sets of edges of a point-to graph. The set of stores is \(\mathcal{P}(X)\) for some set \(X\), \(\Da=\DIa=\mathcal{P}(X)\), and the abstraction is trivial \(\alpha\galstore=\gamma\galstore=\mathit{id}\). Each gen/kill analysis gives, for each assignment, two sets: \(\gen(lv:=e,\m)\) and \(\ki(lv:=e,\m)\). These sets may take the current store \(\m\) into account (e.g. Rugina and Rinard's ``strong flag'' \citerinard); \(\gen\) (resp.\ \(\ki\)) is monotonic (resp.\ decreasing) in \(\m\). We define the concretization of transitions and the abstract operators:\\
\begin{eqnarray*}
 \gamma\DIval(\am) &\egdef& 
\big\{(i,P,\m,h)\rightarrow(i',P',\m',h')\mid  \m' \subset \m\cup\am  \big\}\\
\aecra{lv:=e}{\abstr{C}}&\egdef&(\abstr{C}\smallsetminus \ki(lv:=e,\m)) \cup \gen(lv:=e,\m)  \\
\iaecra{lv:=e}{\abstr{C}}&\egdef& \gen(lv:=e,\m) \\
\acinter{\abstr{I}}(\abstr{C})&\egdef&\abstr{I}\cup\abstr{C}\\
\force_{x}(\m)&\egdef&\m
\end{eqnarray*}

\subsection{Semantics of Commands}

\figbasicabstractsem

\begin{lemma}\label{lemma:L}
 \(\alphl(\concr{S})=\alphl(\cinter{\concr{A}}(\concr{S}))\).
\end{lemma}

\begin{lemma}
 \label{lemma:LL}
  \(\alphl(\cschedule(\concr{S}))=\func{\ell}{\bot}\).
\end{lemma}

\begin{lemma}\label{lemma:Kunion}
 Let \(\concr{G_1}\) and \(\concr{G_2}\) two set of transitions and \(\concr{S}_2=\{s\mid \exists s': (s,s')\in\concr{G_2}\}\).\\
 Hence, \(\alphk(\concr{G_1}\cup\concr{G_2})\leqslant\func{\ell}{\text{if } \ell\in\alphl(\concr{S}_2) 
 \text{ then } \abstr{K}(\ell)\sqcup \iaecra{lv:=e}{\abstr{C}} \text{ else } \abstr{K}(\ell)    }\)
\end{lemma}

The functions of Fig.~\ref{fig:basicabstract} abstract the corresponding functions of the \cname semantics (See Fig.~\ref{concrfcts}). 
\begin{proposition}\label{prop:abstract}
The abstract functions \(\assign{lv:=e}\), 
\(\guard{\guardp}\), \(\aspawn{\ell_2}\), \(\afinit{\ell_2}\), \(\glue{}\) and \(\afp{\lab \cmd}\) are abstractions of the concrete functions \(\osem{\lab lv:=e,\ell'}\), 
\(\osem{\lab\cguard(\guardp),\ell'}\), \(\osem{\li1\crcreate(\ell_2),\ell_3}\), \(\cinit_{\ell_1}\circ \osem{\li1\crcreate(\ell_2),\ell_3}\), \(\ccomb\) and \(\fpd{\osem{\lab \cmd,\lend}}\) respectively.
\end{proposition}
\begin{proof}
The cases of \(\glue{}\) and \(\afp{\lab \cmd}\) are straightforward.
The case of \(\afinit{\ell_2}\) is a straightforward consequence of Lemma \ref{lemma:LL}.

 Let \(\qu\) an abstract configuration and \(\qc=\gamconf\qu\). Therefore \(\concr{S}=\cinter{\concr{A}}(\concr{S})\).
 
 Let \(\qcp'= \osem{\lab lv:=e,\ell'}\) and \(\qup'=\assign{lv:=e}\qu \).
 
 Therefore, by definition, \(\acinter{\abstr{I}}\circ \aecr{\lab lv:=e,\ell'} \circ\acinter{\abstr{I}}\).
 By Proposition \ref{prop:basic}, 
 \(\concr{S}' = \cinter{\concr{A}} \big(\TR{\lab lv:=e,\ell'}\smallsetminus \Sche \langle  \cinter{\concr{A}}(\concr{S})\rangle\big)\).
  Hence \(\alphs(\concr{S'})\leqslant \abstr{C'}\).
 
 According to Proposition \ref{prop:basic}, \(\concr{G'}\subset\concr{G}\cup\Gnew\) with \(\Gnew = \{(s,s')\in\TR{\li1 basic,\ell_2}\mid s\in \cinter{\concr{A}}(\concr{S})\} = \{(s,s')\in\TR{\li1 basic,\ell_2}\mid s\in \concr{S}\}\).
 Hence \(\alphr(\Gnew)\leqslant\iaecr{ \lab lv:=e,\ell' }(\abstr{C})\).
 
 Therefore by Lemma \ref{lemma:Kunion}:\\ \(\alphk(\concr{G}')\leqslant\func{\ell}{\text{if } \ell\in\abstr{L} 
 \text{ then } \abstr{K}(\ell)\sqcup \iaecra{lv:=e}{\abstr{C}} \text{ else } \abstr{K}(\ell)    }\)
 
 If \((s,s')\in\TR{\lab lv:=e,\ell'}\) then,
  \(s'\in\cextract{\ell}\Leftrightarrow s\in\cextract{\ell}\).
  Therefore, by Lemma \ref{lemma:L}, \(\alphl(S)=\alphl(S')\).
 
 Hence \(\alphconf(\qcp')\leqslant\qup'\).
Given that \(\alphr(\TR{\lab\cguard(\guardp),\ell'})=\bot\) and   \(\forall(s,s')\in\TR{\lab\cguard(\guardp),\ell'},s'\in\cextract{\ell}\Leftrightarrow s\in\cextract{\ell}\), we prove in the same way that \(\guard{\guardp}\) is an abstraction of \(\osem{\lab\cguard(\guardp),\ell'}\).

Given that \(\alphr(\TR{\li1\crcreate(\ell_2),\ell_3})=\bot\) and    \(\forall(s,s')\in\TR{\li1\crcreate(\ell_2),\ell_3},s'\in\cextract{\ell}\Leftrightarrow s\in\cextract{\ell}\vee \ell=\ell_2\), 
we prove in the same way that\(\aspawn{\ell_2}\) is an abstraction of\(\osem{\li1\crcreate(\ell_2),\ell_3}\).
\end{proof}

The \(\assign{lv:=e}\) function updates \(\abstr{K}\) by adding the modification of the store to all labels encountered so far (those which are in \(\abstr{L}\)).
It does not change \(\abstr{L}\) because no thread is created.
Notice that in the case of a non-relational store, we can simplify function \(\assign{}\) using the fact that \( \acinter{\abstr{I}}\circ\aecra{x:=e}{\abstr{C}} = \fx{\abstr{C}}{x}{\mathit{val}_{\abstr{C}}(e)\sqcup \abstr{I}(x)} \). 

The abstract semantics is defined by induction on syntax, see Fig.~\ref{fig:abstract}, and, with Prop.\ref{prop:abstract}, it is straightforward to check the soundness of this semantics:
\begin{theorem}[Soundness]\label{sounda}
 \(\asem{\cmd,\ell}\) is an abstraction of \( \osem{\cmd,\ell} \).
\end{theorem}
\figabstractsemantics

\subsection{Example}

Consider Fig.~\ref{fig:example} and the non-relational store of ranges \cite{CousotCousot04-WCC}.
We will apply our algorithm on this example.

Our algorithm computes a first time \(\aexe\), then, the fixpoint is not reached, and then, \(\aexe\) is computed another time.

\begin{enumerate}
 \item Initial configuration : \(\abstr{Q}_0=\langle\abstr{C}_0,\{\fin\},\abstr{K}_0,\bot\rangle \) where \(\abstr{C}_0=[y=?, z=?]\) and \(\abstr{L}_0=\{\fin\}\) and \(\abstr{K}_0=\func{\ell}{\bot}\) and \(\abstr{I}_0=\bot\).
 \item The configuration \(\abstr{Q}_1=\osem{\li1 y:=0;\li2 z:=0,\ell_3}(\abstr{Q}_0)\) is computed. \(\abstr{Q}_1=\langle \abstr{C}_1,\{\fin\},\abstr{K}_1,\bot \rangle \) where \(\abstr{C}_1=[y=0, z=0]\) and \(\abstr{K}_1=\fin\mapsto [y=0,z=0]\).
 The \(\abstr{L}\) and \(\abstr{I}\) componnents are not changed because no new thread is created.
 \item\label{item:child-spawn}  The configuration \(\abstr{Q}_2=\afinit{\ell_3}(\abstr{Q}_1)\) is computed. \(\abstr{Q}_2=\langle \abstr{C}_2, \{\fin\},\abstr{K}_2,\bot \rangle \) where \(\abstr{C}_2=\abstr{C}_1\) and \(\abstr{K}_2=\func{\ell}{\bot}\). Notice that because \(\abstr{K}_1(\ell_3)=\bot\) the equality \(\abstr{C}_2=\abstr{C}_1\) holds.
 \item  The configuration  \(\abstr{Q}_3=\osem{\li4y:=y+z,\lend}(\abstr{Q}_2)\) is computed.   \(\abstr{Q}_3=\langle\abstr{C}_3,\{\fin\},\abstr{K}_3,\bot \rangle\) where \(\abstr{C}_3=[y=0, z=0]\) and \(\abstr{K}_3=\fin\mapsto [y=0] \).
 \item  The configuration  \(\abstr{Q}_4=\glue{\aspawn{\ell_3}(\abstr{Q}_2)}(\abstr{Q}_3) \) is computed.
 \(\abstr{Q}_4=\langle \abstr{C}_4, \{\fin,\ell_3\},\abstr{K}_4,\abstr{I}_4 \rangle\).
  \(\abstr{C}_4=[y=0, z=0]\) and \(\abstr{K}_4=[\fin\mapsto [y=0,z=0]]\) and \(\abstr{I}_4=[y=0]\).
  \item  The configuration  \(\abstr{Q}_5=\osem{\li5 z:=3,\lend} \abstr{Q}_4\) is computed.
   \(\abstr{Q}_5=\langle \abstr{C}_5, \{\fin,\ell_3\},\abstr{K}_5,\abstr{I}_5 \rangle\).
  \(\abstr{C}_5=[y=0, z=3]\) and \(\abstr{K}_5=[\fin\mapsto [y=0,z=[0,3]]\) and \(\abstr{I}_5=\abstr{I}_4\).
\end{enumerate}

Then, we compute a second time \(\aexe\), on a new initial configuration \(\langle \abstr{C}_0,\abstr{L}_0,\abstr{K}_5,\abstr{I}_0 \rangle \).

\myexample
 Noting change, except at the step \ref{item:child-spawn}, when \(\afinit{}\) is applied.
 The configuration obtained is then \(\abstr{Q}'_2=\langle \abstr{C}'_2, \{\fin\},\abstr{K}_5, \abstr{I}'_2 \rangle \) where \(\abstr{C}'_2= [y=0,z=[0,3]] \) and \(\abstr{I}'_2 = [z=3] \).
 Then, the algorithm discovers that the value of \(y\) may be 3.
 
The details of the execution of the algorithm is given in the following tabular:

\n{\bigK}[1]{\begin{array}{lcc}\fin&\mapsto& y=#1, z=[0,3]\\\ell_3&\mapsto& z=3\end{array}}
\n{\ExStore}[2]{\begin{array}{lcc} y &=&#1\\z&=&#2  \end{array}}

\noindent\(
\begin{array}{l|c|c|c|c|}
 & \abstr{C} & \abstr{L} & \abstr{K} & \abstr{I}\\ 
 \hline
 \text{Initial configuration} & \ExStore{?}{?} & \{\fin\} & \func{\ell}{\bot} & \bot \\
 \hline
 \osem{\li1 y:=0,\ell_2} & \ExStore{0}{?} & \{\fin\} & \fin\mapsto y=0 &\bot\\
 \hline
 \osem{\li2 z:=0,\ell_3} & \ExStore{0}{0} & \{\fin\} &\fin\mapsto y=0, z=0 &\bot\\
 \hline
 \afinit{\ell_3}&\ExStore{0}{0} & \{\fin\} &\func{\ell}{\bot} &\bot \\
 \hline
 \osem{\li4y:=y+z,\lend} & \ExStore{0}{0} &  \{\fin\} &\fin\mapsto y=0 &\bot\\
\hline
 \glue{\aspawn{\ell_3}(\cdot)} & \ExStore{0}{0} & \{\fin,\ell_3\} & \fin\mapsto y=0, z=0 & y=0\\
\hline
 \osem{\li5 z:=3,\lend} & \ExStore{0}{3} &\{\fin,\ell_3\} & \bigK{0} &y=0 \\
\hline
\hline
 \text{Initial configuration} & \ExStore{?}{?} & \{\fin\} & \bigK{0} & \bot \\
 \hline
 \osem{\li1 y:=0,\ell_2} & \ExStore{0}{?} & \{\fin\} & \bigK{0} &\bot\\
 \hline
 \osem{\li2 z:=0,\ell_3} & \ExStore{0}{0} & \{\fin\} &\bigK{0} &\bot\\
 \hline
 \afinit{\ell_3}&\ExStore{0}{[0,3]} & \{\fin\} &\func{\ell}{\bot} &z=3 \\
 \hline
 \osem{\li4y:=y+z,\lend} & \ExStore{[0,3]}{[0,3]} &  \{\fin\} &\fin\mapsto y=[0,3] &z=3\\
\hline
 \glue{\aspawn{\ell_3}(\cdot)} & \ExStore{[0,3]}{0} & \{\fin,\ell_3\} & \bigK{[0,3]} & y=[0,3]\\
\hline
 \osem{\li5 z:=3,\lend} & \ExStore{[0,3]}{3}& \{\fin,\ell_3\} & \bigK{[0,3]} &y=[0, 3] \\
\hline
\end{array}\)

\section{Practical Results}\label{algosm}

The abstract semantics is denotational, so we may compute it recursively. 
 This requires to compute fixpoints and may fail to terminate.
For this reason, each time we have to compute \(f\eomega(X)\) we compute instead the overapproximation \(f\ewiden\), where \(\widen\) is a widening operator, in the following way:
 \begin{inparaenum}
 \item Assign \(X_1:=X\)
 \item\label{algo:whilebody} Compute \(X_2:=f(X_1)\)
 \item If \(X_2\leqslant X_1\) then returns \(X_2\), otherwise,
 \item Assign \(X_1:=X_1\widen X_2\) and go back to \ref{algo:whilebody}.
\end{inparaenum} 
Our final algorithm is to compute recursively \(\afp{\lab \cmd,\lend}\) applied to the initial configuration \(\langle \top,\{\fin\},\func{\ell}{\bot},\bot \rangle \), overapproximating all fixpoint computations.

\mytable

We have implemented two tools, \progint{} and \progP{}, in Ocaml with the front-end C2newspeak, 
with two different abstract stores. 
The first one maps variables to integer intervals and computes an overapproximation of the values of the variables. The second one extends the analysis of Allamigeon et al.\@ \cite{AllamigeonGodardHymansSAS06}, which focuses on pointers, integers, C-style strings and structs and detects array overflows.
It analyzes programs in full fledged C (except for dynamic memory allocation library routines) that use the Pthreads multithread library.
We ignore mutexes and condition variables in these implementations. This is sound because mutexes and condition variables only restrict possible transitions. We lose precision if mutexes are used to create atomic blocks, but not if they are used only to prevent data-races.

In Table \ref{table:benchmarks} we show some results on benchmarks of differents sizes. \loc{} means ``Lines of Code''.
``\messag'' is a C file, with 3 threads: one thread sends an integer message to another through a shared variable. ``\embarque'' is extracted from embedded C code with two threads. ``\douze'' and ``\quinze'' are sets of 12 and 15 files respectively, each one focusing on a specific thread interaction.

To give an idea of the precision of the analysis, we indicate how many false alarms were raised. Our preliminary experiments show that our algorithm loses precision in two ways:
\begin{inparaenum}
 \item %
 through the (single-thread) abstraction on stores
 \item %
 by abstraction on interferences.
\end{inparaenum} 
 Indeed, even though our algorithm takes the order of transitions into account for the current thread, it considers that interference transitions may be executed in an arbitrary order and arbitrary many times. This does not cause any loss in ``\messag'', since the thread which send the message never put an incorrect value in the shared variable. Despite the fact that ``\embarque'' is a large excerpt of an actual industrial code, the loss of precision is moderate: 7 false alarms are reported on a total of 27 100 lines.
Furthermore, because of this arbitrary order, our analysis straightforwardly extends to models with "relaxed-consistency" and "temporary" view of thread memory due to the use of cache, e.g., OpenMP.

\section{Complexity}\label{section:complexity}
The complexity of our algorithm greatly depends on widening and narrowing operators. Given a program \(\li0prog,\lend\), the \emph{slowness} of the widening and narrowing in an integer \(w\) such that:
widening-narrowing stops in always at most \(w\) steps on each loop and whenever \(\afp{}\) is computed (which also requires doing an abstract fixpoint computation).
 Let the \emph{nesting depth} of a program be the nesting depth of \(\cwhile\) and of \(\ccreate\) which\footnote{In our Semantics, each \(\ccreate\) needs a fixpoint computation, except \(\ccreate\) with no subcommand \(\ccreate\).} have a subcommand \(\ccreate\).
 
\begin{proposition}\label{prop:complexity}
Let \(d\) be the nesting depth, \(n\) the number of commands of our program, and, \(w\) the slowless of our widening. The time complexity of our analysis is \(O(n w^{d+1})\) assuming operations on  abstract stores are done in constant time.
\end{proposition}

This is comparable to the \(O(n w^{d})\) complexity of the corresponding single-thread analysis, and certainly much better that the combinatorial explosion of interleaving-based analyses. Furthermore, this is beter than polynomial in an exponential number of states \cite{DBLP:conf/spin/FlanaganQ03}.

\begin{proof}

Let \(c(\lab \cmd,\ell')\), \(n(\lab \cmd,\ell')\) and \(d(\lab \cmd,\ell')\) and \(w(\lab \cmd,\ell')\) be the complexity of analyzing \(\lab \cmd,\ell'\), the size of \(\lab \cmd,\ell'\) and the nesting depth of \(\lab \cmd,\ell'\), the slowless of the widening and narrowing on \(\lab \cmd,\ell'\) respectively. 
Let \(a\) and \(k\) the complexity of assign and of reading \(\abstr{K}(\ell)\) respectively.

Proposition \ref{prop:complexity} is a straightforward consequence of the following lemma\footnote{The functions arguments are omitted in the name of simplicity.}:
\begin{lemma}
 The complexity of computing \(\asem{\lab \cmd,\ell'}\abstr{Q}\) is \(O(an(w+k)w^{d-1})\)
\end{lemma}

This lemma is proven by induction.\\
\(c(lv:=e) = a \)\\
\(c(\li1\cmd_1;\li2\cmd_2,\ell_3)=c(\li1\cmd_1,\ell_2)+c(\li2\cmd_2,\ell_3)\)\\
\(c(\li1\cwhile(\cond)\{\li2\cmd\},\ell_3) \leqslant w(\li1\cwhile(\cond)\{\li2\cmd\},\ell_3) \times c(\li2\cmd,\ell_1) \)\\

If \(\li2 \cmd\) does not contain any subcommand \(\ccreate\), then the fixpoint computation terminates in one step: 
\(c(\li1\ccreate(\li2 \cmd),\ell_3) = k + c(\li2 \cmd)\)\\
Else: 
\(c(\li1\ccreate(\li2 \cmd),\ell_3) = k + w(\li1\ccreate(\li2 \cmd),\ell_3)) \times c(\li2 \cmd)\)
\qed
\end{proof}

 \subsection{Complexity of Operations on \(\DIa^{\Lab}\)}
Notice that we have assumed that operation on \(\DIa^{\Lab}\) are done in constant time in Proposition \ref{prop:complexity}.
 This abstract store may be represented in different ways. The main problem is the complexity of the \(\assign{}\) function, which computes a union for each element in \(\abstr{L}\). The naive approach is to represent \(\abstr{K}\in\DIa^{\Lab}\) as a map from \(\wp(\Lab)\) to \(\DIa\). Assuming that operations on maps are done in constant time, this approach yields a \(O(t n  w^{d})\) complexity where \(t\) is the number\footnote{This is different to the number of threads since an arbitrary number of threads may be created at the same location.} of \(\ccreate\)s in the program.
 We may also represent \(\abstr{K}\in\DIa^{\Lab}\) as some map \(\abstr{K}_{M}\) from \(\wp(\Lab)\) to \(\DIa\) such that \(\abstr{K}(\ell)=\bigcup_{\abstr{L}\ni\ell}\abstr{K}_{M}(\abstr{L})\) and the function \(\assign{}\) is done in constant time : \( \assign{lv:=e} \qu\egdef\langle \acinter{\abstr{I}}\circ\aecra{lv:=e}{\abstr{C}}, \abstr{L},\fx{\abstr{K}_{M}}{\abstr{L}}{\abstr{K}_{M}(\abstr{L})\sqcup\iaecra{lv:=e}{\abstr{C}}} ,\abstr{I}\rangle\). Nevertheless, to access to the value \(\abstr{K}(\ell)\) may need up to \(t\) operations, which increases the complexity of \(\afinit{}\) and \(\glue{}\). The complexity is then \(O(n (w+t) w^{d-1})\).

\subsection{Compexity of Widdenning} 
The slowness of the widening and narrowing operators, \(w\), depends on the abstraction. Nevertheless, a widening is supposed to be fast. 

Consider the naive widening on intervals : \([x,x']\widen[y,y'] = [z',z']\) where \(z=
\begin{cases}
 x &\text{if }y\geqslant x\\
 - \infty &\text{else }
\end{cases}
\)
and  \(z'=
\begin{cases}
 x' &\text{if }y\leqslant x\\
 + \infty &\text{else }
\end{cases}
\).\\
This widening never widen more than two times on the same variable. Therefore this naive widening is linear in the worst case.

\subsection{Other form of parallelism}\label{improvement}

Our technique also applies to other forms of concurrency, Fig.~\ref{mfeatures} displays how Rugina and Rinard's \(\cpar\) constructor \citerinard{} would be computed with our abstraction. Correctness is a straightforward extension of the techniques described in this paper.

Our model handle programs that use \(\ccreate\) and \(\cpar\). Then, it can handle OpenMP programs with ``parallel'' and ``task'' constructors.

\extendedsyntax

\section{Conclusion}\label{section:conclusion}

We have described a generic static analysis technique for multithreaded programs parametrized by a single-thread analysis framework and based on a form of rely-guarantee reasoning. To our knowledge, this is the first such \emph{modular} framework: all previous analysis frameworks concentrated on a particular abstract domain. Such modularity allows us to leverage any static analysis technique to the multithreaded case. We have illustrated this by applying it to two abstract domains: an interval based one, and a richer one that also analyzes array overflows, strings, pointers \cite{AllamigeonGodardHymansSAS06}. Both have been implemented.

We have shown that our framework only incurred a moderate (low-degree polynomial) amount of added complexity. In particular, we avoid the combinatorial explosion of all interleaving based approaches.

Our analyses are always correct, and produce reasonably precise information on the programs we tested. Clearly, for some programs, taking locks/mutexes and conditions into account will improve precision. We believe that is an orthogonal concern: the non-trivial part of our technique is already present \emph{without} synchronization primitives, as should be manifest from the correctness proof of our \cname semantics. We leave the integration of synchronisation primitives with our technique as future work. However, locks whose sole purpose are to prevent data races (e.g. ensuring that two concurrent accesses to the same variable are done in some arbitrary sequential order) have no influence on precision. Taking locks into account may be interesting to isolate atomic blocks.

\section{Acknowledgment}
We thank Jean Goubault-Larrecq for helpful comments.


\begin{thebibliography}{10}
\expandafter\ifx\csname url\endcsname\relax
  \def\url#1{\texttt{#1}}\fi
\expandafter\ifx\csname urlprefix\endcsname\relax\def\urlprefix{URL }\fi
\expandafter\ifx\csname href\endcsname\relax
  \def\href#1#2{#2} \def\path#1{#1}\fi

\bibitem{mine:LCTES06}
A.~Min\'e, Field-sensitive value analysis of embedded {C} programs with union
  types and pointer arithmetics, in: ACM SIGPLAN LCTES'06, ACM Press, 2006, pp.
  54--63, \url{http://www.di.ens.fr/~mine/publi/article-mine-lctes06.pdf}.

\bibitem{AllamigeonGodardHymansSAS06}
X.~Allamigeon, W.~Godard, C.~Hymans, Static {A}nalysis of {S}tring
  {M}anipulations in {C}ritical {E}mbedded {C} {P}rograms, in: K.~Yi (Ed.),
  Static Analysis, 13th International Symposium (SAS'06), Vol. 4134 of Lecture
  Notes in Computer Science, Springer Verlag, Seoul, Korea, 2006, pp. 35--51.

\bibitem{steensgaard96pointsto}
B.~Steensgaard, Points-to analysis in almost linear time, in: POPL '96:
  Proceedings of the 23rd ACM SIGPLAN-SIGACT symposium on Principles of
  programming languages, ACM Press, New York, NY, USA, 1996, pp. 32--41.
\newblock \href {http://dx.doi.org/http://doi.acm.org/10.1145/237721.237727}
  {\path{doi:http://doi.acm.org/10.1145/237721.237727}}.

\bibitem{CousotCousot04-WCC}
P.~Cousot, R.~Cousot, Basic Concepts of Abstract Interpretation, Kluwer
  Academic Publishers.

\bibitem{mine:padoII}
A.~Min\'e, A new numerical abstract domain based on difference-bound matrices,
  in: PADO II, Vol. 2053 of LNCS, Springer-Verlag, 2001, pp. 155--172,
  \url{http://www.di.ens.fr/~mine/publi/article-mine-padoII.pdf}.

\bibitem{DBLP:conf/concur/LammichM07}
P.~Lammich, M.~M{\"u}ller-Olm, Precise fixpoint-based analysis of programs with
  thread-creation and procedures, in: L.~Caires, V.~T. Vasconcelos (Eds.),
  CONCUR, Vol. 4703 of Lecture Notes in Computer Science, Springer, 2007, pp.
  287--302.

\bibitem{locksmith}
P.~Pratikakis, J.~S. Foster, M.~Hicks, Locksmith: context-sensitive correlation
  analysis for race detection, in: PLDI '06: Proceedings of the 2006 ACM
  SIGPLAN conference on Programming language design and implementation, ACM
  Press, New York, NY, USA, 2006, pp. 320--331.
\newblock \href {http://dx.doi.org/http://doi.acm.org/10.1145/1133981.1134019}
  {\path{doi:http://doi.acm.org/10.1145/1133981.1134019}}.

\bibitem{701315}
L.~Fajstrup, E.~Goubault, M.~{Rau\ss en}, Detecting deadlocks in concurrent
  systems, in: CONCUR '98: Proceedings of the 9th International Conference on
  Concurrency Theory, Springer-Verlag, London, UK, 1998, pp. 332--347.

\bibitem{posix-but}
D.~R. Butenhof, Programming with POSIX Threads, Addison-Wesley, 2006.

\bibitem{CousotCousot92-1}
P.~Cousot, R.~Cousot, Abstract interpretation and application to logic
  programs, Journal of Logic Programming '92.

\bibitem{andersen94program}
L.~O. Andersen,
  \href{http://repository.readscheme.org/ftp/papers/topps/D-203.ps.gz}{Program
  analysis and specialization for the {C} programming language}, Ph.D. thesis,
  DIKU, University of Copenhagen (May 1994).
\newline\urlprefix\url{http://repository.readscheme.org/ftp/papers/topps/D-203%
.ps.gz}

\bibitem{rugina99pointer}
R.~Rugina, M.~C. Rinard,
  \href{citeseer.ist.psu.edu/rugina99pointer.html}{Pointer analysis for
  multithreaded programs}, in: PLDI, 1999, pp. 77--90.
\newline\urlprefix\url{citeseer.ist.psu.edu/rugina99pointer.html}

\bibitem{rugina-pointer}
R.~Rugina, M.~C. Rinard, Pointer analysis for structured parallel programs, ACM
  Trans. Program. Lang. Syst. 25~(1) (2003) 70--116.
\newblock \href {http://dx.doi.org/http://doi.acm.org/10.1145/596980.596982}
  {\path{doi:http://doi.acm.org/10.1145/596980.596982}}.

\bibitem{996869}
A.~Venet, G.~Brat, Precise and efficient static array bound checking for large
  embedded {C} programs, in: PLDI '04: Proceedings of the ACM SIGPLAN 2004
  conference on Programming language design and implementation, ACM Press, New
  York, NY, USA, 2004, pp. 231--242.
\newblock \href {http://dx.doi.org/http://doi.acm.org/10.1145/996841.996869}
  {\path{doi:http://doi.acm.org/10.1145/996841.996869}}.

\bibitem{DBLP:conf/spin/FlanaganQ03}
C.~Flanagan, S.~Qadeer, Thread-modular model checking, in: T.~Ball, S.~K.
  Rajamani (Eds.), SPIN, Vol. 2648 of Lecture Notes in Computer Science,
  Springer, 2003, pp. 213--224.

\bibitem{DBLP:conf/sas/MalkisPR07}
A.~Malkis, A.~Podelski, A.~Rybalchenko, Precise thread-modular verification,
  in: H.~R. Nielson, G.~Fil{\'e} (Eds.), SAS, Vol. 4634 of Lecture Notes in
  Computer Science, Springer, 2007, pp. 218--232.

\bibitem{360206}
E.~Yahav, Verifying safety properties of concurrent java programs using
  3-valued logic, in: POPL '01: Proceedings of the 28th ACM SIGPLAN-SIGACT
  symposium on Principles of programming languages, ACM Press, New York, NY,
  USA, 2001, pp. 27--40.
\newblock \href {http://dx.doi.org/http://doi.acm.org/10.1145/360204.360206}
  {\path{doi:http://doi.acm.org/10.1145/360204.360206}}.

\bibitem{repsperso}
T.~W. Reps, Personnal communication (2008).

\bibitem{1040299}
C.~Flanagan, S.~N. Freund, M.~Lifshin, Type inference for atomicity, in: TLDI
  '05, ACM Press, 2005, pp. 47--58.

\bibitem{reduction75}
R.~J. Lipton, Reduction: a method of proving properties of parallel programs,
  Commun. ACM 18~(12) (1975) 717--721.

\bibitem{SPLST/Vojdani07}
V.~Vojdani, V.~Vene, Goblint: Path-sensitive data race analysis, in: SPLST,
  2007.

\bibitem{Vene03globalinvariants}
V.~Vene, M.~Muller-olm, Global invariants for analyzing multi-threaded
  applications, in: In Proc. of Estonian Academy of Sciences: Phys., Math,
  2003, pp. 413--436.

\bibitem{ContinuousLatticesandDomains}
G.~Gierz, K.~Hofmann, K.~Keimel, J.~Lawson, M.~Mislove, D.~Scott, Continuous
  Lattices and Domains, Cambridge University Press, 2003.

\end{thebibliography}
\end{document}